\documentclass{elsarticle}

\usepackage[fleqn]{amsmath}
\usepackage{amssymb}
\usepackage{amsthm}
\usepackage[margin=0.6in, paper=a4paper]{geometry}
\usepackage[colorlinks=true, linkcolor=blue, citecolor=blue, urlcolor=cyan]{hyperref}
\AtBeginDocument{\hypersetup{citecolor=blue}}
\hypersetup{pdfauthor={Name}}
\usepackage[nameinlink]{cleveref}
\usepackage{graphicx}
\usepackage{subcaption}
\usepackage{tikz}
\usetikzlibrary{cd}

\counterwithin{equation}{section}
\theoremstyle{definition}
\newtheorem{theorem}{Theorem}[section]
\newtheorem{lemma}[theorem]{Lemma}
\newtheorem{claim}[theorem]{Claim}
\newtheorem{corollary}[theorem]{Corollary}
\newtheorem{definition}[theorem]{Definition}
\newtheorem{example}[theorem]{Example}
\newtheorem{remark}[theorem]{Remark}

\newcommand{\N}{\mathbb{N}} 
\newcommand{\R}{\mathbb{R}} 
\renewcommand{\S}{\mathbb{S}} 

\newcommand{\abs}[1]{\left\lvert#1\right\rvert}
\newcommand{\inner}[2]{\langle#1,#2\rangle}

\newcommand{\linearspan}{\mathop{\rm span}\nolimits}
\newcommand{\Ker}{\mathop{\rm Ker}\nolimits}
\renewcommand{\Im}{\mathop{\rm Im}\nolimits}
\newcommand{\Hom}{\mathop{\rm Hom}\nolimits}

\newcommand{\Aff}{\mathop{\rm Aff}\nolimits} 
\newcommand{\Con}{\mathop{\rm Con}\nolimits} 

\newcommand{\sgn}{\mathop{\rm sgn}\nolimits}   
\newcommand{\OR}{\mathop{\rm or}\nolimits}   
\newcommand{\vol}{\mathop{\rm vol}\nolimits} 

\usepackage{array}
\newcolumntype{C}[1]{>{\centering\let\newline\\\arraybackslash\hspace{0pt}}m{#1}}

\journal{arXiv}

\begin{document}

\begin{frontmatter}

\title{Diffusion in multi-dimensional solids using Forman's combinatorial differential forms}
\author[1]{Kiprian Berbatov\corref{fn1}}
\author[1]{Pieter D. Boom}
\author[2]{Andrew L. Hazel}
\author[1]{Andrey P. Jivkov\corref{fn2}}
\address[1]{Department of Mechanical, Aerospace and Civil Engineering, The University of Manchester, Oxford Road, Manchester M13 9PL, UK}
\address[2]{Department of Mathematics, The University of Manchester, Oxford Road, Manchester M13 9PL, UK}
\tnotetext[t1]{This document is a result of a research sponsored by Engineering and Physical Sciences Research Council via grant EP/N026136/1.}
\cortext[fn1]{Corresponding author: kiprian.berbatov@postgrad.manchester.ac.uk; kberbatov@gmail.com}
\cortext[fn2]{Corresponding author: andrey.jivkov@manchester.ac.uk}

\begin{abstract}
 The formulation of combinatorial differential forms, proposed by Forman for analysis of topological properties of discrete complexes, is extended by defining the operators required for analysis of physical processes dependent on scalar variables. The resulting description is intrinsic, different from the approach known as Discrete Exterior Calculus, because it does not assume the existence of smooth vector fields and forms extrinsic to the discrete complex. In addition, the proposed formulation provides a significant new modelling capability: physical processes may be set to operate differently on cells with different dimensions within a complex. An application of the new method to the heat/diffusion equation is presented to demonstrate how it captures the effect of changing properties of microstructural elements on the macroscopic behavior. The proposed method is applicable to a range of physical problems, including heat, mass and charge diffusion, and flow through porous media.
\end{abstract}

\begin{keyword}
 Discrete structures \sep Topology \sep Metric \sep Diffusion \sep Composites
 \MSC[2010] 52B70 \sep 52B99 \sep 65N50 \sep 80A20
\end{keyword}

\end{frontmatter}

\section{Introduction}
Approximating solid materials as continua simplifies the mathematical descriptions of the physical and mechanical processes operating on them, and of the corresponding conservation laws. These descriptions are typically partial differential equations and their solutions, mostly numerical and rarely analytical, have been serving well all branches of engineering for over two centuries. However, real solids have internal structures, where the ``continuum'' bulk is broken down into discrete regions by ``defects''. These ``defects'' may be introduced by design, for example particles, fibres or platelets of materials different from the bulk are dispersed or arranged to form composites \cite{Boisse2021_Composites}, or by the manufacturing process, for example boundaries between crystals, junctions between boundaries, and particulate precipitates form during solidification of polycrystalline metallic alloys \cite{Sankaran2017_Alloys}. In a number of practical problems the properties of the bulk and the different ``defect'' types that control given physical or mechanical process are different. For example, the thermal conductivity, mass diffusivity, or electrical conductivity of the bulk may differ from those of the different defects to such an extent that the corresponding heat conduction, mass diffusion or charge transfer are either strongly localised or strongly inhibited by the defects. The rapid development of the additive manufacturing techniques for both monolithic \cite{DebRoy2018_AMmetals} and composite \cite{Ngo2018_AMcomposites} materials, coupled with increasing demands for multi-functional materials, calls for efficient methods for analysis of materials with such complex internal structures.

One approach to keep the continuum approximation but account for the effects of internal structures on the macroscopic/bulk behaviour is to modify the governing equations by using derivatives of fractional orders. This is the subject of fractional calculus \cite{West2017_Fractional}, a mathematical area with a number of applications to science and engineering problems \cite{Sun2018_Fractional}. Selected examples include formulations of thermo-elasticity \cite{Sherief2010_Fractional}, viscoelasticity \cite{DiPaola2011_Fractional} and mass diffusion \cite{Bologna2015_Fractional} in structured media. In all cases, the fractional order of derivation is selected such that the results fit experimental data. Hence, fractional calculus can be seen as a way to capture the effects of the underlying structure on the process under investigation. Inversely, an order of derivation, required to describe experimental observations but different from a natural number, indicates the existence of an underlying structure. Despite the demonstrated successes of fractional calculus, the approach lacks explanatory power, i.e., it cannot provide explanations for the relations between underlying structures and observed behaviours. Due to this limitation, fractional calculus cannot be used for rational design of internal structures. An approach that has the potential to provide such explanations needs to account explicitly for the finite and discrete nature of the defects shaping the internal structures.

At a certain length scale of observation the defects can be considered as forming 0-dimensional (particles), 1-dimensional (fibres or junctions), and 2-dimensional (platelets and boundaries) sub-structures, which may not be necessarily connected, but partially tessellate the 3-dimensional bulk. A solid with such defects can be modelled by mapping the defects to a polyhedral complex, which in the terminology of algebraic topology is a 3-complex constructed from 0-cells representing vertices, 1-cells representing line segments, 2-cells representing polygonal areas, and 3-cells representing polyhedral volumes. The mapping sends 0-dimensional defects to some 0-cells, 1-dimensional defects to some 1-cells, 2-dimensional defects to some 2-cells, and the bulk material is mapped to all 3-cells and to the unoccupied cells of lower dimensions. The question requiring investigation is how to formulate the processes and fulfill the conservation laws on such cell complexes, considering the possibility to have different physical properties for different cells. 

One approach to the analysis on cell complexes (specifically simplicial complexes), known as Discrete Exterior Calculus (DEC) \cite{Hirani2003_DEC}, has been developed to translate the operations of smooth exterior calculus on such complexes. The basic structures on cell complexes are the chains (a $p$-chain is a linear combinations of $p$-cells) and the cochains (a $p$-cochain is a linear functional on $p$-chains). DEC assumes smooth exterior forms living in a background (Euclidean) space and defines maps between these and the cochains on the cell complex: smooth $p$-forms are mapped to $p$-cochains. Thus, DEC is intended to be a geometric discretisation method, i.e. to discretise the continuum operators using cochains considered as discrete differential forms. From this perspective it is similar to the Finite Element Exterior Calculus \cite{Arnold2018_FEEC}. DEC has been used for example to formulate transport through porous media \cite{Hirani2015_DarcyDEC} and the mathematically similar incompressible fluid flow \cite{Mohamed2016_NavierDEC}. It has been recently shown that DEC solutions of Poisson problems with existing analytical solutions, i.e. homogeneous domains with specific geometries and boundary conditions, converge to the analytical solutions with mesh refinement (increasing the number of 3-cells in the domain) \cite{Schulz2020_DEC}. This illustrates its conceptual similarity to the finite element method. A software implementation of DEC for massively parallel solution of Poisson problems is also available and described \cite{Boom2021_Poisson}. A DEC-based formulation for elasticity has also been recently proposed \cite{Boom2022_Elasticity}.

However, the approach of DEC is not applicable to the problem at hand - analysis of processes operating differently on cells of different dimensions. Moreover, DEC's spirit is very different from what is needed for an intrinsic formulation of physics and mechanics on discrete topological spaces. Closer in spirit is the so called ``cell method'' advocated by Tonti \cite{Tonti2013_Structure}, which seeks an \emph{ab initio} formulation that does not rely on the existence of a smooth background space. An opportunity to develop an intrinsic formulation is presented by the works of Forman \cite{Forman1998_Morse, Forman2002_Novikov}, which introduce the notions of combinatorial vector fields and differential forms. Forman's combinatorial vector fields and the associated discrete Morse theory have been used for example for computing the homology of \cite{Harker2014_Morse} and flows on \cite{Mrozek2021_Flows} cell complexes. The combinatorial differential forms have been used for constructing combinatorial Laplacians and computing combinatorial Ricci curvatures of cell complexes \cite{Forman2003_Bochner, lange2005combinatorial, Watanabe2019_Ricci}, which in turn have been widely used for analysis of Ricci-flows on such complexes \cite{weber2016forman}.

The aim of this article is to extend Forman's approach to allow for formulating a discrete version of the heat/diffusion equation. This extension follows some ideas from the thesis of Arnold on the topological operations on cell complexes 
\cite{arnold2012discrete} and from the work of Wilson on cochain algebra \cite{wilson2007cochain}. The main contribution is the proposed metric operator on cochains that takes two $p$-cochains and returns a 0-cochain, and from there the formulations of inner product of cochains and adjoint operators to exterior differentials required to represent balance equations. The development is new and can be considered as a step towards an intrinsic description of structured solids. The theory will be completed once a formulation of solid deformations is presented, which is a subject of an ongoing work.

\subsection{Continuum problem}
\label{sec:motivation}
The continuum version of the heat/diffusion equation, for which a discrete analogue is being sought, is presented to help comparisons between the continuum and discrete cases. Consider a bounded open domain $M \subset \R^{d}$, with closure $\overline{M}$ and boundary $\partial M := \overline{M} \setminus M$. With the classical vector calculus notations, the heat equation in $M$ expressing the conservation of a scalar quantity in the absence of body sources, is given by 
\begin{equation}
\label{eq:conservation}
  \frac{\partial u}{\partial t} + {\nabla} \cdot {\bf f} = 0,
\end{equation}
where $u$ is a scalar field, e.g., temperature, concentration, or charge, and ${\bf f}$ is a vector field, e.g., heat, mass, or charge flux. The relation between the scalar and the vector quantities at a point in $M$ is provided by a constitutive relation of the form
\begin{equation}
\label{eq:constitutive}
  {\bf f} = -\alpha \nabla u,
\end{equation}
where $\alpha$ is a material property -- thermal, mass, or electrical diffusivity. In the general case of anisotropic materials, $\alpha$ is a tensor field in $M$. For inhomogeneous, but isotropic materials $\alpha$ is a positive scalar field in $M$, and for homogeneous and isotropic materials $\alpha = const > 0$ across $M$. Irrespective of the material constitution and the dimension of the domain, $d$, the physical dimension of the diffusivity coefficient(s) is $[\alpha]=L^2/T$, where $L$ and $T$ are measures of length and time, respectively. If $[u]$ denotes the physical dimension of the scalar quantity, then the constitutive relation given by \Cref{eq:constitutive} gives the physical dimension of the flux as $[f]=[u]L/T$.

An initial boundary value problem is formulated by prescribing initial values of $u$ in $M$ and boundary conditions on $\partial M$. The boundary is represented as a union of non-overlapping sub-domains, $\partial M_D$ and $\partial M_N$, i.e. $\partial M_D \cup \partial M_N = \partial M$ and $\partial M_D \cap \partial M_N = \varnothing$. The initial boundary value problem, including initial, Dirichlet and Neumann boundary conditions is
\begin{eqnarray*} 
\label{eq:heat}
  \frac{\partial u}{\partial t} &=& {\nabla} \cdot {(\alpha \nabla u)} \\
\label{eq:Initial}
  u &=& u_0 \quad \quad \quad t =0 \quad\quad \textrm{in} \quad M,\\
\label{eq:Dirichlet}
  u &=& \bar{u}(t) \quad \quad t \ge 0 \quad \quad \textrm{on} \quad \partial M_D,\\
\label{eq:Neumann}
  \mathbf{f} \cdot \mathbf{n} &=& \bar{f}(t) \quad \quad t \ge 0 \quad \quad \textrm{on} \quad \partial M_N,
\end{eqnarray*}
where $\bar{u}(t)$ is prescribed scalar quantity at a given point on $\partial M_D$, $\mathbf{n}$ is the outward unit normal to a given point on $\partial M_N$, and $\bar{f}(t)$ is prescribed flux normal to a given point on $\partial M_N$.

In the language of exterior differential forms, the scalar quantity $u$ is a $0$-form, and the flux $\mathbf{f}$ is a $1$-form. This perspective is the bridge to the intrinsic discrete version developed in the paper with an appropriate definition of discrete differential forms together with topological and metric operations on such forms substituting the continuum gradient and divergence. The formulation requires background knowledge of discrete complexes or meshes, as well as of topological and metric operations on meshes. So far the required mathematical apparatus is not as widely known and used as the continuum calculus, and every effort is made to explain its components by formulas and examples.

\subsection{Overview of the paper}
\label{sec:overview}
The purpose of this section is to provide a summary of the notions discussed in this paper. The key steps of the development are presented in the main text and should be accessible to readers familiar with the standard notions of polytopes, meshes, orientation, and the boundary operator on meshes. For readers less familiar with these notions, polytopes and meshes are introduced in \ref{sec:polytope_mesh}, and orientations and boundary operators are introduced in \ref{sec:orientation}. \Cref{sec:topology} summarises and generalises notions considered in \cite{Forman2002_Novikov} and \cite{arnold2012discrete}, such as discrete differential forms, Forman subdivision and cubical cup product. Important for this work is the idea of using the Forman subdivision to model physical problems. \Cref{sec:geometry} is the main theoretical contribution, in particular the use of a metric tensor (and node curvature) to define an inner product and related notions. \ref{sec:hodge} summarises standard topological results that use, but are independent of, an inner product. It can be read independently of \Cref{sec:geometry} and the results there do not affect the rest of the paper, but are included for the benefit of interested readers. The application of the developed intrinsic theory is provided in \Cref{sec:applications} by several numerical examples, demonstrating its use for analysis of materials with complex internal structures. More detailed explanations of the notions in various sections of the paper are presented in the following paragraphs.
 
\paragraph{Key notions from \ref{sec:polytope_mesh}}
A \textbf{flat} region is a $d$-manifold embeddable in a $d$-dimensional affine space.
A $d$-\textbf{polytope} is recursively defined as a flat region, homeomorphic to an open set in $\R^d$, whose boundary is a union of finite number of $(d-1)$-polytopes, with $0$-polytopes being points. According to this definition, regions with holes and self-intersecting regions are excluded, but non-convex regions are allowed. A $d$-polytope will be referred to as $d$-cell, with $0,1,2,3$-cells occasionally called \textbf{nodes} (N), \textbf{edges} (E), \textbf{faces} (F), and \textbf{volumes} (V), respectively. A \textbf{hyperface} of a $d$-cell $c_d$ is one of the $(d-1)$-cells, forming the boundary of $c_d$.
A \textbf{$p$-face} of $c_d$, $0 \leq p \leq d$ is a $p$-cell $b_p$ such that there exists a sequence of cells starting from $b_p$ and finishing at $c_d$ with every element in the chain being a hyperface of the next one. 
In this case we say that $c_d$ is a \textbf{superface} of $b_p$.

A $d$-\textbf{mesh} is a collection of cells with maximal topological dimension $d$ such that any two cells can intersect in a finite (possibly empty) union of other cells.
The \textbf{topological dimension} of a mesh is the maximal dimension of its cells.
The \textbf{embedding dimension} of a mesh is the dimension of the affine space the mesh is embedded in. The embedding dimension is not important when looking intrinsically at the mesh. A \textbf{manifold-like} mesh is a $d$-mesh representing a topological manifold. A crucial property is that all $(d - 1)$-cells have at most two $d$-superfaces.

\paragraph{Key notions from \ref{sec:orientation}}
An \textbf{orientation} of a cell is a choice of equivalence class of homotopical bases, i.e., all the bases that can be continuously deformed from one another. There are two such classes. A different definition using top-dimensional elements of the exterior algebra is presented in order to include the $0$-cells. Furthermore, it allows easier algebraic manipulations.

The \textbf{relative orientation} between two cells $c_p$ and $b_{p - 1}$ is a choice of $1$ or $-1$ depending on the chosen orientations of $c_p$ and $b_{p - 1}$ and any outward-pointing vector to $b_{p - 1}$ with respect to $c_p$.
A \textbf{compatible orientation} on a manifold-like mesh $M$ is an orientation on $M$ such that for any pair of $d$-cells with a common hyperface $b_{d-1}$, their relative orientations with respect to $b_{d-1}$ are opposite.
If a compatible orientation on $M$ exits, $M$ is called \textbf{comaptibly orientable}.
If a compatible orientation on $M$ is chosen, $M$ is called \textbf{compatibly oriented}.

A \textbf{chain} is a formal linear combination of the cells of the complex. In general, the coefficients of the linear combination can belong to a ring, but in this work they are considered to be elements of $\R$.
The \textbf{boundary operator} $\partial$ is a linear map between chains which, applied to a basis cell, returns a linear combination of its boundary faces with coefficients $1$ and $-1$ given by the relative orientations. The boundary operator satisfies $\partial \circ \partial = 0$, where $\circ$ denotes composition of functions, which makes the collection of chain spaces a \textbf{chain complex}.
The \textbf{fundamental class} $[M]$ on a compatibly oriented manifold-like mesh $M$ is the sum of all of its $d$-cells.
A \textbf{cochain} is a linear functional on chains.
The \textbf{coboundary operator} $\delta$ is the dual to the boundary map and induces a \textbf{cochain complex}.

\paragraph{Key notions from \Cref{sec:topology}}
A \textbf{discrete differential $p$-form} is a linear map between chains which, when applied to a basis chain (cell) $b_q$, gives a linear combination of the basis $(q - p)$-chains which are faces of $b_q$.
The \textbf{discrete exterior derivative}, $D$, is a linear map taking $p$-forms to $(p + 1)$-forms and satisfying $D \circ D = 0$.

A $d$-polytope has \textbf{cubical corners} if any of its nodes has exactly $d$ $1$-superfaces.
The \textbf{Forman subdivision} of a mesh $M$ is a mesh $K$ (a subdivision of $M$) such that the cochains in $K$ correspond to the discrete differential forms in $M$.
A (compatible) orientation on $M$ induces a (compatible) orientation on $K$.

Two polytopes are called \textbf{combinatorially equivalent} if there is a bijection between their cells that respects the boundary.
A ($d$-dimensional) \textbf{quasi-cube} is a polytope which is combinatorially equivalent to a ($d$-dimensional) cube.
The quasi-cubical \textbf{cup product} is a bilinear map in the cochain complex of a quasi-cubical mesh which satisfies the Leibnitz rule with respect to $D$.
The \textbf{wedge product} on a mesh $M$ with cubical corners is a bilinear product of forms defined by the cup product on the Forman subdivision $K$ of $M$.

The definition of discrete differential forms makes it clear that in the physical problem at hand, the quantity $u$ (a 0-form) is a map $b_p \to \R$ for $0 \le p \le d$, and the quantity $\mathbf{q}$ (a 1-form) is a map $(c_p, b_{p-1}) \to \R$ for $0 < p \le d$, where $b_{p-1}$ is a hyperface of $c_p$. 

\paragraph{Key notions from \Cref{sec:geometry}}
Let $M$ be a mesh of convex cells with cubical corners and $K$ be the Forman subdivision. For a fixed \textbf{inner product} $\inner{\cdot}{\cdot}$ on $C^\bullet K$ the following linear maps are constructed: \textbf{adjoint coboundary operator} $\delta^{\star}$; \textbf{Laplacian} $\Delta$; \textbf{Hodge star} $\star$.
These can be bijectively transformed to maps on $\Omega^\bullet M$.

A \textbf{discrete metric tensor} $g$ is a bilinear map on $K$ taking a pair of $p$-cohains and returning a $0$-cochain.
A family of metric tensors is constructed via a dimensionless function on the nodes accounting for a possible \textbf{curvature} $\kappa$ of the nodes.
The \textbf{volume cochain} $\vol$ (on $K$ - compatibly oriented) is the $d$-cochain with coefficients the measures of $d$-cells.
The \textbf{Riemann integral} of a zero-cochain $\sigma^0$ is the real number $(\sigma^0 \smile \vol)[K]$.
An inner product, used in this article, is given by the Riemann integral of the metric tensor.

\paragraph{Key notions from \ref{sec:hodge}}
For any metric on a mesh $K$ there is a bijection between the \textbf{cohomology} and the kernel of $\Delta$ (the set of \textbf{harmonic} cochains). For a \textbf{closed mesh} (mesh without a boundary) with cup product and Hodge star, \textbf{Poincar\'e duality} gives bijection between homology and cohomology.

\subsection{Notation and conventions}
\label{sec:notation}
Let $A$ be an affine space and $X \subset A$, $x_1, ..., x_n \in A$.
\begin{itemize}
  \item
    $\Aff(X)$ denotes the affine hull of $X$; $\Aff(x_1, ..., x_n)$ := $\Aff(\{x_1, ..., x_n\})$.
  \item
    $\Con(X)$ denotes the convex hull of $X$; $\Con(x_1, ..., x_n)$ := $\Con(\{x_1, ..., x_n\})$.
\end{itemize}
Let $M$ be a mesh with topological dimension $d > 0$. Variables $p, q \in \{0,...,d\}$ are used to denote cell dimensions. Other standard variable names used throughout the paper include:
\begin{itemize}
  \item
    $a_p, b_p, c_p \in M_p$ denote $p$-cells in $M$. By abuse of notation they denote basis $p$-chains;
  \item
    $\pi_p, \rho_p, \sigma_p \in C_p M$ denote $p$-chains in $M$;
  \item
    $a^p, b^p, c^p$ denote basis $p$-cochains in $M$;
  \item
    $\pi^p, \rho^p, \sigma^p \in C^p M$ denote $p$-cochains in $M$;
  \item
    $\omega^p, \eta^p \in \Omega^p M$ denote discrete differential $p$-forms in $M$.
\end{itemize}
For $p \leq q$ we write $a_p \preceq b_q$ (equivalently $a^p \preceq b^q,\ b_q \succeq a_p,\ b^q \succeq a^p$) if $a_p = b_q$ (hence $p = q$) or $a_p$ is a face of $b_q$ (hence $p < q$). In the later case we write $a_p \prec b_q$ (and similarly in the other $3$ cases).

\section{Topological operations on a mesh.}
\label{sec:topology}
This section introduces some topological operations (of algebraic nature) on a mesh $M$, i.e., operations which do not depend on the coordinates/embedding of the mesh but only on the connections between cells and on the orientation. They give a richer mesh calculus which allows for defining metric properties in \Cref{sec:geometry}.
\cite{Forman2002_Novikov} introduced discrete (combinatorial) vector fields and differential forms and exterior derivative of forms. In the beginning of the proof of \cite[Theorem 1.2]{Forman2002_Novikov} he briefly, but not sufficiently rigorously, described a mesh subdivision of $M$. We refer to this as the Forman subdivision of $M$ and denote it by $K$. Notably, Forman did not use this subdivision anywhere else. \cite{arnold2012discrete} noted that for $M$ simplicial, $K$ consists of topological cubes (she called the cells of $K$ \textbf{kites} \cite[Definition 2.1.11]{arnold2012discrete} and $K$ itself the \textbf{associated kite complex} \cite[Definition 5.2.8]{arnold2012discrete}) for which she developed a cup product and a topological theory, taking also inspiration from Wilson's work on simplicial meshes \cite{wilson2007cochain}. See \Cref{thm:whitney_wilson_arnold} for a further review of the cited literature.

\subsection{Discrete differential forms}
\label{sec:forms}
Let $C_\bullet M = \bigoplus_{p = 0}^{d} C_{p} M$ be the space of all chains.
\begin{definition}
\label{thm:discrtete_differential_form}
  A \textbf{discrete differential $p$-form} is a map $\omega \colon C_{\bullet} M \to C_{\bullet} M$ with the properties
  \begin{enumerate}
    \item
      $\omega$ is a $\R$-linear map;
    \item
      for each $q \geq p$, $\omega[C_{q}(M)] \subseteq C_{q - p}(M)$, i.e., $\omega$ is a map of degree $(-p)$ on $C_{\bullet}M$;
    \item
      for each $q > p$ and $c_q \in M_{q}$, $\omega(c_q) \in \linearspan (\{b_{q - p} \in M_{q - p} \mid b_{q - p} \preceq c_q\})$, i.e., the map is local.
  \end{enumerate}
The space of all $p$-forms on $M$ is denoted by $\Omega^{p} M$ and the space of discrete differential forms $\Omega^{\bullet} M$ is the graded vector space $\Omega^{\bullet} M = \bigoplus_{p = 0}^{d} \Omega^{p} M$.
\end{definition}
\noindent For short, ``forms'' is used instead of ``discrete differential forms''.
\begin{remark}
  Let $f \colon M \to \R$, i.e., $f$ takes all cells as arguments. Then $f$ gives rise to a $0$-form $\mathfrak{f} \colon C_\bullet M \to C_\bullet M$ defined on a basis $p$-chain $c_p$ as $ \mathfrak{f}(c_p) = f(c_p) c_p$.
  Inversely, a $0$-form $\mathfrak{f}$ corresponds to a function $f \colon M \to \R$, by taking the coefficient before $\mathfrak{f}(c_p)$ as the value of $f(c_p)$. Indeed, the definition of forms guarantees that the a $0$-form applied to a $p$-cell $c_p$ gives a linear combination of the $p$-faces of $c_p$, i.e., $c_p$ multiplied by a number.
\end{remark}
\begin{example}
  A canonical example of a $0$-form is the identity function $\mathfrak{1}_M$ on $C_\bullet M \to C_\bullet M$ which corresponds to the constant function $\mathfrak{1}_M \colon M \to \R$, $\mathfrak{1}_M(c_p) = 1$.
\end{example}
\begin{example}
  A canonical example of a $1$-form is the boundary map $\partial$. Indeed: it is linear; it maps $p$-chains to $(p - 1)$-chains; when applied to a $p$-cell $c_p$ it gives a linear combination (with coefficients $\pm 1$) of the boundary $(p - 1)$-faces of $c_p$.
\end{example}
\begin{remark}
  The standard basis of $\Omega^{p} M$ consists of the forms of the type $(c_q \to b_{q - p})$, where $b_{q - p} \preceq c_q$, defined for any $a_r \in M_r$ by
  \begin{equation}
    (c_q \to b_{q - p})(a_r) := 
      \begin{cases}
        0   , & a_r \neq c_q \\
        b_{q - p}, & a_r = c_q
      \end{cases}.
  \end{equation}
  The form is extended by linearity for any $\sigma \in C_{\bullet} M.$ In other words, a basis $p$-form is nonzero at exactly one of the cells $c_q$ of $M$ and applied to this cell gives one of its $(q - p)$-faces $b_{q - p}$.
\end{remark}
\begin{definition}
\label{thm:discrete_exterior_derivative}
  Let  $(M, \OR)$ be an oriented mesh,
  $(C_{\bullet}M, \partial)$ and $(C^{\bullet}M, \delta)$
  be the associated chain and cochain complexes.
  The \textbf{discrete exterior derivative} is the linear map $D \colon \Omega^{\bullet} M \to \Omega^{\bullet} M$, defined on $p$-forms by
  \begin{equation}
    D^p \omega^p 
      = \omega^p \circ \partial
        - (-1)^{p} \partial \circ \omega^p \in \Omega^{p + 1} M.
  \end{equation}
\end{definition}
\begin{remark}
  In \cite{Forman2002_Novikov} the map is defined by
  \begin{equation*}
    D^p \omega^p = \partial \circ \omega^p - (-1)^{p} \omega^p \circ \partial,
  \end{equation*}
  i.e., the definition here differs by sign of $(-1)^{p}$. The choice adopted here does not affect the essential properties of $D$ but has the benefit that the orientation of the Forman mesh of $M$, defined by using \Cref{eq:discrete_exterior_derivative}, inherits the orientation of $M$ in a natural geometric way.
\end{remark}
\begin{example}
  For the canonical $1$-form, $\partial$, it is trivial to see that $D \partial = 0$. Indeed,
  $D \partial
    = \partial \circ \partial - (-1)^{1} \partial \circ \partial
    = 0 + 0
    = 0.$
\end{example}
\begin{claim}
  $D^2 = 0$, i.e., $(\Omega^{\bullet} M, D)$ is a cochain complex.
\end{claim}
\begin{proof}
  It is enough to prove the claim for $p$-forms for any $p$. Let $\omega = \omega^p \in \Omega^{p} M$. Then
  \begin{equation*}
    \begin{split}
      D(D \omega)
        & = D(\omega \circ \partial
          - (-1)^{p} \partial \circ \omega) \\
        & = ((\omega \circ \partial) \circ \partial
          - (-1)^{p + 1} \partial \circ (\omega \circ \partial))
          - (-1)^{p}((\partial \circ \omega) \circ \partial
          - (-1)^{p + 1} \partial \circ (\partial \circ \omega)) \\
        & = 0 + (-1)^{p} \partial \circ \omega \circ \partial
          - (-1)^{p} \partial \circ \omega \circ \partial + 0 \\
        & = 0.
    \end{split}
  \end{equation*}
\end{proof}

\subsection{Forms as cochains on the Forman subdivision}
\label{sec:forman}
The expression for $D^p \omega^{p}$ for a basis $p$-form $\omega^p = (c_q \to b_{q - p})$ is given by
\begin{equation}
\label{eq:discrete_exterior_derivative}
  \begin{split}
    D^p \omega^p
      & = D^p (c_q \to b_{q - p}) \\
      & = (c_q \to b_{q - p}) \circ \partial - 
        (-1)^p \partial \circ (c_q \to b_{q - p}) \\
      & = \sum_{a \in M} (c_q \to b_{q - p})(\partial a) 
        - \sum_{a \in M} (-1)^p \partial ((c_q \to b_{q - p})(a)) \\
      & = \sum_{a_{q + 1} \succ c_q} 
          \varepsilon(a_{q + 1}, c_q)\ (a_{q + 1} \to b_{q - p})
        - (-1)^p \sum_{a_{q - p - 1} \preceq b_{q - p}} 
          \varepsilon(b_{q - p}, a_{q - p - 1})\ (c_q \to a_{q - p - 1})
  \end{split}.
\end{equation}
For mesh $M$ with convex cells, \Cref{eq:discrete_exterior_derivative} has a nice interpretation on a subdivision $K$ of $M$, which is referred to as the \textbf{Forman subdivision} of $M$. $K$ is defined as follows:
\begin{itemize}
  \item
    the $0$-cells of $K$ correspond to the $p$-cells of $M$ for $0 \le p \le d$;
  \item
    the cochains of $K$ correspond to the $p$-forms on $M$ as follows: if $(c_q \to b_{q - p})$ is a basis cochain of $K$ (thought as a basis cell of $K$), then its coboundary consists of the cells corresponding to the cochains $(a_{q + 1} \to b_{q - p})$ for $a_{q + 1} \succ c_{q}$ with relative orientation $\varepsilon(a_{q + 1}, c_q)$, and the cochains $(c_q \to a_{q - p - 1})$ for $a_{q - p - 1} \preceq b_{q - p}$ with relative orientation $(-1)^{p + 1} \varepsilon(b_{q - p}, a_{q - p - 1})$;
  \item
    the coboundary operator $\delta_K$ on $K$ is given by \Cref{eq:discrete_exterior_derivative} under the above-mentioned identification of forms on $M$ with cochains on $K$.
\end{itemize}
For the topological operations on a mesh, the geometric positions of the $0$-cells of $K$ with respect to the corresponding $p$-cells of $M$ is irrelevant. However, for the metric operations these positions are essential. The choice adopted here is to place the $0$-cells of $K$ at the centroids of the corresponding $p$-cells of $M$. It is not claimed that this choice is optimal, but it is used in the numerical simulations presented later. 
\begin{definition}
  Call the constructed bijection between $\Omega^{\bullet} M$ and $C^{\bullet} K$ the \textbf{Forman isomorphism} $F$. Indeed, it provides an isomorphism
  \begin{equation*}
    (\Omega^{\bullet} M, D) \cong (C^{\bullet} K, \delta).
  \end{equation*}
\end{definition}
\begin{remark}
  In the construction of $K$ the orientations of the cells of $K$ were not given and, yet, we stated that $\delta_K$ is the coboundary operator on $K$. In fact, since $\delta_K \circ \delta_K = 0$, and $\delta_K$ has the topological structure of a coboundary operator, we can recover the  orientations of the cells such that $\delta_K$ is the associated coboundary operator of $K$. This is done recursively as follows. First orient $0$-cells positively. Then orient $1$-cells in such a way that $\delta_K^0$ acts properly, i.e., in such a way that the relative orientations coincide with the respective coefficients in the matrix representation of the coboundary operator. Continue this for $\delta_K^1$ and $2-$cells, ..., $\delta_K^{d - 1}$ and $d$-cells. 
\end{remark}
\begin{example}
  \Cref{fig:triangulation_0p5} shows: (a) $M$ - a square divided into two triangles region; and (b) $K$ - the Forman subdivision of $M$. Both meshes are numerated and oriented; the orientation of $K$ is the induced one as discussed. Different colours are used for different types of $p$-cells in $K$:

  The nodes in $K$: from $1$ to $4$ coincide with those in $M$; from $5$ to $9$ are the midpoints of the edges of $M$; and $10$ and $11$ are the centroids of the faces of $M$. The nodes in $K$ are coloured according to the colour of the original cell they represent.

  The edges in $K$: from $1$ to $10$ correspond to edge-node pairs in $M$; from $11$ to $16$ correspond to face-edge pairs. Different colours are used for the two different types of edges.

  The faces in $K$ correspond to face-node pairs in $M$.

  Regarding orientation, the following examples using \Cref{eq:discrete_exterior_derivative} are given (they are consistent with the drawn orientation of the cells of $K$):
  \begin{equation*}
    \begin{split}
      D^0(E_3 \to E_3)
        & = \varepsilon(F_1, E_3) (F_1 \to E_3) + \varepsilon(F_2, E_3) (F_2 \to E_3) - \varepsilon(E_3, N_2) (N_3 \to E_2) - \varepsilon(E_3, N_4) (E_3 \to N_4) \\
        & = (F_1 \to E_3) - (F_2 \to E_3) + (E_2 \to N_3) - (E_3 \to N_4)
    \end{split}
  \end{equation*}
  which in $K$ corresponds to $\delta_K^0 N^7 = E^{13} - E^{14} + E^5 - E^6$;
  \begin{equation*}
    D^0(F_1 \to E_1) = \varepsilon(E_1, N_1) (F_1 \to N_1) + \varepsilon(E_1, N_4) (F_1 \to N_4) = (F_1 \to N_1) - (F_1 \to N_4)
  \end{equation*}
  which in $K$ corresponds to $\delta_K^1 E^{11} = F^2 - F^1$;
  \begin{equation*}
    D^0(E_4 \to N_3) = \varepsilon(F_2, E_4) (F_2 \to N_3) = (F_2 \to N_3)
  \end{equation*}
  which in $K$ corresponds to $\delta_K^1 E^8 = F^5$.
\end{example}
\begin{figure}[!ht]
  \begin{subfigure}{.45\textwidth}
    \centering
    \includegraphics[scale=.85]{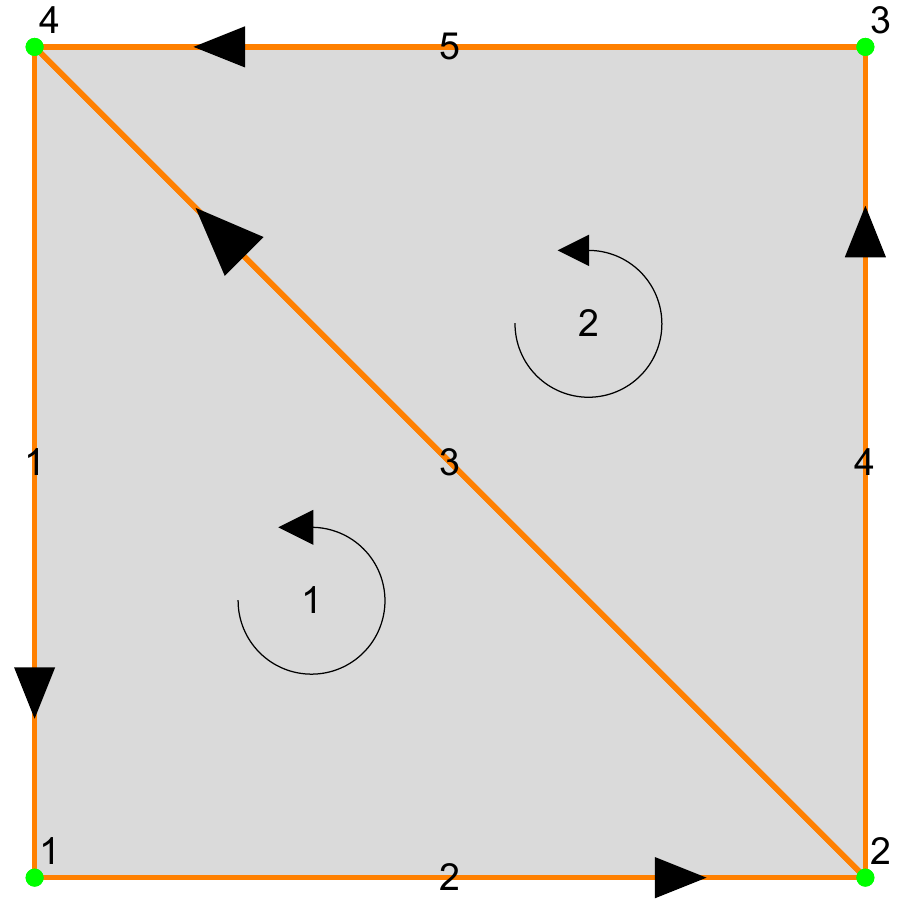}
    \caption{}
  \end{subfigure}	
  \begin{subfigure}{.45\textwidth}
    \centering
    \includegraphics[scale=.85]{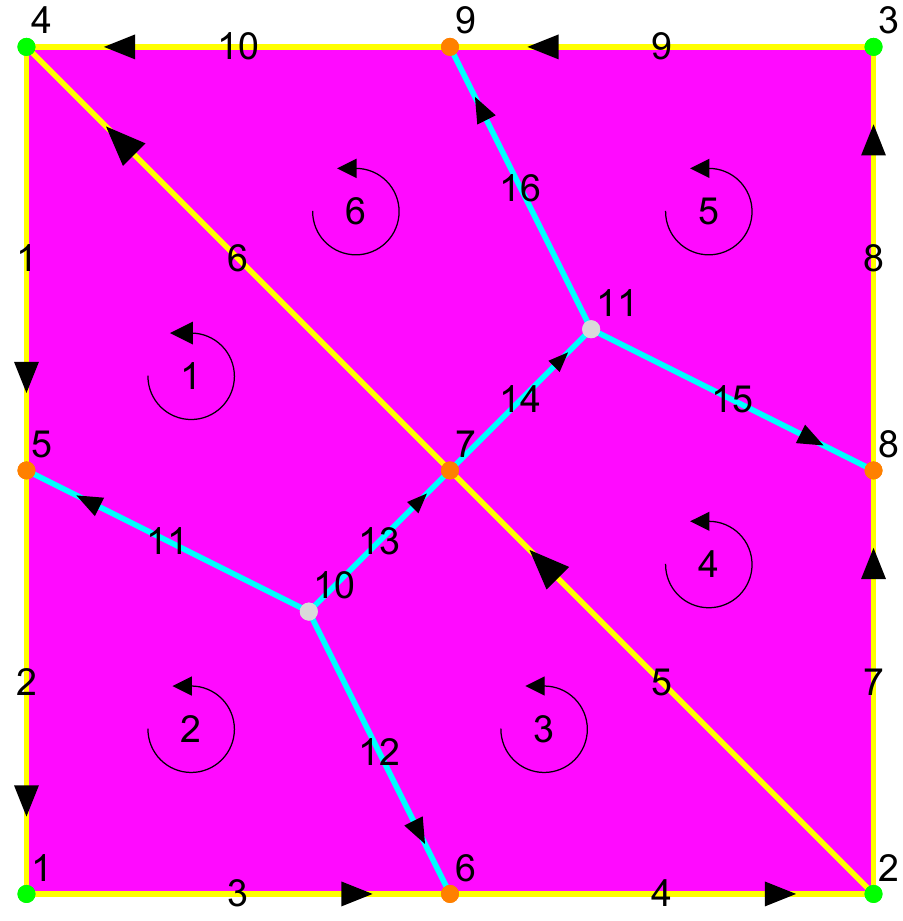}
    \caption{}
  \end{subfigure}
  \caption{Triangulation (a) and its Forman subdivision (b)}
  \label{fig:triangulation_0p5}
\end{figure}
\begin{remark}
  The geometric construction of $K$ is not suitable for all types of convex meshes with $d>2$. For example for any convex mesh $M$ with $d =2$ the resulting $K$ consists of quadrilaterals, but for a convex mesh $M$ with $d>2$ several outcomes for $K$ are possible. It is sufficient to consider $M$ as a single 3-cell. The possible outcomes for $K$ are illustrated in \Cref{fig:3d_forman}:
  \begin{itemize}
    \item
      if $M$ is a parallelepiped, then $K$ consists of $8$ equal parallelepipeds (\Cref{fig:cube_forman});
    \item
      if $M$ is a tetrahedron, then any $K$ consists of $4$ hexahedrons, referred to as quasi-cubes (\Cref{fig:tetrahedron_forman});
    \item
      if $M$ is a general polyhedron with \textbf{cubical corners}, i.e., every corner is connected to exactly $3$ edges, then $K$ consists only of quasi-cubes. However, the faces with one vertex at the centroid of $M$ may be non-planar quadrilaterals. Hence, $K$ may not be strictly a polytopal mesh (\Cref{fig:hexahedron_forman});
    \item
      if $M$ is a general polyhedron with at least one non-cubical corner, then $K$ is not quasi-cubical and the theory developed in this work does not apply (\Cref{fig:pyramid_forman}).
  \end{itemize}
  The developments in this work exclude meshes with non-cubical corners, so that $K$ is always a quasi-cubical mesh. Furthermore, if $K$ is not strictly polytopal, the areas of non-planar quadrilaterals are found by dividing these quadrilaterals into two triangles and summing up the two areas. This process is used to calculate also the volumes of 3-cells with non-planar boundary 2-cells.
\end{remark}
\begin{figure}[!ht]
  \begin{subfigure}{.45\textwidth}
    \centering
    \includegraphics[scale=1]{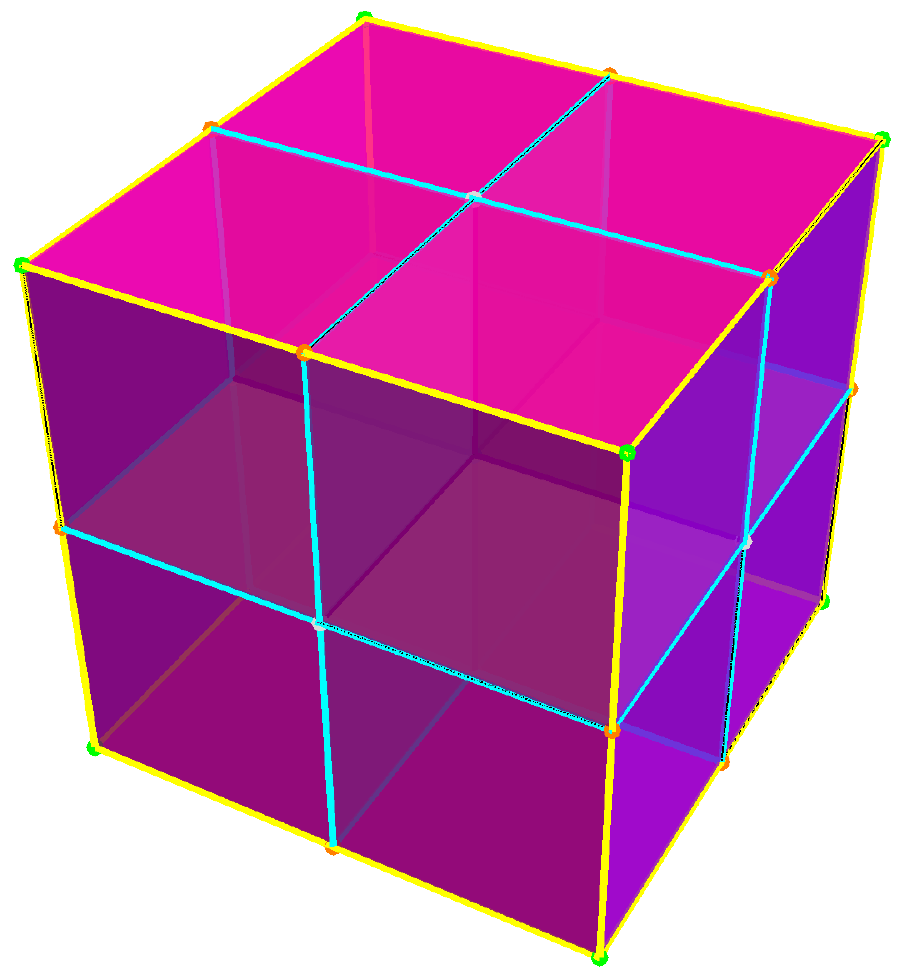}
    \caption{}
    \label{fig:cube_forman}
  \end{subfigure}	
  \begin{subfigure}{.45\textwidth}
    \centering
    \includegraphics[scale=1]{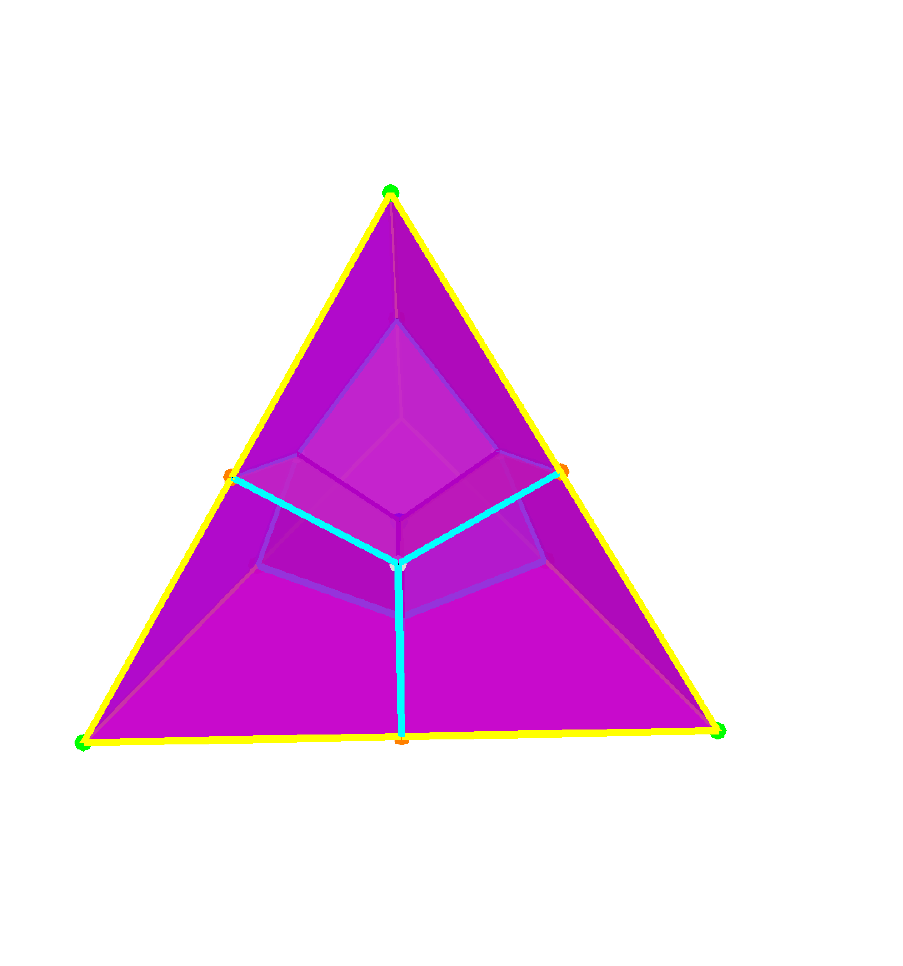}
    \caption{}
    \label{fig:tetrahedron_forman}
  \end{subfigure}
  
  \begin{subfigure}{.45\textwidth}
    \centering
    \includegraphics[scale=1]{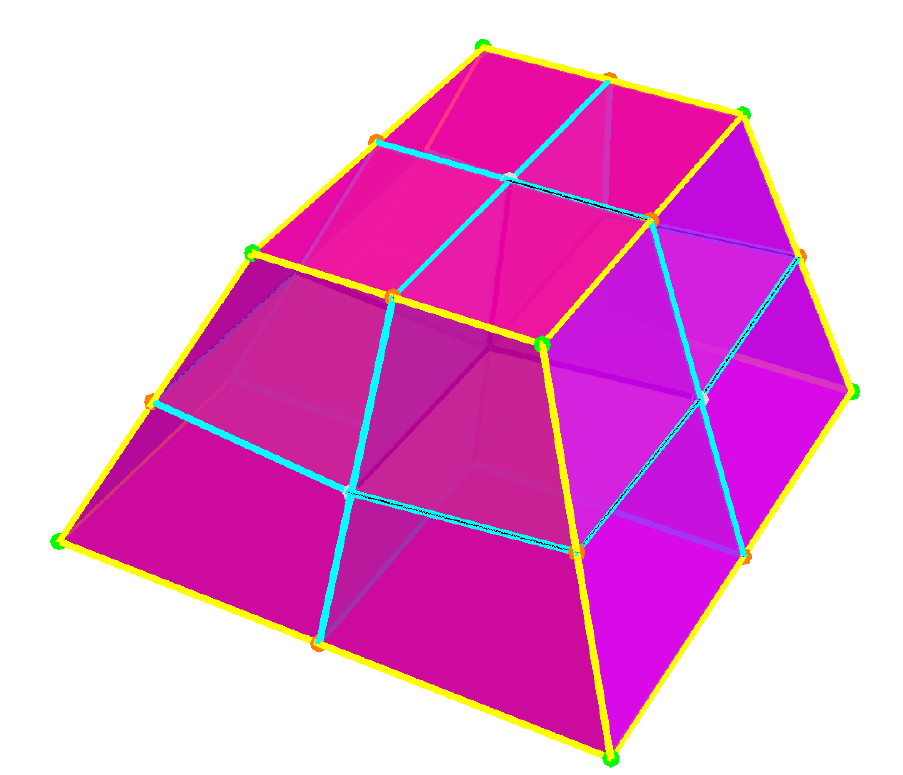}
    \caption{}
    \label{fig:hexahedron_forman}
  \end{subfigure}	
  \begin{subfigure}{.45\textwidth}
    \centering
    \includegraphics[scale=1]{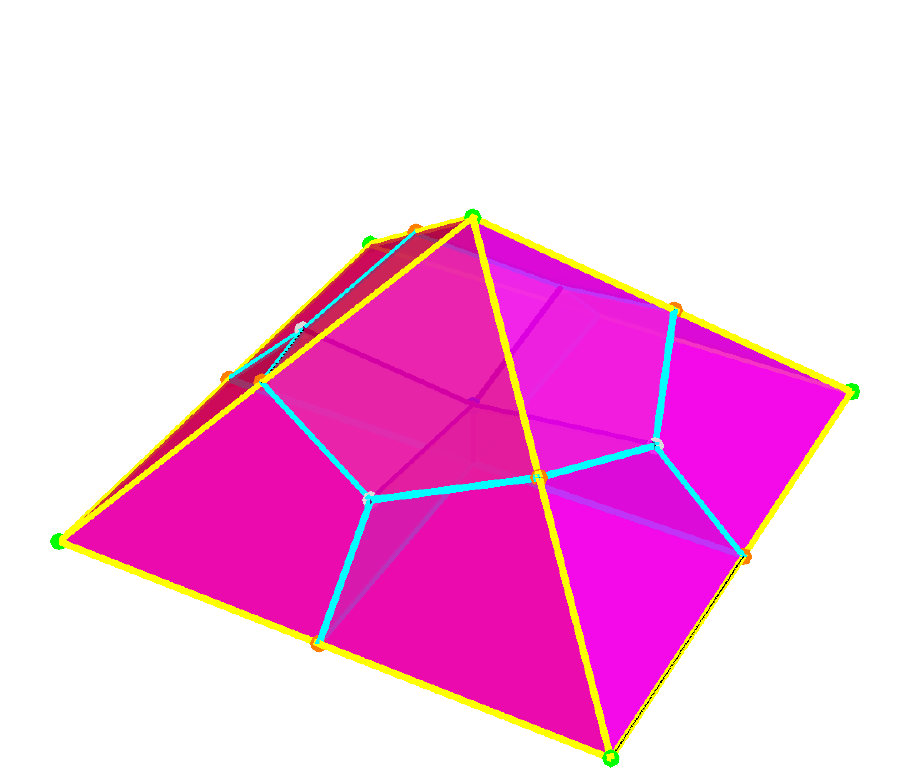}
    \caption{}
    \label{fig:pyramid_forman}
  \end{subfigure}
  \caption{The Forman subdivision of: (a) cube; (b) tetrahedron; (c) hexahedron; (d) square pyramid}
  \label{fig:3d_forman}
\end{figure}

\subsection{Quasi-cubical cup product of cochains and wedge product of forms}
\label{sec:wedge}
The following is a generalisation of
\cite[Definition 3.2.1]{arnold2012discrete} from cubes to arbitrary quasi-cubes.
\begin{definition}
\label{thm:cup_product}
  Let $K$ be an oriented quasi-cubical mesh. The \textbf{quasi-cubical cup product} is the unique bilinear map $\smile \colon C^{\bullet} K \times C^{\bullet} K \to C^{\bullet} K$ (with $\pi^p \smile \rho^q \in C^{p + q} K$) defined for basis cochains $a^p$ and $b^q$ as follows. 
  \begin{itemize}
    \item
      If $a_p \cap b_q = \emptyset$ or $\dim({\rm Aff}(a_p \cup b_q)) < p + q$, then $a^p \smile b^q := 0$.
    \item
      If $a_p \cap b_q \neq \emptyset$ and $\dim({\rm Aff}(a_p \cup b_q)) = p + q$, in which case $a_p \cap b_q$ is a point, then since $K$ is a convex mesh there exists at most one $c_{p +q} \in K_{p + q}$ such that $c_{p + q} \succeq a_p$ and $c_{p + q} \succeq b_q$.
      If there is no such $c_{p + q}$, then again $a^p \smile b^q := 0$. If there is such $c_{p + q}$, then
      \begin{equation}
        \label{eq:cup_product}
        a^p \smile b^q := \frac{1}{2^{p + q}}
          \frac{\OR(a_p) \wedge \OR(a_q)}{\OR(c_{p + q})}\ c^{p + q}.
      \end{equation}
      In this formula, $\wedge$ denotes the wedge product of the exterior algebra and $\OR(\alpha)$ is the orientation of a cell $\alpha$ as discussed in \ref{sec:orientation_mesh}.
  \end{itemize}
\end{definition}
\begin{example}
  Some values of the cup product of basis cochains are given in \Cref{fig:triangulation_0p3_forman_cup}.

  $N^6 \smile N^6 = N^6$, $N^2 \smile E^{16} = E^{16} \smile N^2 = E^{16} / 2$, $N^{16} \smile F^8 = F^8 \smile N^{16} = F^8 / 4$.

  $E^5 \smile E^{10} = -E^{10} \smile E^5 = F^{12} / 4$.

  $E^7 \smile E^8 = 0$ because the $1$-cells do not share a common $2$-cell, although they intersect.

  $E^7 \smile E^{22} = 0$ because the $1$-cells do not intersect, although they share a common $2$-cell.
\end{example}
\begin{figure}[!ht]
  \begin{subfigure}{.45\textwidth}
    \centering
    \includegraphics[scale=.85]{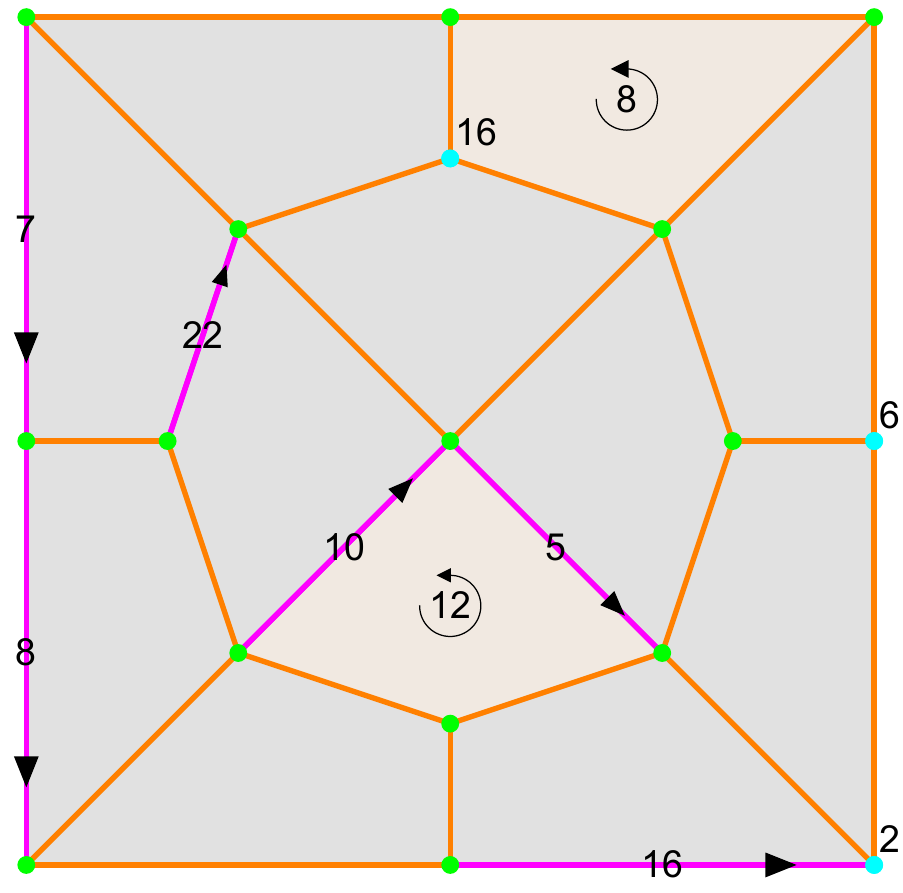}
    \caption{}
    \label{fig:triangulation_0p3_forman_cup}
  \end{subfigure}
  \begin{subfigure}{.45\textwidth}
    \centering
    \includegraphics[scale=.85]{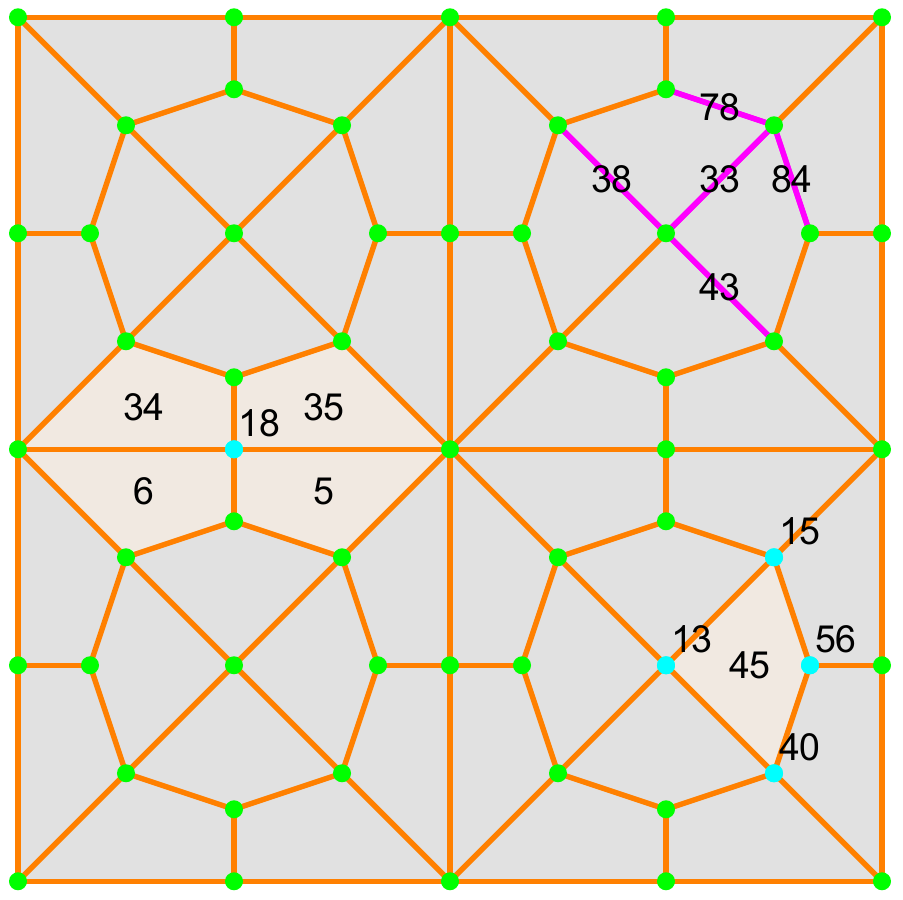}
    \caption{}
    \label{fig:triangulation_0p1-hodge}
  \end{subfigure}
  \label{fig:triangulation_0p1}
  \caption{Examples of: (a) cup product; (b) Hodge star}
\end{figure}
Let $\overline{\mathfrak{1}_M} := F(\mathfrak{1}_M)$.
\begin{theorem}
  $\smile$ satisfies the following properties (see
  \cite[Definition 2.3.2]{arnold2012discrete}):
  \begin{enumerate}
    \item
      if $a^p \smile b^q$ is nonzero, then those cells are boundaries of a common $(p + q)$-cell;
    \item
      $\delta(a^p \smile b^q) = (\delta a^p) \smile b^q + (-1)^p a^p \smile (\delta b^q)$;
    \item
      $\overline{\mathfrak{1}_M} \cup a^p = a^p \cup \overline{\mathfrak{1}_M} = a^p$.
  \end{enumerate}
\end{theorem}
\begin{proof}
  See \cite[Theorem 3.2.3]{arnold2012discrete}.
\end{proof}
\begin{definition}
\label{thm:wedge_product}
  Let $M$ be an oriented mesh with cubical corners and $\smile$ be the cup product on the Forman subdivision $K$. The \textbf{discrete wedge product} is the bilinear map $\wedge \colon \Omega^\bullet \times \Omega^\bullet \to \Omega^\bullet$ (with $\omega^p, \eta^q \in \Omega^{p + q} M$) defined on basis forms by ``transporting'' the definition of the cup product using the Forman isomorphism $F$:
  \begin{equation}
    \omega^p \wedge \eta^q := F^{-1}(F(\omega^p) \smile F(\eta^q)).
  \end{equation}
\end{definition}
The following identities are transformed from $C^\bullet K$ to $\Omega^\bullet M$ using $F^{-1}$.
\begin{align}
  \mathfrak{1}_M \wedge\ \omega & = \omega \wedge \mathfrak{1}_M = \omega. \\
  D(\omega^p \wedge \eta^q)
    & = (D\omega^p) \wedge \eta^q + (-1)^{p} \omega^p \wedge (D \eta^q). \\
  (\Omega^{\bullet} M, \wedge, D) 
    & \cong (C^{\bullet} K, \smile, \delta).
\end{align}
\begin{remark}
\label{thm:whitney_wilson_arnold}
  Scott Wilson introduced a similar cup product on simplicial meshes \cite[Definition 5.1]{wilson2007cochain} using Whitney forms \cite{whitney1957geometric}. He then proved that they were purely combinatorial, i.e., did not depend on the coordinates of the nodes \cite[Theorem 5.2]{wilson2007cochain}. He also showed that for suitable mesh refinements, the cochains are approximations of smooth differential forms. Rachel Arnold defined a cup on cubes \cite[Definition 3.2.1]{arnold2012discrete} and showed \cite[Theorem 3.2.12]{arnold2012discrete} that it coincided with the cup product defined using cubical Whitney forms \cite[Definition 3.2.10]{arnold2012discrete}. She did not define it in the general form for quasi-cubes shown in \Cref{thm:cup_product}. On the other hand, she defined a product of discrete forms on $M$ \cite[Definition 5.4.1]{arnold2012discrete} and constructed a cup product on $K$ \cite[Definition 5.4.3]{arnold2012discrete} using the Forman isomorphism. This is in contrast to our approach, when we start with an arbitrary polytopal mesh $M$ (with cubical corners) and use the combinatorial regularity of $K$ to define the cup product there and to ``pull it back'' to a wedge product on $M$.
\end{remark}

\section{Discrete metric operations}
\label{sec:geometry}
This section develops mesh operations requiring an additional structure on a mesh, namely an inner product. This allows to define adjoint coboundary operator, Laplacian and Hodge star, which have been explored in the literature but mainly for: (1) topological results independent of the choice of an inner product, such as the discrete Hodge theory discussed in \ref{sec:hodge_general} and \ref{sec:hodge_decomposition} and the discrete Poincar\`e duality referenced in \ref{sec:hodge_poincare}; and (2) convergence results dependent on a particular choice of an inner product - this is discussed in \cite{wilson2007cochain}. 

The main goal here is to define an inner product and its derivative notions suitable for solving physical problems. The novelty is the introduction of a discrete metric tensor (in fact, a whole class of metrics) and the resulting inner product defined via Riemann integration along a volume cochain. The theory is then applied to physical problems in \Cref{sec:applications}.

The section is developed from general to specific. First, the adjoint coboundary operator, $\delta^\star$, the Laplacian, $\Delta$, and the Hodge star operator, $\star$, are defined in \Cref{sec:inner_general} for a general choice of an inner product, $\inner{\cdot}{\cdot}$. These ingredients are sufficient for development of a topological theory - \ref{sec:hodge}. Second, explicit formulas for $\delta^\star$, $\Delta$ and $\star$ are given in \Cref{sec:inner_orthogonal} for the case where the basis cochains form an orthogonal basis of $C^p K$ with respect to $\inner{\cdot}{\cdot}$. The most important contribution is in \Cref{sec:metric_tensor} where a class of metric tensors is proposed leading to dimensional orthogonal inner products.

\subsection{General inner product}
\label{sec:inner_general}
Let $\inner{\cdot}{\cdot} \colon C^{p} K \times C^{p} K \to \R$ be an inner product (symmetric and positive definite bilinear map).
\begin{definition}
  The \textbf{adjoint coboundary} operator $\delta^{\star}_{p} \colon C^{p} K \to C^{p - 1} K$ is defined as the adjoint of $\delta_{p + 1},$ i.e., for any $\sigma^{p} \in C^{p} K,\ \tau^{p + 1} \in C^{p + 1} K$,
  \begin{equation}
    \inner{\delta_{p} \sigma^{p}}{\tau^{p + 1}} 
      = \inner{\sigma^{p}}{\delta^{\star}_{p + 1} \tau^{p + 1}}.
  \end{equation}
\end{definition}
\begin{definition}
  The \textbf{discrete Laplacian} is given by
  \begin{equation}
    \Delta_{p} 
      =   \delta_{p - 1} \circ \delta^{\star}_{p} 
        + \delta^{\star}_{p + 1} \circ \delta_{p}.
  \end{equation}
\end{definition}
\begin{definition}
  The \textbf{discrete Hodge star} on $p$-forms is the unique map $\star \colon C^{p} K \to C^{d - p} K$ such that for any $\sigma^{d - p} \in C^{d - p} K ,\ \tau^{p} \in C^{p} K,$
  \begin{equation}
    \inner{\sigma^{d - p}}{\star_{p} \tau^{p}} 
      = (\sigma^{d - p} \smile \tau^{p})[K],
  \end{equation}
  where $[K]$ is the fundamental class of the compatibly oriented manifold-like mesh $K$.
\end{definition}

\subsection{Orthogonal inner product}
\label{sec:inner_orthogonal}
When the basis cochains form an orthogonal basis with respect to the inner product, the operations introduced in \Cref{sec:inner_general} have nice closed forms. Similar formulas are derived in \cite[Section 2]{Forman2003_Bochner} where the adjoint coboundary operator and the Laplacian are defined on the chains of $M$ and the values of the inner product at basis chains are called \textbf{weights}.

\paragraph{Adjoint coboundary operator}
To compute $\delta^{\star}$, let 
\begin{equation*}
  \delta^{\star}_{p + 1} c^{p + 1} = \sum_{b^p} \lambda_{b^p} b^p
\end{equation*} 
for the unknown coefficients $\lambda_{b^p} \in \R$. Then
\begin{equation*}
  \inner{\delta_p a^p}{c^{p + 1}}
    = \inner{a^p}{\delta^{\star}_{p + 1} c^{p + 1}} 
    = \inner{a^p}{\sum_{b^p} \lambda_{b^p} b^p} 
    = \lambda_{a^p} \inner{a^p}{a^p}.
\end{equation*}
Hence, $\lambda_{a^p} = \inner{\delta_p a^p}{c^{p + 1}} / \inner{a^p}{a^p}$ and therefore
\begin{equation}
\label{eq:adjoint_coboundary}
  \delta^{\star}_{p + 1} c^{p + 1} 
    = \sum_{a^p} \lambda_{a^p} a^p
    = \sum_{a^p} \frac{\inner{\delta_p a^p}{c^{p + 1}}}{\inner{a^p}{a^p}} a^p 
    = \sum_{a^p \prec c^{p + 1}}
        \varepsilon(c_{p + 1}, a_p)
        \frac{\inner{c^{p + 1}}{c^{p + 1}}}{\inner{a^p}{a^p}} a^p \cdot
\end{equation}
The matrix of the adjoint coboundary operator with respect to the standard bases of $C^p K$ and $C^{p +1} K$ has the same structure as the matrix of the boundary operator with respect to the standard bases of chains of $C_p K$ and $C_{p +1} K$. Moreover, the signs are also the same as is evident by \Cref{eq:adjoint_coboundary} and the positive-definiteness of the inner product.

\paragraph{Laplacian of $0$-chains}
\begin{equation}
\label{eq:laplacian}
  \begin{split}
    \Delta_0\, c^0 = \delta^\star_1 (\delta_0 c^0)
      & = \delta^\star_1 \left( \sum_{b_1 \succ c_0}\, \varepsilon(b_1, c_0)\, b^1 \right) \\
      & = \sum_{b_1 \succ c_0} \sum_{a_0 \prec b_1} \varepsilon(b_1, c_0)\, \varepsilon(b_1, a_0)\, \frac{\inner{b^1}{b^1}}{\inner{a^0}{a^0}} a^0 \\
      & = \frac{1}{\inner{c^0}{c^0}} \left(\sum_{b_1 \succ c_0} \inner{b^1}{b^1} \right) c^0 
        - \sum_{\substack{a_0 \parallel c_0 \\ b^1 = \mathcal{E}(a^0, c^0)}} \frac{\inner{b^1}{b^1}}{\inner{a^0}{a^0}} a^0
  \end{split},
\end{equation}
where $a_0 \parallel c_0$ means that $a_0$ and $c_0$ share a common edge, and $\mathcal{E}(a^0, c^0)$ is the basis $1$-cochain corresponding to that edge.

\paragraph{Hodge star operator}
Let $q = d - p$ and
\begin{equation*}
  \star_p c^p
    = \sum_{b^q \in C^q K} \lambda_{b^q} b^q
\end{equation*}
for the unknowns $\lambda_{b^q} \in \R$. Then
\begin{equation*}
  (a^q \smile c^p)[K]
    = \inner{a^q}{\star_{p} c^p} 
    = \sum_{b^q} \lambda_{b^q} \inner{a^q}{b^q} 
    = \lambda_{a^q} \inner{a^q}{a^q}.
\end{equation*}
Hence, $\lambda_{a^q} = (a^q \smile c^p)[K] / \inner{a^q} {a^q}$ and therefore
\begin{equation}
\label{eq:hodge_star}
  \star_p c^p 
    = \sum_{a^{d - p}} \frac{(a^{d - p} \smile c^p)[K]}{\inner{a^{d - p}} {a^{d - p}}} a^{d - p} \cdot
\end{equation}
Evidently, the Hodge star is local, because nonzero numerators in the above summands are connected to to the $(d - p)$-faces of the $d$-superfaces of $c_p$. An example with the contributors to the Hodge star of different cells is give is given in \Cref{fig:triangulation_0p1-hodge}. 
\begin{align*}
  \star_0 N^{18} & =
      \frac{(F^5 \smile N^{18})[K]}{\inner{F^5}{F^5}} F^5
    + \frac{(F^6 \smile N^{18})[K]}{\inner{F^6}{F^6}} F^6
    + \frac{(F^{34} \smile N^{18})[K]}{\inner{F^{34}}{F^{34}}} F^{34}
    + \frac{(F^{35} \smile N^{18})[K]}{\inner{F^{35}}{F^{35}}} F^{35} \\
    & = \frac{1}{4} 
      \left( 
          \frac{F^5}{\inner{F^5}{F^5}}
        + \frac{F^6}{\inner{F^6}{F^6}}
        + \frac{F^{34}}{\inner{F^{34}}{F^{34}}}
        + \frac{F^{35}}{\inner{F^{35}}{F^{35}}}
      \right). \\
  \star_1 E^{33} & =
      \frac{(E^{33} \smile E^{38})[K]}{\inner{E^{38}}{E^{38}}} E^{38}
    + \frac{(E^{33} \smile E^{43})[K]}{\inner{E^{43}}{E^{43}}} E^{43}
    + \frac{(E^{33} \smile E^{78})[K]}{\inner{E^{78}}{E^{78}}} E^{78}
    + \frac{(E^{33} \smile E^{84})[K]}{\inner{E^{84}}{E^{84}}} E^{84} \\
    & = \frac{1}{4} 
      \left( 
        \pm \frac{E^{38}}{\inner{E^{38}}{E^{38}}}
        \pm \frac{E^{43}}{\inner{E^{43}}{E^{43}}}
        \pm \frac{E^{78}}{\inner{E^{78}}{E^{78}}}
        \pm \frac{E^{84}}{\inner{E^{84}}{E^{84}}}
      \right). \\
  \star_2 F^{45} & =
      \frac{(F^{45} \smile N^{13})[K]}{\inner{N^{13}}{N^{13}}} N^{13}
    + \frac{(F^{45} \smile N^{15})[K]}{\inner{N^{15}}{N^{15}}} N^{15}
    + \frac{(F^{45} \smile N^{40})[K]}{\inner{N^{40}}{N^{40}}} N^{40}
    + \frac{(F^{45} \smile N^{56})[K]}{\inner{N^{56}}{N^{56}}} N^{56} \\
    & = \frac{1}{4} 
      \left( 
          \frac{N^{13}}{\inner{N^{13}}{N^{13}}}
        + \frac{N^{15}}{\inner{N^{15}}{N^{15}}}
        + \frac{N^{40}}{\inner{N^{40}}{N^{40}}}
        + \frac{N^{56}}{\inner{N^{56}}{N^{56}}}
      \right)
\end{align*}
(the signs are generally different in the expression for $\star_1 E^{33}$ and depend on the orientation).

\subsection{Orthogonal inner product via a metric tensor}
\label{sec:metric_tensor}
Several discrete inner products have been proposed in the literature for topological studies, but these have not been developed in a canonical way by introduction of a discrete metric tensor and use of discrete Riemann integration. The canonical path is taken here by defining a class of discrete metric tensors. Two main variants are discussed: a trivial one and an extended one with curvature at nodes. The latter is used for the physical applications in \Cref{sec:applications}.

For a cell $c_p \in K_p$ let $\mu(c_p)$ denote the geometric measure of $c_p$: $1$ for 0-cells, length for 1-cells, area for 2-cells, and volume for 3-cells.
\begin{definition}
  The \textbf{discrete metric tensor} $g_p \colon C^p K \times C^p K \to C^0 K$ is the unique bilinear map such that $g_p(b^p, c^p) = 0$ if $b^p \neq c^p$ and
  \begin{equation}
    g_p(c^p, c^p) := \frac{1}{\mu(c_p)^2} \frac{1}{2^p} \sum_{b^0 \preceq c^p} \kappa(b_0) b^0,
  \end{equation}
  where $\kappa(b_0)$ is a dimensionless weight of $b_0$. Different choices will be discussed shortly.
\end{definition}
\begin{remark}
  The physical dimension of $g_p$ is $L^{-2 p}$.
\end{remark}
\begin{remark} 
  While the metric tensor is defined for quasi-cubical meshes, the formula can be applied to simplicial meshes by replacing $2^p$ in the denominator with $p + 1$ (the number of nodes of a simplex).
\end{remark}
\begin{definition}
  The \textbf{volume cochain} on $K$ is the $d$-cochain
  \begin{equation}
    \vol := \sum_{c^d \in C^d K} \mu(c_d) c^d.
  \end{equation}
\end{definition}
\begin{remark}
  The physical dimension of $\vol$ is $L^d$. 
\end{remark}
\begin{definition}
\label{thm:inner_product}
  The \textbf{inner product} of $p$-forms, corresponding to the metric $g$, is the symmetric bilinear map $\inner{\cdot}{\cdot} \colon C^{p} K \times C^{p} K \to \R$ defined by
  \begin{equation}
    \inner{\sigma^p}{\tau^p} := (g(\sigma^{p}, \tau^{p}) \smile \vol)[K].
  \end{equation}
\end{definition}
\begin{claim}
  The map defined in \Cref{thm:inner_product} is positive definite, i.e., it is indeed an inner product.
\end{claim}
\begin{remark}
  The physical dimension of $\inner{\cdot}{\cdot}_p$ is $L^{d - 2p}$. 
\end{remark}
\begin{remark}
  The physical dimension of $\delta^\star_p$ is $L^{(d - 2p) - (d - 2 (p - 1)} = L^{-2}$. The same holds for $\Delta_p$.
\end{remark}
\begin{remark}
  The physical dimension of the $\star_p$ is $L^{-(d - 2(d - p))} = L^{d - 2p}$. 
\end{remark}
\begin{remark}
  The definition of inner product is analogous to the one given in smooth Riemannian geometry. Indeed, define the \textbf{Riemann integral} of a $0$-cochain $f$ by
  \begin{equation*}
    \int_{K} f \smile \vol := (f \smile \vol)[K].
  \end{equation*}
  Then
  \begin{equation*}
    \inner{\sigma^p}{\tau^p}
      = \int_{K} g(\sigma^{p}, \tau^{p}) \smile \vol.
  \end{equation*}
\end{remark}
Obviously $\inner{a^p}{b^p} = 0$ if $a^p \neq b^p$. Otherwise,
\begin{equation}
\label{eq:inner_product}
  \inner{c^p}{c^p}
    = \frac{1}{2^p \mu(c_p)^2}
      \sum_{b_0 \preceq c_p} \kappa(b_0) (b^0 \smile \vol)[K]
    = \frac{1}{2^{p + d} \mu(c_p)^2} 
      \sum_{b_0 \preceq c_p} \kappa(b_0) \sum_{a_d \succeq b_0} \mu(a_d).
\end{equation}
For a particular example of the contributors to the inner product, see \Cref{fig:triangulation_0p3_forman_inner}, where 
\begin{equation*}
  \inner{E^5}{E^5} = \frac{1}{2 \mu(E_5)^2} (\kappa(N_5) (\mu(F_3) + \mu(F_6) + \mu(F_7) + \mu(F_{12})) + \kappa(N_8) (\mu(F_1) + \mu(F_3) + \mu(F_{11}) + \mu(F_{12}))).
\end{equation*}
\begin{figure}[!ht]
  \begin{subfigure}{.45\textwidth}
    \centering
    \includegraphics[scale=.85]{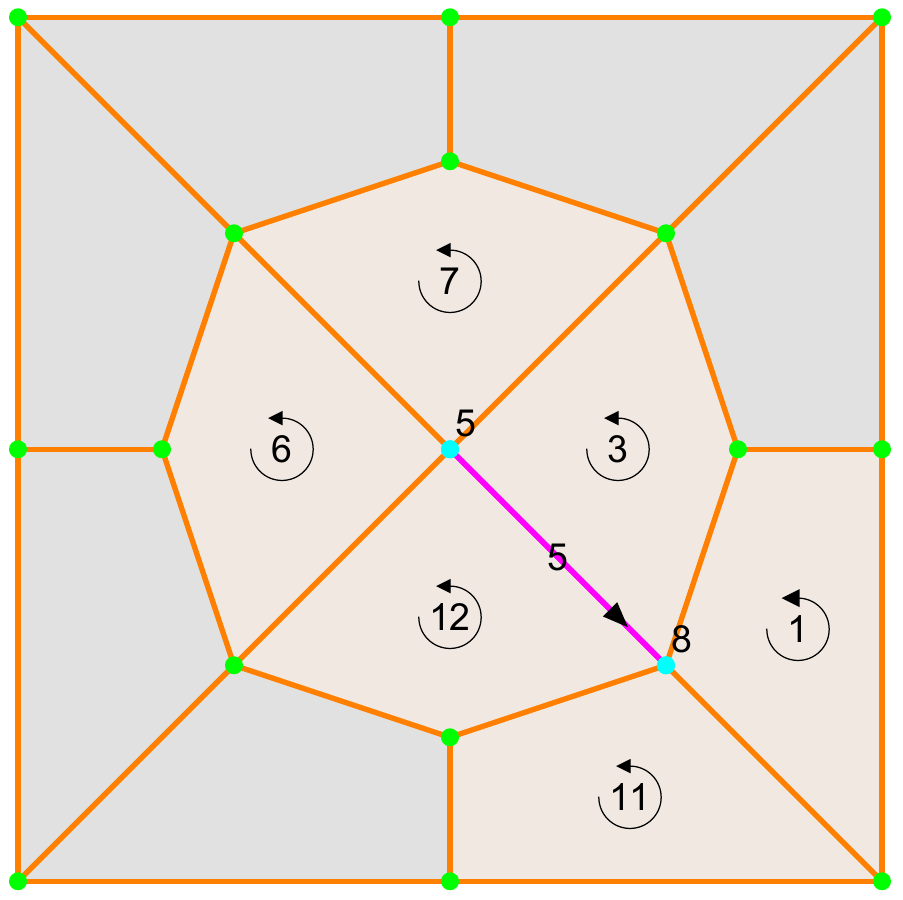}
    \caption{}
    \label{fig:triangulation_0p3_forman_inner}
  \end{subfigure}
  \begin{subfigure}{.45\textwidth}
    \centering
    \includegraphics[scale=.85]{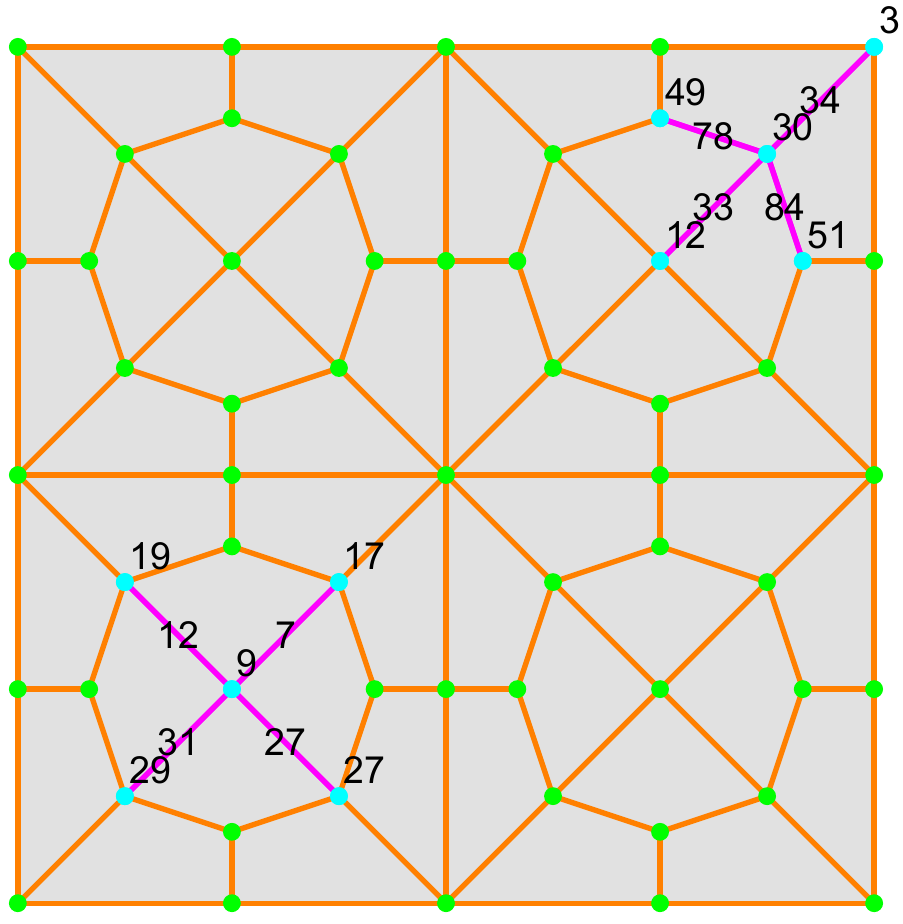}
    \caption{}
    \label{fig:triangulation_0p1_forman_laplacian}
  \end{subfigure}
  \caption{Examples of: (a) inner product; (b) Laplacian}
\end{figure}
A trivial choice for the weights of 0-cells is $\kappa = 1$. This leads to some nice properties similar to the continuum case as shown by the following three examples.
\begin{example}
  \begin{equation}
    g(\overline{\mathfrak{1}_M}, c^0)
      = g\left(\sum_{a^0 \in C^0 K} a^0, c^0\right)
      = g(c^0, c^0)
      = c^0.
  \end{equation}
  In particular, 
  $g(\overline{\mathfrak{1}_M}, \overline{\mathfrak{1}_M})
    = \overline{\mathfrak{1}_M}$.
\end{example}
\begin{example}
  \begin{equation}
    \inner{\overline{\mathfrak{1}_M}}{\overline{\mathfrak{1}_M}}
      = (\overline{\mathfrak{1}_M} \smile \vol)[K]
      = \vol[K]
      = \mu(K)
      = \mu(M).
  \end{equation}
  (here $\mu(K) = \mu(M)$ is the measure of the whole region $M$ and $K$ represent).
\end{example}
\begin{example}
  \begin{equation}
    \star_n \vol
      = \sum_{a^0} \frac{(a^0 \smile \vol)[K]}{\inner{a^0} {a^0}} a^0
      = \sum_{a^0} \frac{(a^0 \smile \vol)[K]}{(a^0 \smile \vol)[K]} a^0
      = \sum_{a^0} a^0
      = \overline{\mathfrak{1}_M}.
  \end{equation}
\end{example}
While the choice $\kappa = 1$ for all 0-cells leads to known identities from smooth Riemannian geometry, and might be the appropriate choice for closed discrete manifolds, it is not suitable for discrete manifolds with boundary. The calculation of the adjoint coboundary operator close to the boundary requires special attention as illustrated with the following remark.
\begin{remark}
  Consider $M$ to be a partition of the interval $[0, 1]$ into $n \geq 2$ equal parts (in the picture below $n = 2$). Then $K$ is the partition of $[0, 1]$ into $2n$ equal parts. Label the nodes from $0$ to $2n$, the edges from $1$ to $2n$, and orient all edges in the positive direction, i.e., $\varepsilon(E_i, N_{i - 1}) = -1,\ \varepsilon(E_i, N_i) = 1$, $i = 1, ..., 2n$. Let $h = 1/ (2n)$.
  
  \begin{tikzpicture}[thick]
    \draw [->] (0,0) -- (2,0);
    \draw [->] (2,0) -- (4,0);
    \draw [->] (4,0) -- (6,0);
    \draw [->] (6,0) -- (8,0);
    \draw (0,0) node[anchor=south]{$N_0$};
    \filldraw [black] (0,0) circle (1pt);
    \draw (2,0) node[anchor=south]{$N_1$};
    \filldraw [black] (2,0) circle (1pt);
    \draw (4,0) node[anchor=south]{$N_2$};
    \filldraw [black] (4,0) circle (1pt);
    \draw (6,0) node[anchor=south]{$N_3$};
    \filldraw [black] (6,0) circle (1pt);
    \draw (8,0) node[anchor=south]{$N_4$};
    \filldraw [black] (8,0) circle (1pt);
    \draw (1,0) node[anchor=south]{$E_1$};
    \draw (3,0) node[anchor=south]{$E_2$};
    \draw (5,0) node[anchor=south]{$E_3$};
    \draw (7,0) node[anchor=south]{$E_4$};
  \end{tikzpicture}
  
  Since boundary nodes have only half the volumes around them,
  \begin{equation*}
    \begin{split}
      \inner{N^i}{N^i} = h / 2,\ i \in \{0, 2n \} \\
      \inner{N^i}{N^i} = h,\ i \in \{1, ..., 2n - 1 \}
    \end{split},
  \end{equation*}
  \begin{equation*}
    \begin{split}
      \inner{E^i}{E^i} = \frac{1}{2 h^2}(h / 2 + h) = \frac{3}{4h},\ i \in \{1, 2n\} \\
      \inner{E^i}{E^i} = \frac{1}{2 h^2}(h + h) = \frac{1}{h},\ i \in \{2, ..., 2n - 1\}
    \end{split}.
  \end{equation*}
  Hence, on the one hand
  \begin{equation*}
    \begin{split}
      \Delta_0 N^2 
        & = \frac{1}{\inner{N^2}{N^2}} (\inner{E^2}{E^2} + \inner{E^3}{E^3}) N^2 
          - \frac{\inner{E^2}{E^2}}{\inner{N^1}{N^1}} N^1
          - \frac{\inner{E^3}{E^3}}{\inner{N^3}{N^3}} N^3 \\
        & = \frac{1 / h + 1 / h}{h / 2}  N^2 - \frac{1 / h}{h} N^1 - \frac{1 / h}{h} N^3 \\
        & = \frac{1}{h^2} \left(2 N^2 - N^1 - N^3 \right),
    \end{split}
  \end{equation*}
  which has the same form as the finite difference Laplacian. However, on the other hand
  \begin{equation*}
    \begin{split}
      \Delta_0 N^1 
        & = \frac{1}{\inner{N^1}{N^1}} (\inner{E^1}{E^1} + \inner{E^2}{E^2}) N^1 
          - \frac{\inner{E^1}{E^1}}{\inner{N^0}{N^0}} N^0
          - \frac{\inner{E^2}{E^2}}{\inner{N^2}{N^2}} N^2 \\
        & = \frac{3 / (4 h) + 1 / h}{h / 2}  N^1 - \frac{3 / (4 h)}{h / 2} N^0 - \frac{1 / h}{h} N^2 \\
        & = \frac{1}{h^2} \left( \frac{7}{8} N^1 - \frac{3}{8} N^0 - N^2 \right),
    \end{split}
  \end{equation*}
  which is not what is expected from an interior node. The problem arises when the Laplacian acting on a node adjacent to the boundary uses the inner product of a boundary $0$-cochain with itself, which does not have all the volumes around. This is typical for any regular grid: the equations corresponding to interior nodes which do not have boundary neighbours are given by (minus) the finite difference Laplacian, but the equations corresponding to boundary nodes and almost boundary nodes (nodes with boundary neighbours) are different.
  
  For a 2D example, see \Cref{fig:triangulation_0p1_forman_laplacian}. The vertex $N_9$ has full volumes around its neighbouring cells, and so full Laplacian
  \begin{equation*}
    \begin{split}
      \Delta_0 N^9 & =
        \frac{1}{\inner{N^{9}}{N^{9}}}
          (\inner{E^{7}}{E^{7}} + \inner{E^{12}}{E^{12}} + \inner{E^{27}}{E^{27}} + \inner{E^{31}}{E^{31}}) N^9 \\
        & - \frac{\inner{N^{17}}{N^{17}}}{\inner{E^{7}}{E^{7}}} N^{17}
        - \frac{\inner{N^{19}}{N^{19}}}{\inner{E^{12}}{E^{12}}} N^{19}
        - \frac{\inner{N^{27}}{N^{27}}}{\inner{E^{27}}{E^{27}}} N^{27}
        - \frac{\inner{N^{29}}{N^{29}}}{\inner{E^{31}}{E^{31}}} N^{29}.
    \end{split}
  \end{equation*}
  The Laplacian at $N^{30}$ is:
  \begin{equation*}
    \begin{split}
      \Delta_0 N^{30} & =
        \frac{1}{\inner{N^{9}}{N^{9}}}
          (\inner{E^{33}}{E^{33}} + \inner{E^{34}}{E^{34}} + \inner{E^{78}}{E^{78}} + \inner{E^{84}}{E^{84}}) N^{30} \\
        & - \frac{\inner{N^{12}}{N^{12}}}{\inner{E^{33}}{E^{33}}} N^{12}
        - \frac{\inner{N^{3}}{N^{3}}}{\inner{E^{34}}{E^{34}}} N^{3}
        - \frac{\inner{N^{49}}{N^{49}}}{\inner{E^{78}}{E^{78}}} N^{49}
        - \frac{\inner{N^{51}}{N^{51}}}{\inner{E^{84}}{E^{84}}} N^{51}.
    \end{split}
  \end{equation*}
  Because $N_3$ is on the boundary, $\inner{N^{3}}{N^{3}}$ and $\inner{E^{34}}{E^{34}}$ will have less contributions from surrounding volumes and $\Delta_0 N^{30}$ is not ``full'' Laplacian.
\end{remark}

This problem is addressed here by selecting the weights of 0-cells to be equal to 0-cell curvatures defined as follows. Let $A$ be an affine space. For $z_0 \in A$ and $r \in\R^+$ let $\S_{d - 1}(c_0, r)$ be the sphere with centre $c_0$ and radius $r$, $\S_{d - 1}(1)$ be the unit sphere with any centre (its measure, used below, does not depend on the centre).
\begin{definition}
  Let $c_d$ be a $d$-polytope with affine hull $A$, $a_0$ be a node of $c_d$, $\Theta$ be the cone in $A$ centered at $a_0$ and bounded by the edges starting from $a_0$.
  Let $\S_{d - 1}(a_0, r)$ be the $(d - 1)$-sphere in $A$ centered at $a_0$ with radius $r$, i.e., the boundary of the corresponding $d$-ball. The \textbf{angle measure} $\angle(c_d, a_0)$ of $\Theta$ is defined as the ratio between the surface measure of $\Theta \cap \S_{d - 1}(a_0, r)$ and $\S_{d - 1}(a_0, r)$. The definition does not depend on the radius $r$ because both measures are proportional to $r^{d - 1}$.
\end{definition}
\begin{remark}
  If $d = 0$, then $\theta$ is always $1$; if $d = 1$, then $\theta$ is always $1 / 2$; if $d = 2$, then $\theta$ is the radian measure of a planar angle (between $0$ and $2 \pi$ and less than $\pi$ for convex polygons); if $d = 3$, then $\theta$ is the steradian measure of a solid angle (between $0$ and $4 \pi$ and less than $2 \pi$ for convex polyhedrons).
\end{remark}
\begin{definition}
  Let $K$ be a $d$-mesh embeddable in $\R^d$ and $a_0$ be a node in $K$. The \textbf{(node) curvature} of $a_0$, and therefore its weight in the metric tensor, is defined by
  \begin{equation}
    \kappa(a_0) = \frac{\mu(\S_{d - 1}(1))}{\sum_{c_d \succeq a_0} \angle(c_d, a_0)} \in [1, \infty).
  \end{equation}
\end{definition}
\begin{remark}
  The curvature of an interior node of $K$ is always $1$; the curvature of a boundary node is different from $1$. Consider, for example, the case of a cubical domain. If $a_0$ lies on a domain face, but not on a domain edge, its curvature is $2$. If $a_0$ lies on a domain edge, but not on a domain corner, its curvature is $4$. If $a_0$ is a domain corner, its curvature is $8$.
\end{remark}
For a regular grid this choice of $\kappa$ leads to the same Laplacian as in the finite difference method for all interior nodes. The equations at the boundary also have nice form for such a grid, but are not presented here, as the theory is developed and valid for general polyhedral meshes.

\section{Applications}
\label{sec:applications}
Applied to $0$-forms, the Laplacian has the form $\Delta = \delta^{\star} \circ \delta$, explicitly given by \Cref{eq:laplacian}. The discrete version of the heat/diffusion equation is obtained by introducing the physical property, diffusivity, to modify the Laplacian to
\begin{equation}
\label{eq:mod_laplacian}
  \Delta_0^\alpha := \delta^\star_1 \circ \alpha \circ \delta_0,
\end{equation}
where $\alpha \colon C^1 K \to C^1 K$ is a symmetric positive definite map. Specifically for the examples in this work, it is assumed that $\alpha(b^1) = \alpha_{b^1} b^1$, where $\alpha_{b^1} > 0$ is the local diffusivity of cell $b_1$. In such case the exact formula for $\Delta_0^\alpha\, c^0$ is calculated analogously to \Cref{eq:laplacian} and reads
\begin{equation}
  \Delta_0^\alpha c^0 = 
    \frac{1}{\inner{c^0}{c^0}} \left(\sum_{b_1 \succ c_0} \alpha_{b^1} \inner{b^1}{b^1} \right) c^0 
      - \sum_{\substack{a_0 \parallel c_0 \\ b^1 = \mathcal{E}(a^0, c^0)}} \alpha_{b^1} \frac{\inner{b^1}{b^1}}{\inner{a^0}{a^0}} a^0.
\end{equation}
The construction of $K$ from $M$ provides three types of $1$-cells in $K$ associated with: (1) pairs $(b_0 \prec b_1) \in M$ (i.e., along $1$-cells/edges of $M$); (2) pairs $(b_1 \prec b_2) \in M$ (i.e., along $2$-cells/faces of $M$); and (3) pairs $(b_2 \prec b_3) \in M$ (i.e., through $3$-cells/volumes of $M$). Considering that $M$ is a representation of a material with internal structure, the three types of $1$-cells in $K$ allow for associating different diffusivity to components of $M$ with different geometric dimensions. This provides a considerable advantage for the proposed theory - simultaneous analysis of processes taking place with different rates on microstructural components of different dimensions - which cannot be accomplished with numerical methods based on the continuum formulation.

The discrete version of the heat/diffusion equation without body sources reads
\begin{equation}
\label{eq:heat_discrete}
  \frac{\partial \sigma^0}{\partial t} = \Delta_0^\alpha\, \sigma^0,
\end{equation}
where $\sigma^0$ is the $0$-cochain of the unknown scalar variable. Importantly, in the discrete formulation the fluxes correspond to area integrated continuum fluxes, i.e. to the total fluxes rather than to flux densities used in \Cref{sec:motivation}. Specifically, the discrete flux along a 1-cell, $b_1$ is given by 
\begin{equation}
\label{eq:flux_discrete}
    f(b_1) = -\alpha_{b^1}\, \inner{b^1}{b^1}\, (\delta_0 \sigma^0)(b_1),
\end{equation}
with physical dimension $[f]=[\sigma^0]L^3/T$.

All mathematical operations described in the paper, leading to the system \Cref{eq:heat_discrete} are implemented in MATLAB and the code is available at: \href{https://github.com/boompiet/Forman_MATLAB}{https://github.com/boompiet/Forman\_MATLAB}. This repository contains also the meshes used for the simulations described in the following sub-sections.

\subsection{Numerical simulations}

Simple numerical simulations are presented to highlight some of the features of the proposed theory. These include simulations on regular-orthogonal and quasi-random meshes to show the influence of geometric variation, as well as application to electrical diffusivity of a composite including graphene nano-plates (2D) and carbon nano-tubes (1D) in a polymer matrix (3D) to demonstrate simultaneous simulation on elements of different geometric dimension. The quasi-random meshes are generated using the freely available software for Voronoi-type tessellations Neper: \href{https://neper.info}{https://neper.info}.

Dirichlet boundary conditions are applied by multiplying the columns of the system matrix associated with points on the given boundary by the prescribed values and subtracting the sum from the left-hand side of the equation. The rows and columns of the system associated with these points are then removed before final solution. Neumann boundary conditions are in principle applied by prescribing fluxes, but the example considered here involve zero fluxes, hence no modification of the system of equations was required. The solution of the transient problem, given by \Cref{eq:heat_discrete}, can be obtained by a standard time integration scheme. However, it has been confirmed separately that the transient solutions reach the steady-state results, albeit after different time intervals, depending on the mesh and prescribed local diffusivity coefficients. For computational efficiency, the results presented hereafter are obtained by steady-state solutions. 

Following solution of the system, the flux through a domain boundary surface with Dirichlet boundary condition is computed by
\begin{equation}
  F = \sum_{b_1} \abs{f(b_1)},
\end{equation}
where the sum is taken over all $1$-cells with one interior and one surface $0$-cell. The calculation of the effective diffusivity of a material domain is analogous the the experimental determination of such a parameter. Dirichlet boundary conditions with different values, $u_0$ and $u_1>u_0$, are applied at two parallel boundary surfaces, which are at normal distance $h$ and have area $A$. With the calculated flux at either boundary, the effective diffusivity is given by
\begin{equation}
  \alpha_\text{eff} = \frac{F h}{(u_1-u_0)\,A},
\end{equation}
where $[\alpha_\text{eff}]=L^2/T$.

\subsection{Diffusion on regular and irregular meshes}

To compare the influence of geometric regularity, we consider the diffusion of a scalar quantity through a unit cube. The diffusion coefficient associated with all $1$-cells in $K$ is unity. Dirichlet boundary conditions are applied to the top surface with unit value and to the bottom surface with a zero value. On the remaining four exterior surfaces, zero flux Neumann boundary conditions are enforced. With this setup, the computed fluxes at the top and bottom surface equal in numerical values the effective diffusivity of the domain.

One example uses a regular mesh $M$ composed of $20\times20\times20=8,000$ cubic cells, leading to a cubical mesh $K$ with $68,921$ vertices and $201,720$ edges. A second example uses irregular mesh $M$ composed of $2500$ Voronoi cells, leading to a quasi-cubical mesh with $68,545$ vertices and $194,466$ edges. The $M$ meshes are shown in \Cref{fig:meshes}.

\begin{figure}[!ht]
  \centering
  \includegraphics[width=0.495\textwidth]{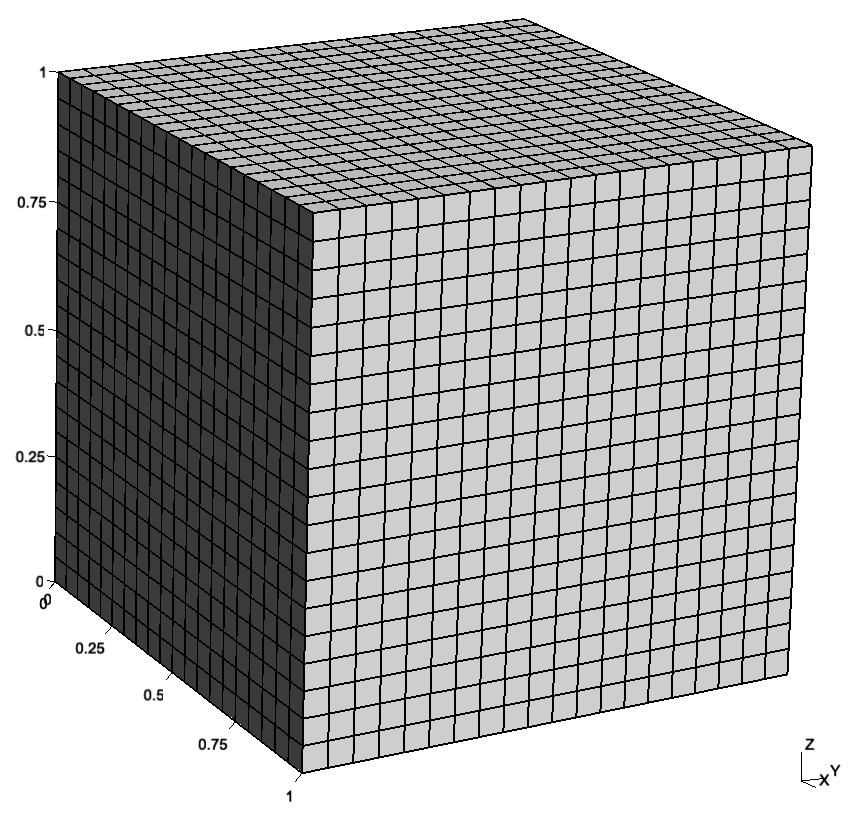}
  \includegraphics[width=0.495\textwidth]{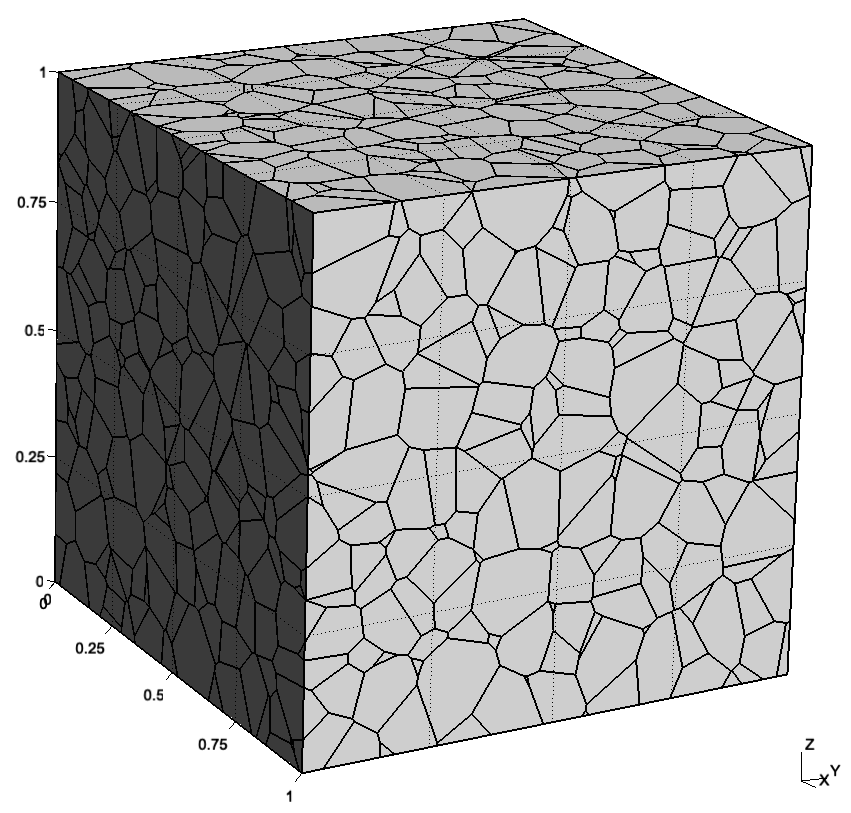}
  \caption{Regular and irregular meshes (extended complex not shown).}
  \label{fig:meshes}
\end{figure}

The variation of the scalar variable as a function of z-coordinate is shown in \Cref{fig:reg}. On the regular grid (left) the result is identical to finite differences, as well as discrete exterior calculus, and exactly reproduces the linear profile of the continuum solution. The solution values on the irregular grid (center) are scattered (right) about the continuum solution, highlighting the influence of the meshes underlying geometry. It needs to be emphasised, that it is principally incorrect to compare the analytical solution, which is based on the continuum description of diffusion, with the results from the fully intrinsic discrete formulation developed in this work. The agreement between the steady-state spatial distribution of the analysed quantity with a regular mesh with the analytical solution (and the classical finite difference scheme) only shows that a completely regular internal structure with constant local diffusivity is indistinguishable from a continuum. However, a divergence from regularity in the internal structure leads to deviations of the steady-state spatial distribution of the quantity from the continuum result, even for constant local diffusivity values - a demonstration of how structure controls behaviour.

\begin{figure}[!ht]
  \centering
  \includegraphics[width=0.32\textwidth]{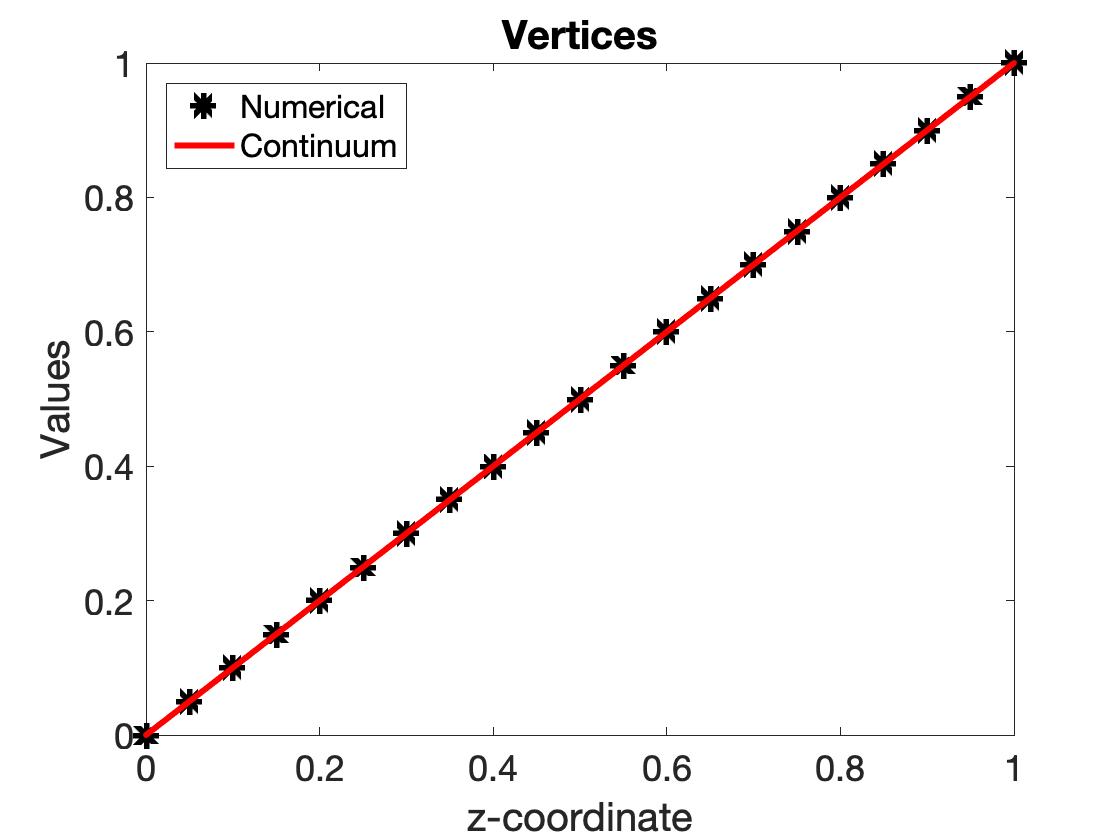}
  \includegraphics[width=0.32\textwidth]{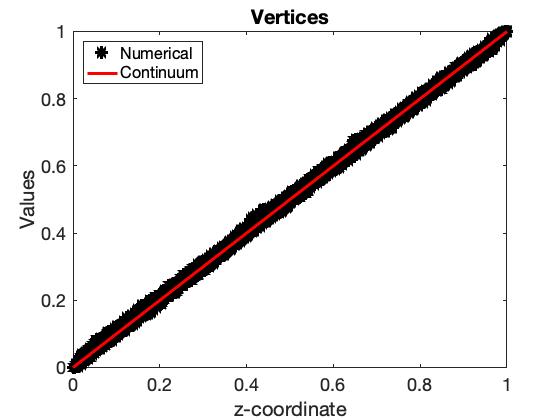}
  \includegraphics[width=0.32\textwidth]{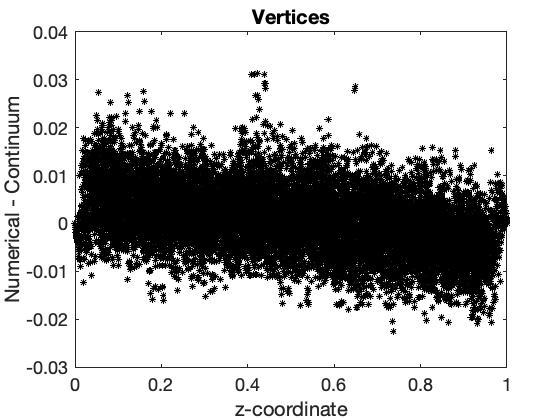}\\
  \includegraphics[width=0.32\textwidth]{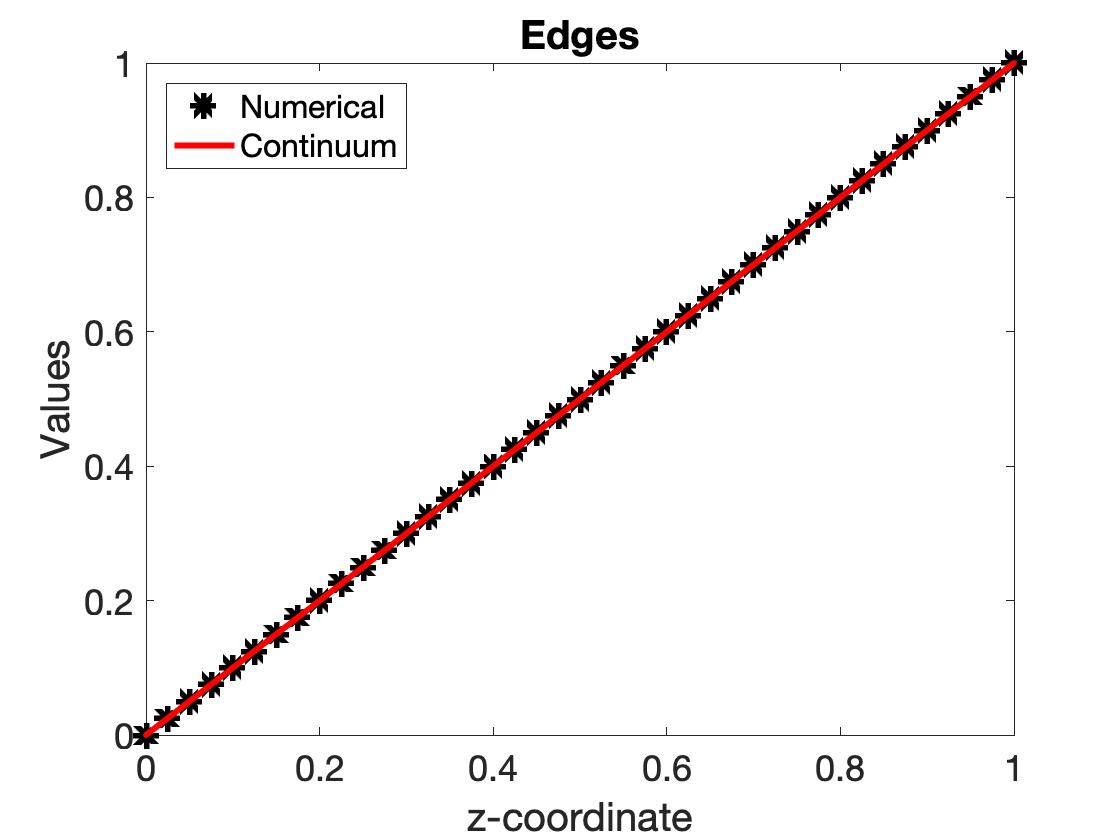}
  \includegraphics[width=0.32\textwidth]{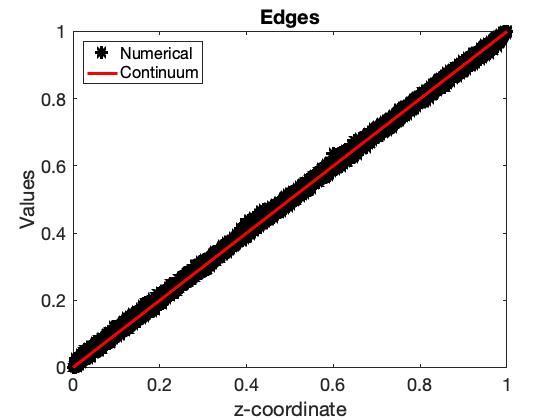}
  \includegraphics[width=0.32\textwidth]{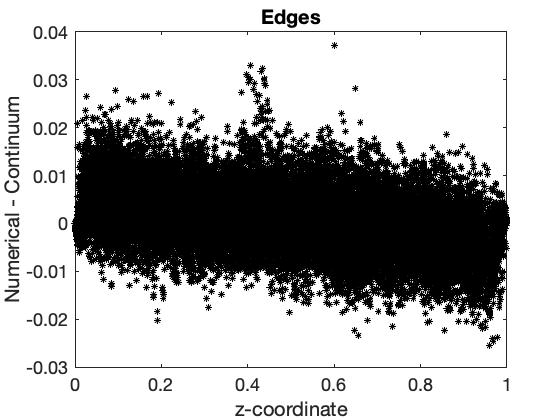}\\
  \includegraphics[width=0.32\textwidth]{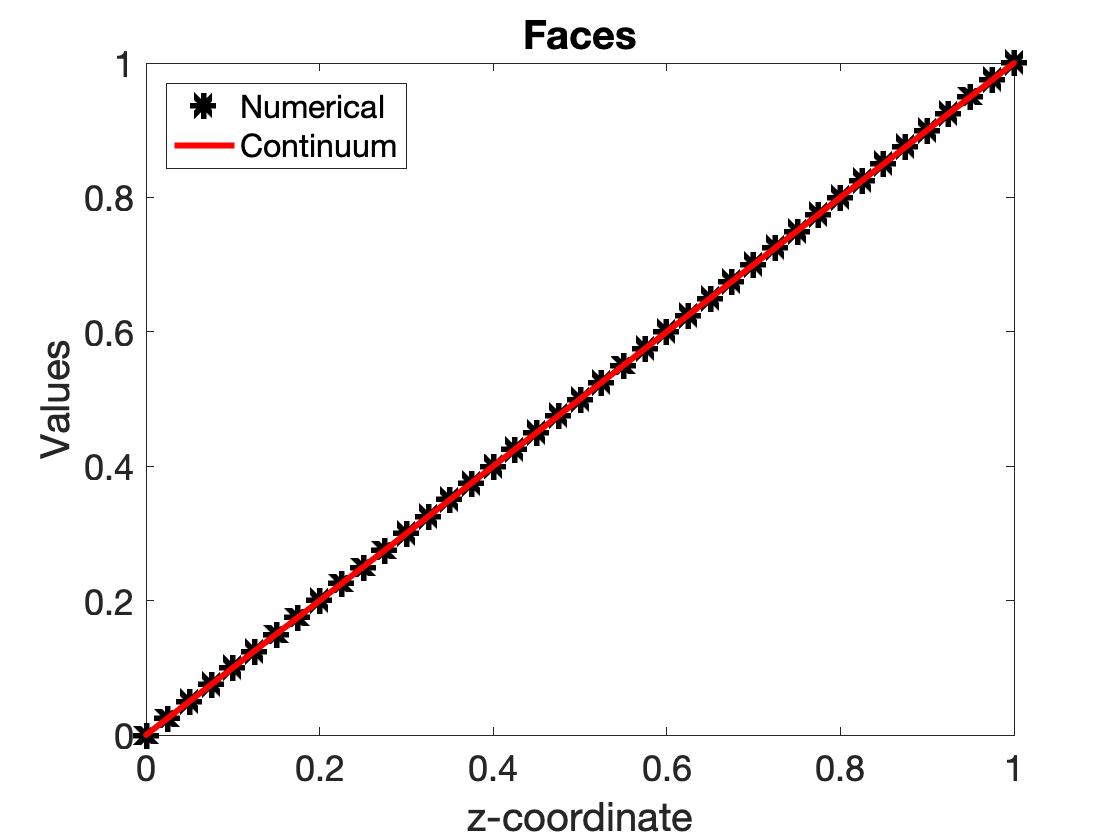}
  \includegraphics[width=0.32\textwidth]{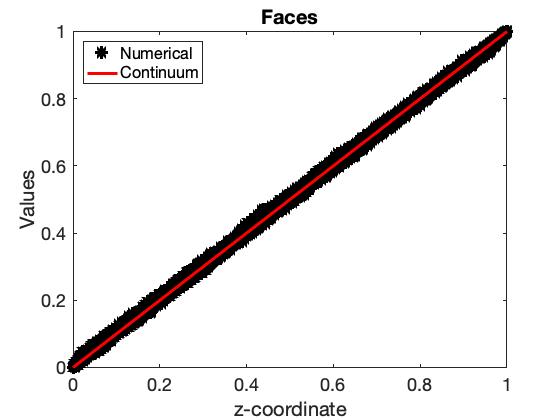}
  \includegraphics[width=0.32\textwidth]{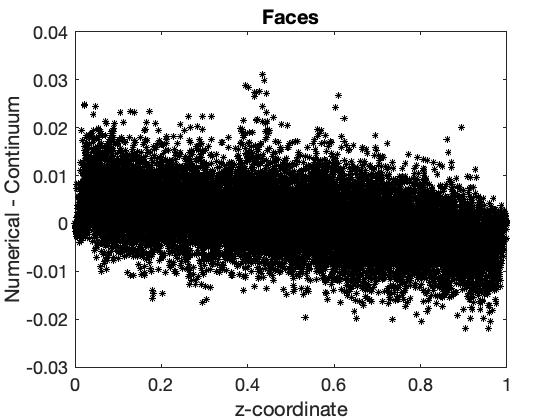}\\
  \includegraphics[width=0.32\textwidth]{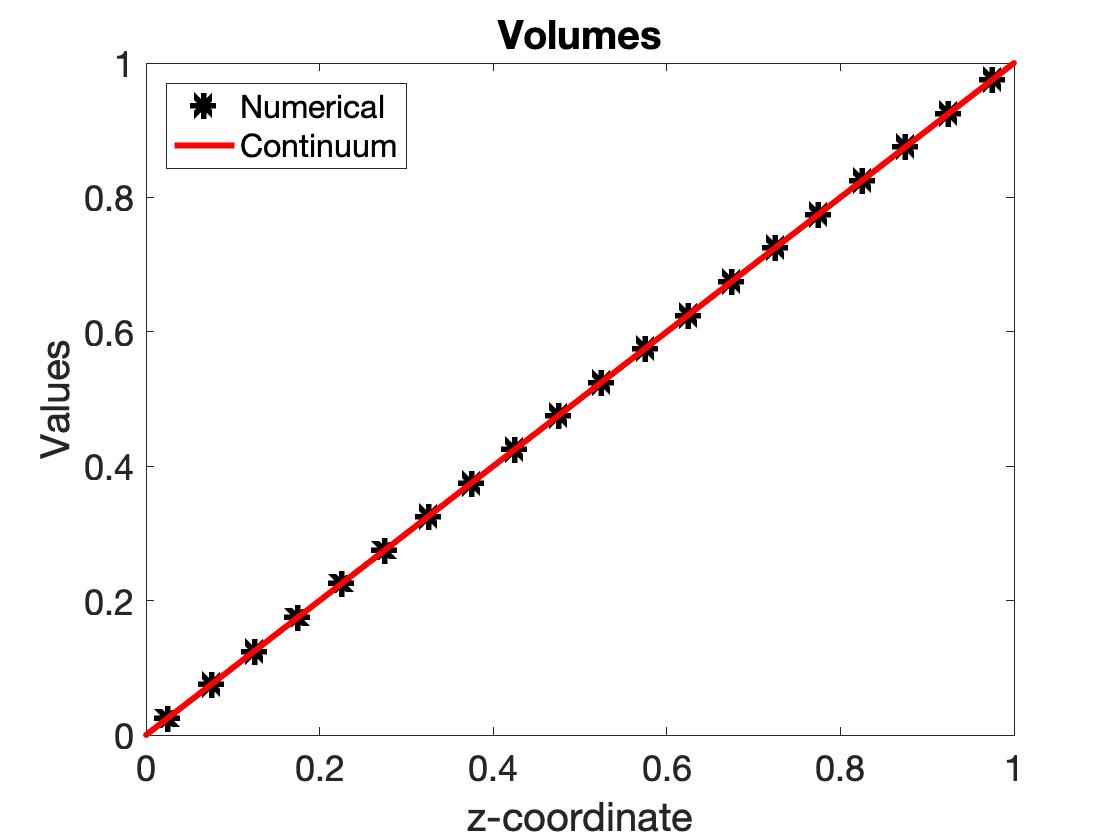}
  \includegraphics[width=0.32\textwidth]{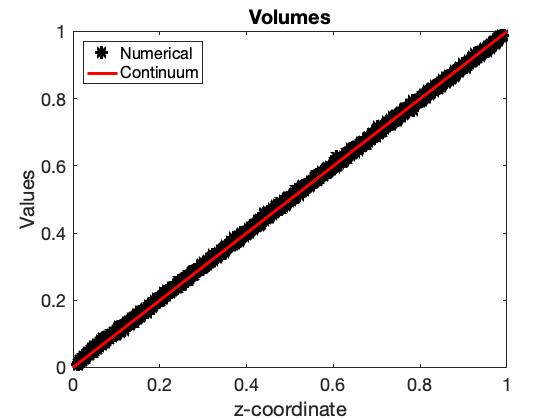}
  \includegraphics[width=0.32\textwidth]{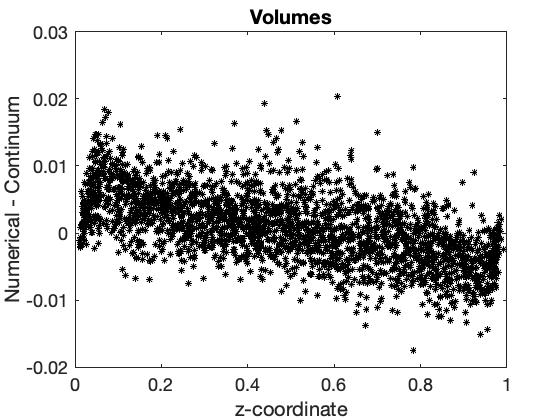}
  \caption{Value of scalar variable at vertices, edges, faces and volumes in $M$ (vertices of $K$) plotted as a function of $z$-coordinate for both regular (left) and irregular (center) meshes. Also shown (right) is the scatter of the numerical solution on the irregular mesh about the continuum solution.}
  \label{fig:reg}
\end{figure}

Further, and stronger evidence for how structures control behaviour, are the computed fluxes (effective diffusivity) for the regular and irregular meshes given in  \Cref{tab:ir_reg2}. These are also broken down by the contribution from each type of $1$-cell in $K$, representing different cells in $M$. It is clear, that while the spatial distribution of the quantity follows exactly (\Cref{fig:reg} --- left) or with small deviations (\Cref{fig:reg} --- right) the analytical result from a continuum formulation, the internal structure has a strong impact on the effective diffusivity. In the analytical solution there is no difference between local and effective diffusivity coefficients. In a material with an internal structure, different components contribute differently to the effective diffusivity, which in both cases is found larger than unity. 

\begin{table}[!ht]
    \centering
    \begin{tabular}{l| C{2.5cm} | C{2.5cm} }
                    &  \multicolumn{2}{c}{Value of flux (effective diffusivity)}\\
        \hline
                    & $z=0$ & $z=1$ \\
        \hline
        Regular     & $1.0506$ & $1.0506$ \\
         - from $c_1 \in M$ & $0.2756$ & $0.2756$ \\
         - from $c_2 \in M$ & $0.5250$ & $0.5250$ \\
         - from $c_3 \in M$ & $0.2500$ & $0.2500$ \\
        \hline
        Irregular   & $1.1604$ & $1.1604$ \\
         - from $c_1 \in M$ & $0.1959$ & $0.1959$ \\
         - from $c_2 \in M$ & $0.5238$ & $0.5238$ \\
         - from $c_3 \in M$ & $0.3752$ & $0.3752$ \\
    \end{tabular}
    \caption{Flux (effective diffusivity) of regular and irregular meshes, broken down by the contribution from each type of $1$-cell in $K$.}
    \label{tab:ir_reg2}
\end{table}

To clarify this point further, an investigation of the effect of cell size, analogous to a mesh convergence study, was also carried out, and the results are shown in \Cref{tab:ir_reg}. In both the regular and irregular meshes, the flux (effective diffusivity) is decreasing with decreasing ratio between cell and system sizes. This is an indication that the continuum formulation provides a lower limit for the diffusivity of a solid, namely a diffusivity of a structure-less solid. A different perspective is that the continuum is an approximation for materials with structures when the internal arrangements are "forgotten". Regarding the different rates of flux reduction with cell size, note that while the regular meshes do form a natural family of nested meshes, the irregular meshes are generated independently for a given number of cells in $M$.

\begin{table}[!ht]
    \centering
    \begin{tabular}{l | C{2cm} | C{2.5cm} | C{2.5cm} }
                    & Number of & \multicolumn{2}{c}{Value of flux (effective diffusivity)}\\
                    & cell in $M$& $z=0$ & $z=1$ \\
        \hline         
        Regular     & $2^3$  & $1.5625$ & $1.5625$ \\
                    & $4^3$  & $1.2656$ & $1.2656$ \\
                    & $8^3$  & $1.1289$ & $1.1289$ \\
                    & $16^3$ & $1.0635$ & $1.0635$ \\
                    & $32^3$ & $1.0315$ & $1.0315$ \\
        \hline
        Irregular   & $8$     & $1.3754$ & $1.3754$ \\
                    & $64$    & $1.2490$ & $1.2490$ \\
                    & $512$   & $1.1815$ & $1.1815$ \\
                    & $4096$  & $1.1524$ & $1.1524$ \\
                    & $32768$ & $1.1406$ & $1.1406$ \\
    \end{tabular}
    \caption{The effect of cell size on the value of flux (effective diffusivity) computed on regular and irregular meshes.}
    \label{tab:ir_reg}
\end{table}

\subsection{Effective diffusivity of composites with graphene nano-plates and carbon nano-tubes}

After clarifying the effect of different components of $M$ on the effective property of $M$, the proposed theory is applied to a practical engineering problem: electrical diffusivity of a composite with a polymer or a ceramic matrix (3D) and dispersed graphene nano-plates (GNP, 2D) or carbon nano-tubes (CNT, 1D). The matrix has a low electrical diffusivity, whereas the GNP and the CNT have substantially higher electrical diffusivity. It is expected that at a given (mass/volume) fraction of GNP or CNT, the composite will exhibit a sharp increase in effective diffusivity, as the dispersed inclusions form a percolating path across the domain. While analysis of percolation is not new in the studies of critical phenomena, the theory developed in this work allows for quantitative analysis of the macroscopic property in question - effective diffusivity.

For these simulations, the $3$-cells of $M$ contain the polymer matrix, and the $1$-cells of $K$ inside these $3$-cells are assigned a negligible electrical diffusivity of $10^{-10}$. For a composite with GNP, the nano-plates reside on select $2$-cells of $M$, and the $1$-cells of $K$ inside these $2$-cells and along their boundary $1$-cells are assigned an electrical diffusivity of $1$. For a composite with CNT, the nano-tubes reside on select $1$-cells of $M$, and the $1$-cells of $K$ along these $1$-cells are assigned an electrical diffusivity of $1$.

Simulations are performed with the same irregular mesh from the previous case. The faces and edges representing the GNP and CNT, respectively, are selected from a random uniform distribution of integers without replacement. Selected faces or edges are added successively from $0\%$ to $100\%$ of the total number. $200$ Monte Carlo paths are run to obtain a statistical average flux (effective diffusivity) for the composite as a function of percent number or volume/mass fraction.

It should be note here that introducing diffusion coefficients varying by 10 orders of magnitude, coupled with cell volumes, face areas and edge lengths varying by 2, 9 and 5 orders of magnitude, respectively, can lead to a system that is poorly conditioned. The degree to which the system is poorly conditions depends on the specific combination of faces and edges are assigned the higher diffusion coefficient. The result of this poor conditioning is visible as spikes in the data plotted in this section. Ongoing work includes developing strategies to mitigate this numerical stiffness.

The evolution of average effective diffusivity with dispersed GNP is shown in \Cref{fig:GNP} as functions of the fraction of 2-cells covered by GNP and of total GNP area. The expected step change in effective electrical diffusivity of the composite is recovered. The fraction of 2-cells covered by GNP at the step change, referred to as the percolation threshold, is consistent with results presented in \cite{borodin2021nanoceramics} for nano-composites with ceramic matrix and graphene-oxide inclusions. Experimentally, the percolation threshold has been determined in terms of the inclusions' volume fraction, $\phi$, and has varied between $0.38\%$ and $7.3\%$. \cite{nistal2018nanoplates}. The particular value determined in \cite{nistal2018nanoplates} is $\phi=1.18\,\%$. For an average plate/inclusion thickness $t$ (dimension $L$), and a given area covered by GNP at percolation $A_c$ (non-dimensional), a simple calculation shows that a percolation threshold $\phi$ re-scales the unit cube to edge length $L_c=A_c \cdot t / \phi$ (dimension $L$). This gives an average 3-cell size $d=L_c / 2500^{1/3}$, which is also approximately equal to the average in-plane diameter of GNP. Considering $A_c \approx 8$ units from \Cref{fig:GNP}, $d \approx 0.589\, t / \phi$. Assuming an average plate thickness $t=3\,nm$ \cite{borodin2021nanoceramics} this gives the following GNP diameters: $d=465\,nm$ for $\phi=0.38\,\%$; $d=150\,nm$ for $\phi=1.18\,\%$; and $d=24\,nm$ for $\phi=7.3\,\%$. The result suggests an explanation for the different volume fraction thresholds reported in the literature: the sizes of GNP used in different experiments were different. As observed previously \cite{borodin2021nanoceramics}, the effective conductivity can be increased by larger plates with lower volume fraction.

\begin{figure}[!ht]
  \centering
  \includegraphics[width=0.495\textwidth]{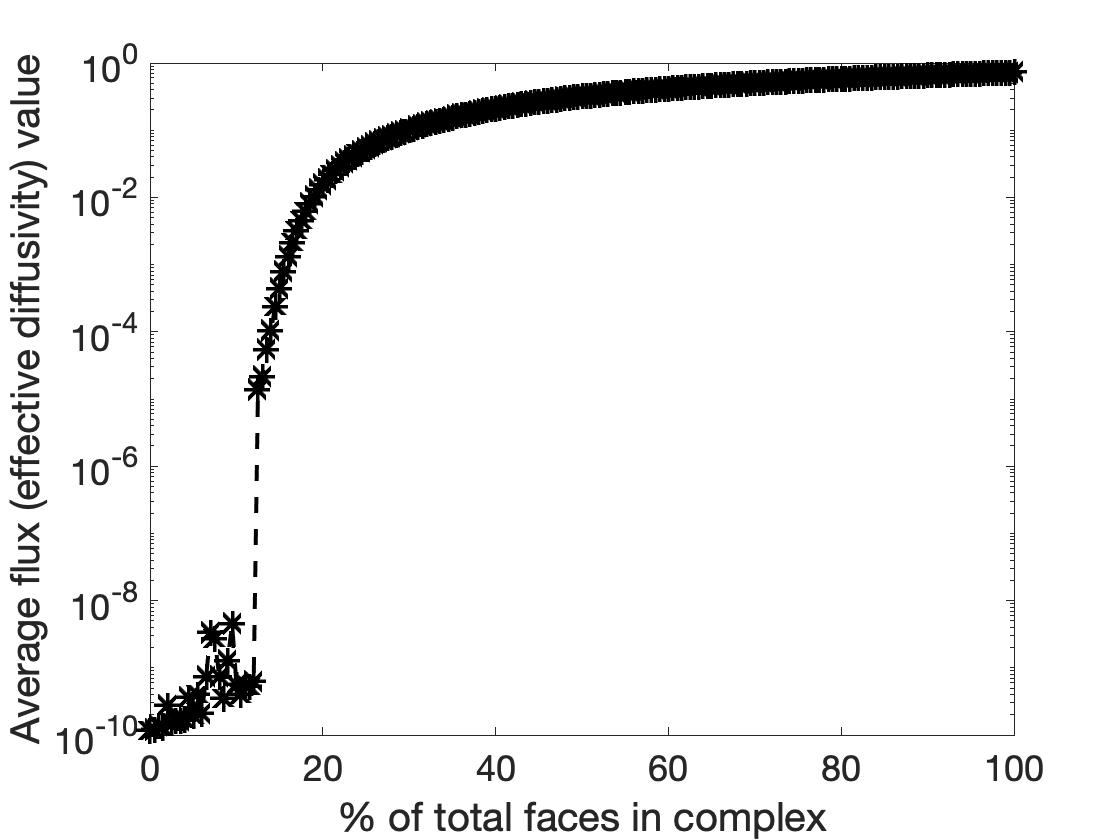} 
  \includegraphics[width=0.495\textwidth]{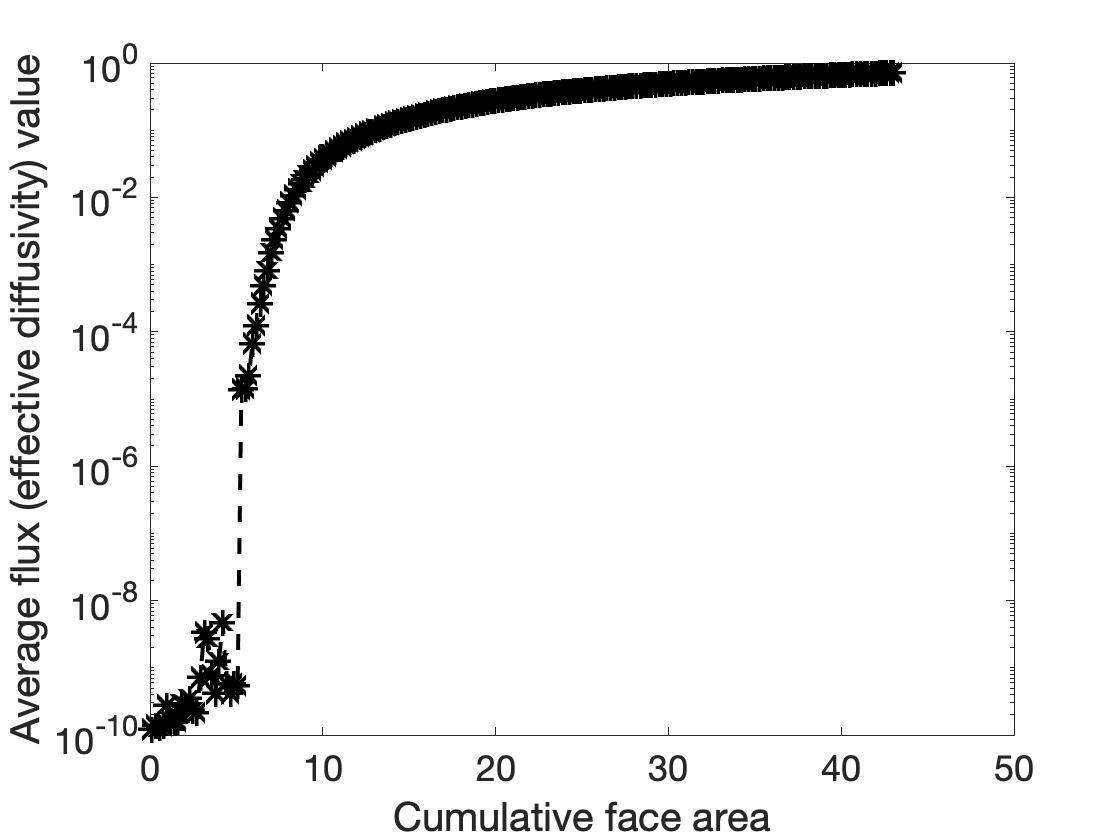}
  \caption{Effective diffusivity of GNP composite as a function of: fraction of faces covered by GNP (left); and cumulative area of GNPs (right). }
  \label{fig:GNP}
\end{figure}

The evolution of the average effective diffusivity of the composite with dispersed CNT is shown in \Cref{fig:GNT}. The expected step change is still clearly visible, but occurs at a much higher fraction of 1-cells covered by CNT. Experimental data for polymer composite with carbon nano-tubes \cite{micaela2016nanotubes} shows that the percolation threshold is around 0.5 wt\%. Considering the polymer density $\rho_p=1.07\,g/m^3$ and the CNT density $\rho_c=2.16\,g/m^3$ reported in \cite{micaela2016nanotubes}, this translates into a CNT volume fraction at percolation $\phi=1\,\%$. For an average nano-tube radius $r$ (dimension $L$), and a given length of $1$-cells occupied by CNT at percolation $L_c$ (non-dimensional), a calculation similar to the last one shows that a percolation threshold $\phi$ re-scales the unit cube to face area $A_c=L_c \cdot \pi \cdot r^2 / \phi$ (dimension $L^2$). This gives an average 3-cell size $d=A^{1/2} / 2500^{1/3}$, which is also approximately equal to the average CNT length. Considering $L_c \approx 450$ from \Cref{fig:GNT}, $d \approx 2.77\, r / \phi^{1/2}$. For the measured CNT volume fraction at percolation $\phi=1\,\%$, the relation between length and radius of CNT is $d \approx 27.7\, r$, suggesting that the experiment was performed with relatively short CNT. Generally, the diameter of CNT can vary from less than $1\,nm$ to over $100\,nm$, and the length can vary from several nano-metres to several centimetres. Taking example from the case with GNP, it can be suggested that the effective diffusivity can be increased by longer nano-tubes with lower volume fraction.

\begin{figure}[!ht]
  \centering
  \includegraphics[width=0.495\textwidth]{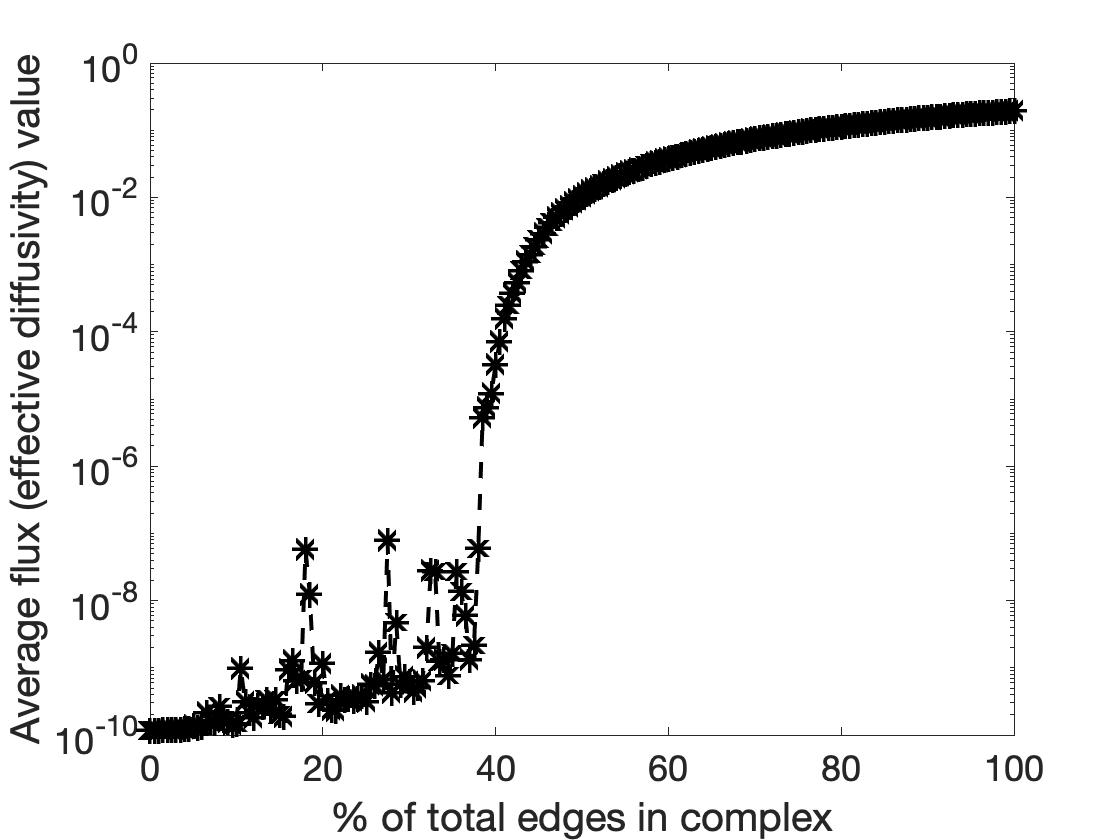}
  \includegraphics[width=0.495\textwidth]{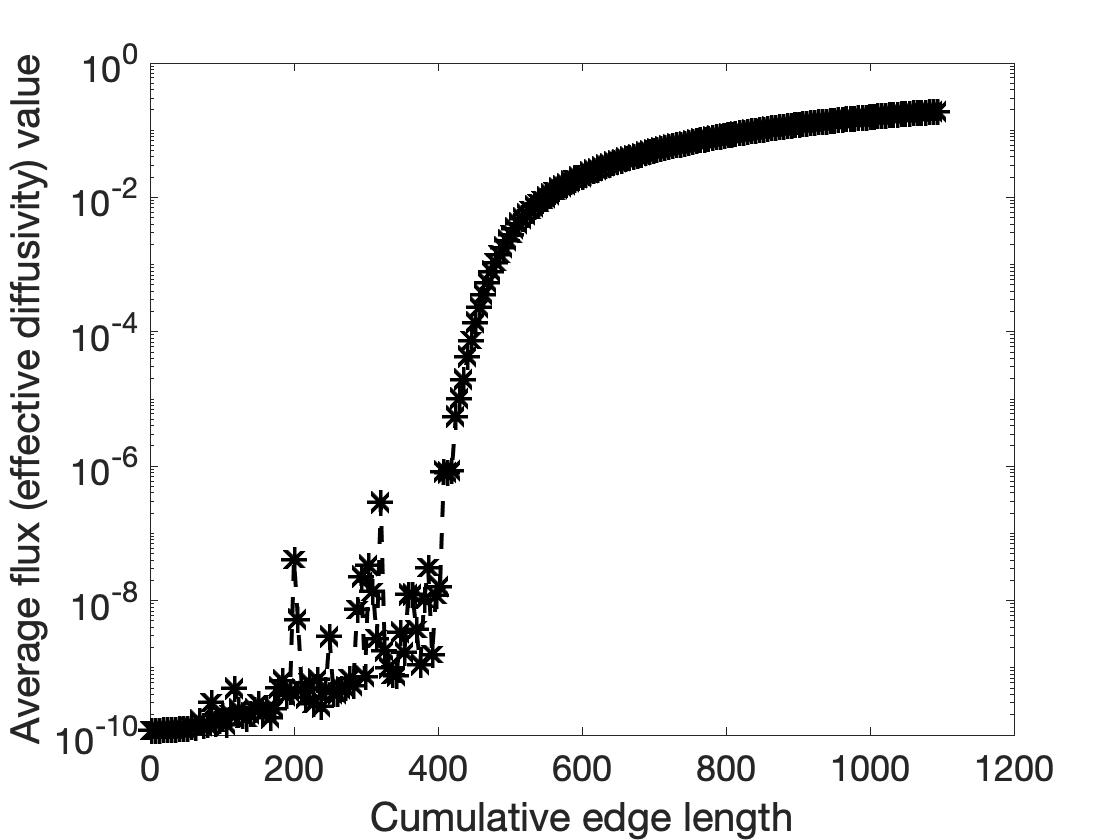}
  \caption{Effective diffusivity of CNT composite as a function of: fraction of edges covered by CNT (left); and cumulative length of CNTs (right). }
  \label{fig:GNT}
\end{figure}

\section{Conclusions}
\label{sec:conclusions}
The paper presented a geometric development of Forman's combinatorial formulations of discrete differential forms and exterior derivatives of such forms on discrete topological spaces. The key steps in this development are the introduction of a metric tensor as a proper bi-linear map from $p$-cochains to $0$-cochains and the use of discrete Riemann integration to define inner product in the space of $p$-cochains. The resulting adjoint coboundary operator, Laplacian, and Hodge-star operator, are constructed canonically for a given discrete space, and do not rely on a construction of a dual space as in the existing discrete exterior calculus. The complete theory is applied to analysis of physical processes dependent on a scalar variable operating in materials with internal structures. Importantly, the theory allows for processes to operate at different rates in volumes, surfaces, and lines, allowing for modelling complex internal structures and their effects on the macroscopic/effective material properties. Examples with diffusivity of composites are shown to illustrate how the theory can be applied to quantify the effect of inclusions on the macroscopic composite behavior. The work forms the first important step towards an intrinsic formulation of vector problems (mechanical deformation) on discrete topological spaces.

\appendix

\section{Polytopes and meshes}
\label{sec:polytope_mesh}
In this section we define and give examples of polytopes and meshes. This section is concerned primarily with the topology and the geometry of a mesh.

\subsection{Polytopes}
\label{sec:polytopes}
Polytopes generalise polygons and polyhedrons in any dimension. They are the building blocks of meshes and discrete geometry.
\begin{definition}
  Let $A$ be an affine space, $B \subseteq A$, $B$ is a $d$-manifold with boundary. We say that $B$ is \textbf{flat}, if there exists a $d$-dimensional subspace $A' \subseteq A$ such that $B \subseteq A'$.
\end{definition}
\begin{definition}
\label{thm:polytope}
  Let $A$ be an affine space, $d \in \N$. A $d$-dimensional \textbf{polytope} (or $d$-polytope) in $A$ is a closed subset of $A$ defined recursively as follows.
  \begin{enumerate}
    \item
      A $0$-polytope is a singleton subset of $A$ (i.e., a single point).
    \item
      For $d > 0,$ $X \subseteq A$ is $d$-polytope in $A$ if $X$ is homeomorphic to a closed $d$-ball, $\dim {\Aff}(X) = d$ and there exist $k \in N$ and $(d - 1)$-polytopes $Y_{1}, ..., Y_{k}$ in $A$ such that $\partial X = Y_{1} \cup ... \cup Y_{k}.$
  \end{enumerate}
\end{definition}
\begin{definition}
  Two $d$-polytopes $c_d$ and $c'_d$ are called \textbf{combinatorially equivalent} if there exists a bijection $f$ between the sets of their faces (the meshes defined by that polytopes, see \Cref{thm:mesh_from_polytope}) such that if $a_p$ is a face of $b_q$, then $f(a_p)$ is a face of $f(b_q)$.
\end{definition}
\begin{remark}
  $0$-polytopes are one-element (singleton) sets,
  $1$-polytopes are \textbf{line segments},
  $2$-polytopes are \textbf{polygons} (convex or not),
  $3$-polytopes are \textbf{polyhedrons} (convex or not).
\end{remark}
\begin{example}
  General classes of polytopes (defined for any dimension $d$) include:
  \begin{itemize}
    \item
      \textbf{simplices} (generalisation of triangles and tetrahedrons) - they are discussed in \ref{sec:simplex};
    \item
      \textbf{parallelotopes} (generalisation of parallelograms and parallelepipeds with (hyper-)cubes being a particular case) - they are discussed in \ref{sec:parallelotope};
    \item
      \textbf{quasi-cubes} (polytopes, combinatorially equivalent to (hyper-)cubes) - they are the basis of the mathematical formalism in this paper (\Cref{sec:topology} and \Cref{sec:geometry}).
  \end{itemize}
\end{example}
\begin{theorem}[\textbf{Diamond property}]
\label{thm:polytope_diamond_property}
  Let $d \geq 2$, $c_d$ be a $d$-polytope. Then for any $a_{d - 2} \prec c_d$ there exist exactly two $(d - 1)$-cells $b_{d - 1}$ and $b'_{d - 1}$ between them, i.e., $a_{d - 2} \prec b_{d - 1} \prec c_d$ and $a_{d - 2} \prec b'_{d - 1} \prec c_d$.
\end{theorem}
\begin{proof}
  See \cite[Theorem 2.7 (iii)]{ziegler1995lectures} (for convex polytopes).
\end{proof}

\subsection{Meshes}
\label{sec:meshes}
This subsection is devoted to polytopal meshes, i.e., collections of ``glued'' polytopes. (The term ``complex'' is also used in the literature, but we prefer the term ``mesh'' because we use ``complex'' for the chain and cochain complexes induced from an orientation on a mesh, as discussed in \ref{sec:complexes}.) We consider general meshes and manifold-like meshes. The latter are more restrictive which allow more specific notions to be defined on them, like compatible orientations.
\begin{definition}
\label{thm:mesh}
  A (nonempty) $d$-dimensional \textbf{polytopal mesh} (or simply a $d$-\textbf{mesh}) is a finite (possibly empty) set $M$ of polytopes (the \textbf{cells} of $M$) such that:
  \begin{enumerate}
    \item
      for any cell $b_p \in M$, if $c_q$ is a face of $b_p$, then $c_q \in M$;
    \item
      For any two cells $X, Y \in M,$ there exists $k \in \N$ (possibly zero) and faces $Z',...,Z^{(k)} \in M$ of both $X$ and $Y$, such that $X \cap Y = Z' \cup ... \cup Z^{(k)}$.
  \end{enumerate}
\end{definition}
\begin{remark}
  When the mesh consist of convex polytopes, the second condition becomes stronger - the intersection of two polytopes is either the empty set or another cell of the mesh, i.e., $k \in \{0 , 1\}$ in the above definition.
\end{remark}
\begin{definition}
  Let $M$ be a non-empty mesh. The maximal dimension $d$ of the cells of the mesh is called the \textbf{topological dimension} of the mesh. In this case we say that $M$ is a $d$-mesh.
  
  If $M$ is a $d$-mesh, then $M = \cup_{p = 0}^{d} M_p$, where $M_p$ is the set of $p$-cells in $M$.
\end{definition}
\begin{example}
\label{thm:mesh_from_polytope}
  Any $d$-polytope $c_d$ induces a $d$-mesh: the $p$-cells, $0 \leq p \leq d$ are the faces of $c_d$ ($c_d$ being the unique $d$-cell).
\end{example}
\begin{remark}
  We use the following standard names for low-dimensional cells: \textbf{nodes} (N) for $0$-cells; \textbf{edges} (E) for $1$-cells; \textbf{faces} (F) for $2$-cells; \textbf{volumes} (V) for $3$-cells.
  
  For example, in a cube $M$, considered as a $3$-mesh, its cells can be written as:
  
  $M_0 = \{N_1, ..., N_8\},\
   M_1 = \{E_1, ..., E_{12}\},\
   M_2 = \{F_1, ..., F_6\},\
   M_3 = \{V_1\}$.
\end{remark}
\begin{definition}
  We say that a mesh is \textbf{combinatorially regular} if all top-dimensional polytopes are combinatorially equivalent. Otherwise, we say it is \textbf{combinatorially irregular}, i.e., there exist two top-dimensional polytopes which are not combinatorially equivalent.
\end{definition}
\begin{example}
  Classes of regular meshes include \textbf{simplicial}, \textbf{cubical}, \textbf{quasi-cubical}. Irregular meshes used in this work include \textbf{Voronoi diagrams}.
\end{example}
In the next lines we are going to define what a manifold-like mesh is (for a simplicial mesh also known as a simplicial manifold). An obvious definition is to require that the underlying set $\abs{M}$ of $M$ (the union of all cells in $M$) is a $C^0$-manifold with boundary \cite[Definition 2.3.9]{Hirani2003_DEC}. However, we prefer a more intrinsic point of view, following \cite[Defnition 7.21]{chen2014digital}
\begin{definition}
  Let $d \geq 2$ $M$ be a nonempty $d$-mesh, $p \leq d - 2$, $a_p \in M_p$, $c_{p + 2}'$ and $c_{p + 2}''$ are $(p + 2)$-superfaces of $a_p$. We say that $c_{p + 2}'$ and $c_{p + 2}''$ are \textbf{$(p + 1)$-connected around $a_p$} if there exist a sequence $c_{p + 2}' = b_{p + 2}^{(1)}, ..., b_{p + 2}^{(k)} = c_{p + 2}''$ of $(p + 2)$-superfaces of $a_p$, such that any two adjacent cells have a common $(p + 1)$-cell (which is a $(p + 1)$-superface of $a_p$ as well).
\end{definition}
\begin{definition}
  Let $d \geq 2$ $M$ be a nonempty $d$-mesh, $p \leq d - 2$. $M$ is said to be \textbf{$p$-regular}, if for each $a_p \in M_p$ any two of its $(p + 2)$-superfaces are $(p + 1)$-connected around $a_p$.
\end{definition}
\begin{remark}
  A mesh consisting of two polygons sharing only a common node (e.g., a bow tie ) is not $0$-regular because the two polygons are not $0$-connected around that node - no edge connects them.
\end{remark}
\begin{definition}
\label{thm:mesh_manifold_like}
  Let $M$ be a mesh. $M$ is called \textbf{manifold-like} mesh if it is empty or it is a $d$-mesh satisfying:
  \begin{enumerate}
    \item
      for any $p < d$ and any $p$-cell $b_p \in M_p$ there exists $c_d \in M_d$ such that $b_p \preceq c_d$;
    \item
      any $(d - 1)$-cell in $M$ has at most two $d$-superfaces;
    \item
      $M$ is $p$-regular for any $p \in \{0, ..., d - 2\}$.
  \end{enumerate}
\end{definition}
\begin{definition}
  Let $d > 0$ and $M$ be a manifold-like $d$-mesh. A cell $c_{d - 1}$ is called:
  \begin{itemize}
    \item
      \textbf{boundary}, if it has one $d$-superface;
    \item
      \textbf{interior}, if it has two $d$-superfaces.
  \end{itemize}
\end{definition}
\begin{definition}
\label{thm:topological_boundary}
  Let $M$ be a manifold-like mesh. If $M$ is a $d$-mesh, $d > 0$, then the set $\partial M$ of all boundary $(d - 1)$-faces and their faces is called the \textbf{boundary of $M$}. If $M$ is a $0$-mesh or $M = \emptyset$, then we set $\partial M = \emptyset$.
\end{definition}
\begin{definition}
  Let $M$ be a a manifold-like mesh. We say that $M$ is a \textbf{closed mesh} if it has no boundary, i.e., if $\partial M = \emptyset$.
\end{definition}
\begin{claim}
  Let $M$ be a manifold-like mesh. Then $\partial M$ is a closed mesh, i.e., $\partial(\partial M) = \emptyset$.
\end{claim}
\begin{proof}
  See \cite[Theorem 7.2]{chen2014digital}.
\end{proof}

\section{Orientation}
\label{sec:orientation}
This section is devoted to the notion of orientation (on a vector space, affine space, polytope and mesh). Induced orientations on the boundary are also considered together with the boundary and coboundary operators on a mesh.

\subsection{Orientation on a vector space}
\label{sec:orientation_space}
\begin{definition}
\label{thm:orientation_vector_space_bases}
  Let $V$ be a finite dimensional real vector space. An \textbf{orientation class} is a set of ordered bases of $V$ such that the change of basis matrix between them has positive determinant (homotopic to the identity matrix). There are exactly two orientation classes on $V$. An \textbf{orientation on} $V$ is a choice of orientation class $\omega.$ The pair $(V, \omega)$ is called \textbf{oriented vector space}.
\end{definition}
\begin{example}
  The standard orientation of $\R^d$ is given by the orientation class of its ordered standard basis $((1, 0, ..., 0), ..., (0, ..., 0, 1))$.
\end{example}
\begin{remark}
  There is one minor problem with \Cref{thm:orientation_vector_space_bases} - it cannot handle zero-dimensional spaces as they have exactly one basis - the empty space. A better approach to is to use top-dimensional elements of the \textbf{exterior algebra} $(\Lambda V, \wedge)$. Using $\Lambda V$ allows is also beneficial when manipulations with orientations are needed (e.g., \Cref{thm:orientation_relative}, \Cref{eq:cup_product}).
\end{remark}
\begin{definition}
\label{thm:orientation_vector_space_forms}
  Let $V$ be a $d$-dimensional real vector space. We say that two nonzero $d$-vectors $\omega$ and $\eta$ are in the same (opposite) orientation class if the unique $\lambda \in \R \setminus \{0\}$, such that $\omega = \lambda \eta$, is positive (negative). Clearly, this again gives two different orientations.
\end{definition}
\begin{remark}
  For $d > 0$ there is a clear bijection between the two definitions: to an orientation class (in the sense of \Cref{thm:orientation_vector_space_bases}) represented by the ordered basis $(e_1, ..., e_d)$ of $V$ assign the orientation class (in the sense of \Cref{thm:orientation_vector_space_forms}) represented by the nonzero $d$-vector $e_1 \wedge ... \wedge e_d$ and vice-versa. The correctness of this map is guaranteed by the fact that if $(f_1, ..., f_d)$ is another basis of $V$ and $A$ is the change of basis matrix, then $f_1 \wedge ... \wedge f_d = (\det A) e_1 \wedge ... \wedge e_d.$
  
  However, if $d = 0$, then $\Lambda^0 V = \R$ and the two orientation classes are given by positive and negative numbers respectively. As was said, this is not possible with the definition using ordered bases.
\end{remark}
\begin{definition}
  Define an an action of the multiplicative group $\{1, -1\}$ on orientation classes: for $\varepsilon \in \{-1, 1\}$, let $\varepsilon [\omega] := [\varepsilon \omega]$. This action is transitive and free which allows to define a division
  \begin{equation}
    \frac{[\omega]}{[\eta]} := 
      \begin{cases}
        1, & [\omega] = [\eta] \\
        -1 & [\omega] = - [\eta]
      \end{cases}.
  \end{equation}
\end{definition}
\begin{remark} 
  Let $V$ be a finite dimensional vector space and $U$ be a subspace of $V.$ Then $\Lambda^p U$ is naturally embedded in $\Lambda^p V$ as $\Lambda^p U$ can be thought as all $p$-vectors formed by some arbitrary basis of $U$. Hence, we will identify $\Lambda^p U$ as this subspace of $\Lambda^p V.$
\end{remark}
\begin{definition}
  \textbf{Orientation on an affine space} $A$ is an orientation on the associated vector space $V$. \textbf{Orientation on a polytope} $\gamma_d$ is an orientation on $A$ (i.e., orientation on $V$).
\end{definition}
\begin{claim}
  Let $(A,V)$ be a $d$-dimensional affine space, $(B,W)$ a hyperplane of $(A,V)$, $A'$ and $A''$ be the two half-spaces determined by $B$, $\omega$ be a nonzero $(d - 1)$-vector on $W.$ Let $b \in B$ be arbitrary. Then for all $a' \in A'$ and $a'' \in A''$ the $n$-vectors $(a' - b) \wedge \omega$ and $(a'' - b) \wedge \omega$ induce opposite orientations of $V$.
\end{claim}
In other words, having an oriented affine hyperplane $B$ of $A$, a choice of one of two sides of $B$ determines an orientation of $A$. Inversely, if $A$ is oriented, a choice of one of the two sides $B$ determines an orientation on $B$.

\subsection{Orientation on polytopes and meshes}
\label{sec:orientation_mesh}
Except in \ref{sec:simplex} and \ref{sec:parallelotope} we abuse the notation and include the orientation of a polytope in its symbolic representation (although both of its orientations are of equal significance). In other words the symbol $c_p$ represents an oriented $p$-polytope and its orientation is denoted by $\OR(c_p)$ (used in \ref{sec:simplex} and \ref{sec:parallelotope} as well).
\begin{definition}
\label{thm:orientation_relative}
  Let $(A, V)$ be a $d$-dimensional affine space, $c_d$ be an oriented polytope in $A$, $b_{d -1}$ be an oriented hyperface of $c_d$. Then $\Aff(b_{d - 1})$ divides $A$ into two sides. Locally around each interior point of $b_{d - 1}$ the points on the one side live in the interior of $c_d$ while the points of the other side are at the exterior of $c_d$. Take an outward-pointing vector ${\bf n}$. The \textbf{relative orientation} between $c_d$ and $b_{d -1}$ is defined by
  \begin{equation}
  \label{eq:orientation_relative}
    \epsilon(c_d, b_{d -1}) := \frac{[{\bf n}] \wedge \OR(b_{d - 1})}{\OR(c_d)}.
  \end{equation}
\end{definition}
\begin{remark}
  If $\epsilon(c_d, b_{d -1}) = 1$ in \Cref{thm:orientation_relative}, we say that $b_{d - 1}$ has the \textbf{induced orientation}.
\end{remark}
\begin{remark}
  If $d = 1$, $E$ is a line segment with endpoints $N_1$ and $N_2$, oriented positively, then if $E$ ``points'' from $N_1$ to $N_2$, then $\varepsilon(E, N_1) = -1,\ \varepsilon(E, N_2) = 1$.
  
  If $d = 2$, $F$ is a polygon in $\R^2$ with orientation $[e_1 \wedge e_2]$ (where $e_1 = (1,0),\ e_2 = (0,1)$) and $E$ is a boundary edge of $F$, then $\varepsilon(F, E)$ is $1$ if $E$ follows the counter-clockwise orientation on the boundary, and $-1$ otherwise.
  
  Example (with the full boundary operator) is given in \Cref{thm:bd_cbd_example}.
\end{remark}
\begin{claim}
\label{thm:orientation_main_property}
  Let $c_d$ be a $d$-polytope, $a_{d - 2}$ be a $(d - 2)$-face of $c_d$, $b_{d - 1, 1}$ and $b_{d - 1, 2}$ be the two $(d - 1)$-cells between $c_d$ and $a_{d - 2}$. Then, whatever the orientations of these $4$ polytopes are, 
  \begin{equation}
    \varepsilon(c_d, b_{d - 1, 1}) \varepsilon(b_{d - 1}, a_{d - 2})
    + \varepsilon(c_d, b_{d - 1, 2}) \varepsilon(b_{d - 1, 2}, a_{d - 2})
      = 0.
  \end{equation}
\end{claim}
\begin{proof}
  We will prove the claim for a convex polytope.
  Let $\omega_{d - 2} \in \Lambda^{d - 2}$ be a representative of the orientation class of $a_{d - 2}$. Let $e_1$ and $e_2$ be vectors parallel to $b_{d - 1, 1}$ and $b_{d - 1, 2}$ respectively and which are not parallel to $a_{d - 2}$. The orientation classes are then represented by $s_1 e_1 \wedge \omega_{d - 2}$ and $s_2 e_2 \wedge \omega_{d - 2}$ respectively, $s_1, s_2 \in \{-1, 1\}$. Then $s_0 e_1 \wedge e_2 \wedge \omega_{d - 2}$ is a representative for the orientation of $c_d$, $s_0 \in \{-1 , 1\}$.
  For $i = 1, 2$ the vector $-e_i$ is outward-pointing to both $a_{d - 2}$ with respect to $b_{d - 1, i}$ and $b_{d - 1, 2 - i}$ with respect to $c_d$. It is easy to see that: $\varepsilon(b_{d - 1, i}, a_{d - 2}) = s_i\ (i = 1, 2)$; $\varepsilon(c_d, b_{d - 1, 1}) = -s_0 s_2$; $\varepsilon(c_d, b_{d - 1, 1}) = s_0 s_1$. Then
  \begin{equation*}
    \varepsilon(c_d, b_{d - 1, 1}) \varepsilon(b_{d - 1}, a_{d - 2})
    + \varepsilon(c_d, b_{d - 1, 2}) \varepsilon(b_{d - 1, 2}, a_{d - 2})
    = -s_0 s_2 s_1 + s_0 s_1 s_2 = 0. \qedhere
  \end{equation*}
\end{proof}
\begin{definition}
\label{thm:orientation_compatible}
  Let $M$ be a manifold-like $d$-mesh. A \textbf{compatible orientation} on $M$ is an orientation on $M$ such that the for any interior $a_{d - 1} \in M_{d - 1}$, if $b_d$ and $c_d$ are its $d$-superfaces, then $\epsilon(b_d, a_{d - 1}) = - \epsilon(c_d, a_{d - 1})$.
\end{definition}
\begin{remark}
  In \Cref{thm:orientation_compatible} the orientation of an interior $(d - 1)$-cell is arbitrary. However, it is convenient to assume that for any boundary $a_{d - 1} \in M_{d - 1}$, if $b_d$ is its $d$-superface, then $\epsilon(b_d, a_{d - 1}) = 1$. This is important in \Cref{thm:fundamental_class_boundary}.
\end{remark}
\begin{definition}
  A \textbf{compatibly orientable mesh} is a manifold-like mesh for which a compatible orientation exists.
\end{definition}
\begin{remark}
  If a $d$-mesh is compatibly oriented, there are two compatible orientations on its $d$-cells which are opposite.
\end{remark}
\begin{remark}
  A manifold-like mesh does not have a compatible orientation if it is a discretisation of a non-orientable manifold (e.g., M\"{o}bius strip).
\end{remark}
In the next two paragraphs we will derive some basic results concerning orientations which are crucial for the definition of the boundary operator in \ref{sec:complexes}. In the literature the expression for the boundary operator is usually given for simplices as a definition (as in \cite[Section 2.1]{hatcher2002topology}, where the author gives only some intuition behind the definition). However, we will show that it can be actually derived (and thus rigorously clarify the choice of signs) by the more general approach considered here.

\subsubsection{Orientation on simplices and simplicial meshes}
\label{sec:simplex}
For a simplex (and therefore for a simplicial mesh) there is a standard way to define its orientation based on the order of nodes. Moreover, the relative orientation is calculated quite easily.
\begin{definition}
  Let $x_0, ..., x_d \in A$ be affinely independent points. The \textbf{simplex} with vertices $x_0, ..., x_d$ is the region
  \begin{equation*}
    S(x_0, ..., x_d) := \Con(x_0, ..., x_d).
  \end{equation*}
  The definition of a simplex does dot depend on the order of the points $x_0, ..., x_d$. However, if an order is chosen, define the associated oriented simplex by
  \begin{equation*}
    \overline{S}(x_0, ..., x_d) := (S(x_0, .., x_d), [v_1 \wedge ... \wedge v_p]),
  \end{equation*}
  where $v_p = x_p - x_0$, $p = 0,..., d$. For convenience, denote $v_0 = x_0 - x_0 = 0$.
\end{definition}
\begin{claim}
  Let $\sigma$ be a permutation of $\{0, ...., n\}$. and $\sgn(\sigma)$ denotes the parity of $\sigma$. Then
  \begin{equation}
    \OR(\overline{S}(x_{\sigma_0}, ..., x_{\sigma_d})) =
    \sgn(\sigma) \OR(\overline{S}(x_{0}, ..., x_{d})) 
  \end{equation}
\end{claim}
\begin{proof}
  Denote the orientation representatives of the two oriented simplices by $\omega_{\sigma}$ and $\omega$ respectively. Let $p = \sigma^{-1}_0$, i.e., $\sigma_p = 0$. Consider two cases.
  \begin{enumerate}
    \item
      $p = 0$. Denote by $\tilde{\sigma}$ the restriction of $\sigma$ onto $\{1, ..., d\}$. Obviously, $\sgn(\sigma) = \sgn( \tilde{\sigma})$. On the other hand,
      \begin{equation*}
        \omega_{\sigma}
          = (x_{\sigma_1} - x_0) \wedge ... \wedge (x_{\sigma_d} - x_0) 
          = v_{\sigma_1} \wedge ... v_{\sigma_d} = \sgn(\tilde{\sigma}) v_1 \wedge ... \wedge v_d = \sgn(\sigma) \omega.
      \end{equation*}
    \item
      $p \neq 0$. Then
      \begin{equation*}
        \begin{split}
          \omega_{\sigma}
            & = (v_{\sigma_1} - v_{\sigma_0}) \wedge ... \wedge (v_{\sigma_d} - v_{\sigma_0}) \\
            & = (v_{\sigma_1} - v_{\sigma_0}) \wedge ... \wedge (v_{\sigma_{p - 1}} - v_{\sigma_0}) \wedge (-v_{\sigma_0}) \wedge (v_{\sigma_{p + 1}} - v_{\sigma_0}) \wedge ... \wedge (v_{\sigma_d} - v_{\sigma_0}) \\
            & = - v_{\sigma_1} \wedge ... \wedge v_{\sigma_{p - 1}} \wedge v_{\sigma_0} \wedge v_{\sigma_{p + 1}} ... \wedge v_{\sigma_d} \\
            & = - \sgn
              \begin{pmatrix}
                0 & 1 & \cdots & p - 1 & p & p + 1 & \cdots & d \\
                \sigma_p & \sigma_1 &  \cdots & \sigma_{p - 1} & \sigma_0 & \sigma_{p + 1}& \cdots & \sigma_d
              \end{pmatrix}
              \omega \\ 
            & = \sgn(\sigma) \omega
        \end{split}.
      \end{equation*}
  \end{enumerate}
  Hence, $[\omega_{\sigma}] = \sgn(\sigma) [\omega]$, as claimed.
\end{proof}
\begin{claim}
  Let $p \in \{0, ..., d\}$. Then
  \begin{equation}
    \varepsilon
      \left(
        \overline{S}(x_{0}, ..., x_{d}),
        \overline{S}(x_{0}, ...,  \widehat{x_{p}}, ..., x_{d})
      \right)
      = (-1)^p
  \end{equation}
  ($\widehat{x_{p}}$ means $x_p$ omitted).
\end{claim}
\begin{proof}
  Consider two cases.
  \begin{enumerate}
    \item
      $p \neq 0$. An outward-pointing vector to $a_{d - 1} := \overline{S}(x_{0}, ...,  \widehat{x_{p}}, ..., x_{d})$ with respect to $b_d := \overline{S}(x_{0}, ..., x_{d})$ is $-v_p$ Hence,
      \begin{equation*}
        \varepsilon(b_d, a_{d - 1})
          = \frac{[-v_p] \wedge [v_1 \wedge ... \wedge \widehat{x_p} \wedge ... \wedge x_d]}{[v_1 \wedge ... \wedge v_d]}
          = \frac{(-1).(-1)^{p - 1} [v_1 \wedge ... \wedge v_d]}{[v_1 \wedge ... \wedge v_d]}
          = (-1)^p.
      \end{equation*}
    \item
      $p = 0$. Using the previous case ($p = 1$ to be precise), we calculate:
      \begin{equation*}
        \varepsilon
          \left(
            \overline{S}(x_0, ..., x_d),
            \overline{S}(x_1, ..., x_d)
          \right)
          = 
        - \varepsilon
          \left(
            \overline{S}(x_1, x_0, x_2, ..., x_d),
            \overline{S}(x_1, ..., x_d)
          \right)
          = - (-1)^1
          = (-1)^0. \qedhere
      \end{equation*}
  \end{enumerate}
\end{proof}

\subsubsection{Orientation on parallelotopes and grids}
\label{sec:parallelotope}
In this paragraph we give a standard way to orient  any $d$-parallelotope (the term $d$-parallelepiped is also used). The construction generalises to grids (meshes of parallelotopes with parallel direction vectors) in a trivial manner because scaling of direction vectors and changing the centres of cells do not change the orientation.
\begin{definition}
  Let $x_0 \in A$, $e_1, ..., e_d \in V$ be linearly independent vectors. Define the \textbf{parallelotope} with centre $x_0$ and direction vectors $e_1, ..., e_d$ by
  \begin{equation*}
    \Pi(x_0; e_1, ..., e_d) := \Con\left(\big\{x_0 + \sum_{p = 0}^{d} \lambda_p e_p \mid \lambda_p \in \{-1, 1\},\ p = 1, ..., d\big\}\right).
  \end{equation*}
  The definition of parallelotope does not depend on the order of the vectors $e_1, ..., e_d$. However, if an order is chosen, define the associated oriented parallelotope by
  \begin{equation*}
    \overline{\Pi}(x_0; e_1, ..., e_d) = (\Pi(x_0; e_1, ..., e_d), [e_1 \wedge ... \wedge e_d]).
  \end{equation*}
\end{definition}
Let $p \in \{0, ..., d\}$ and $I$ is a subsequence of $(1, ..., d)$ with $p$-elements. Let $J$ be the complement sequence, $\abs{J} = d - p =:q$. Any $p$-face of $\Pi(x_0; e_1, ..., e_d)$ is given by a choice of $\lambda_{J_1}, ..., \lambda_{J_q} \in \{-1, 1\}^q$ as the region
\begin{equation}
  \Pi\left(x_0 + \sum_{r \in J} \lambda_r e_r; e_{I_1}, ..., e_{I_p}\right)
\end{equation}
(the number of $p$-faces is $\binom{d}{p} 2^{d - p}$).
\begin{claim}
  \begin{equation*}
    \varepsilon
      \left(
        \overline{\Pi}(x_0; e_1, ..., e_d),
        \overline{\Pi}\left(x_0 + \lambda_p e_p; e_1, ..., \widehat{e_p}, ...,  e_d\right)
      \right)
      = (-1)^p \lambda_p.
  \end{equation*}
\end{claim}
\begin{proof}
  An outward pointing vector to $a_{d - 1} := \Pi\left(x_0 + \lambda_p e_p; e_1, ..., \widehat{e_p}, ...,  e_d\right)$ with respect to $b_d := \Pi(x_0; e_1, ..., e_d)$ is $- \lambda_p e_p$. Hence,
  \begin{equation*}
    \varepsilon(b_d, a_{d -1})
      = - \frac{[\lambda_p e_p] \wedge [e_1 \wedge ... \wedge \widehat{e_p} \wedge ... \wedge e_d]}{[e_1 \wedge ... \wedge e_d]}
      = \frac{-\lambda_p (-1)^{p - 1} [e_1 \wedge ... \wedge e_d]}{[e_1 \wedge ... \wedge e_d]}
      = (-1)^p \lambda_p. \qedhere
  \end{equation*}
\end{proof}

\subsection{Polytopal (co-)chain complexes and the (co-)boundary operators}
\label{sec:complexes}
\begin{definition}
\label{thm:mesh_chain_complex}
  The set $C_{p} M$ is the free real vector space over the set $M_{p}$. Its elements are called $p$-\textbf{chains} and a general element has the form
  \begin{equation}
    \sum_{c_p \in M_p} \lambda_{c_p} c_{p},
  \end{equation}
  where $\lambda_{c_p} \in \R$ for all $c_p \in M_p$.
  Define the graded vector space $C_{\bullet} M$ as
  \begin{equation}
    C_\bullet M = \bigoplus_{p = 0}^{d} C_{p} M.
  \end{equation}
\end{definition}
\begin{definition}
\label{thm:boundary_operator}
  Let $M$ be an oriented mesh (i.e. all cells are given with orientations). For $p \geq 1$ define the linear maps $\partial_{p} \colon C_p M \to C_{p - 1} M$ which act on basis elements by
  \begin{equation}
    \partial_p c_p = \sum_{b_{p - 1} \preceq c_p} \varepsilon(c_p, b_{p - 1})\ b_{p - 1}.
  \end{equation}
  The collection $\partial$ of all $\partial_p$ is called the \textbf{boundary operator} on $C_\bullet M$ (on $M$).
\end{definition}
\begin{theorem}
  With the above notation, $\partial_{p - 1} \circ \partial_p = 0,\ p = 2, ..., d.$ In other words, we have the chain complex $(C_\bullet M, \partial)$:
  \begin{equation}
    \begin{tikzcd}
      0 \arrow[r, ""] &
      C_d M \arrow[r, "\partial_d"] &
      C_{d - 1} M \arrow[r, "\partial_{d - 1}"] &
      \cdots \arrow[r, "\partial_2"] &
      C_1 M \arrow[r, "\partial_1"] &
      C_0 M \arrow[r, ""] &
      0.
    \end{tikzcd}
  \end{equation}
\end{theorem}
\begin{proof}
  Directly follows from \Cref{thm:polytope_diamond_property} and \Cref{thm:orientation_main_property}.
\end{proof}
\begin{definition}
\label{thm:fundamental_class}
  Let $M$ be a compatibly oriented manifold-like mesh. The \textbf{fundamental class} of $M$ is the $d$-chain
  $[M] = \sum_{c_d \in M_d} c_d$.
\end{definition}
\begin{claim}
\label{thm:fundamental_class_boundary}
  Let $M$ be compatibly oriented manifold-like mesh. Then $\partial [M] = [\partial M]$ (the first $\partial$ is the one defined in \Cref{thm:boundary_operator}, while the second one is the one defined in \Cref{thm:topological_boundary}), where $\partial M$ has the induced orientation on its top-dimensional cells.
  
  In particular, $\partial [M]$ gives rise to a compatible orientation on $\partial M$ (the summands being the oriented cells). 
\end{claim}
\begin{proof}
  If $M$ is empty or $M$ is a $0$-mesh, there is nothing to prove. Otherwise, let $M$ be a $d$-mesh. Let $a_{d - 1} \in M_{d - 1}$. Consider two cases. 
  \begin{enumerate}
    \item
      If $a_{d -1}$ is an interior cell, it has two $d$-superfaces $b_d$ and $c_d$ with $\epsilon(b_d, a_{d - 1}) = - \epsilon(c_d, a_{d - 1})$ and therefore the coefficient before $a_{d -1}$ in $\partial [M]$ is $0$ (and $a_{d - 1} \notin \partial M$).
    \item
      If $a_{d - 1}$ is a boundary cell, then it has unique $d$-superface $c_d$ and since $\partial M$  has the induced orientation, then $\epsilon(c_d, a_{d - 1}) = 1$. Hence the coefficient before $a_{d - 1}$ in $\partial [M]$ is $1$ as in $[\partial M]$.
  \end{enumerate}
  Hence, $\partial [M] = [\partial M]$.
\end{proof}
\begin{definition}
\label{thm:mesh_cochain_complex}
  Let $(M, \OR)$ be an oriented mesh, $(C_{\bullet} M, \partial)$ be the induced chain complex. By $(C^{\bullet} M, \delta)$ denote the dual cochain complex of $(C_{\bullet} M, \partial)$, i.e.,
  \begin{equation}
    \begin{tikzcd}
      0 \arrow[r, ""] &
      C^0 M \arrow[r, "\delta^0"] &
      C^1 M \arrow[r, "\delta^1"] &
      \cdots \arrow[r, "\delta^{d - 2}"] &
      C^{d - 1} M \arrow[r, "\delta^{d - 1}"] &
      C^d M \arrow[r, ""] &
      0
    \end{tikzcd}
  \end{equation}
  ($C^i M = \Hom(C_i M, \R)$ and $\delta^i = \partial_{i - 1}^\star$).
  Elements of $C^\bullet M$ are called \textbf{cochains}, and $\delta$ is the \textbf{coboundary operator}.
\end{definition}
\begin{example}
\label{thm:bd_cbd_example}
  In \Cref{fig:triangulation_0p5_forman_bd}:
  $\partial_2 F_5 = E_8 + E_9 - E_{16} + E_{15}$, $\partial_1 E_{11} = N_5 - N_{10}$.
  
  In \Cref{fig:triangulation_0p5_forman_cbd}: $\delta^0 N^4 = - E^1 + E^6 + E^{10}$, $\delta^1 E^5 = F^3 - F^4$.
\end{example}
\begin{figure}[!ht]
  \begin{subfigure}{.45\textwidth}
    \centering
    \includegraphics[scale=.6]{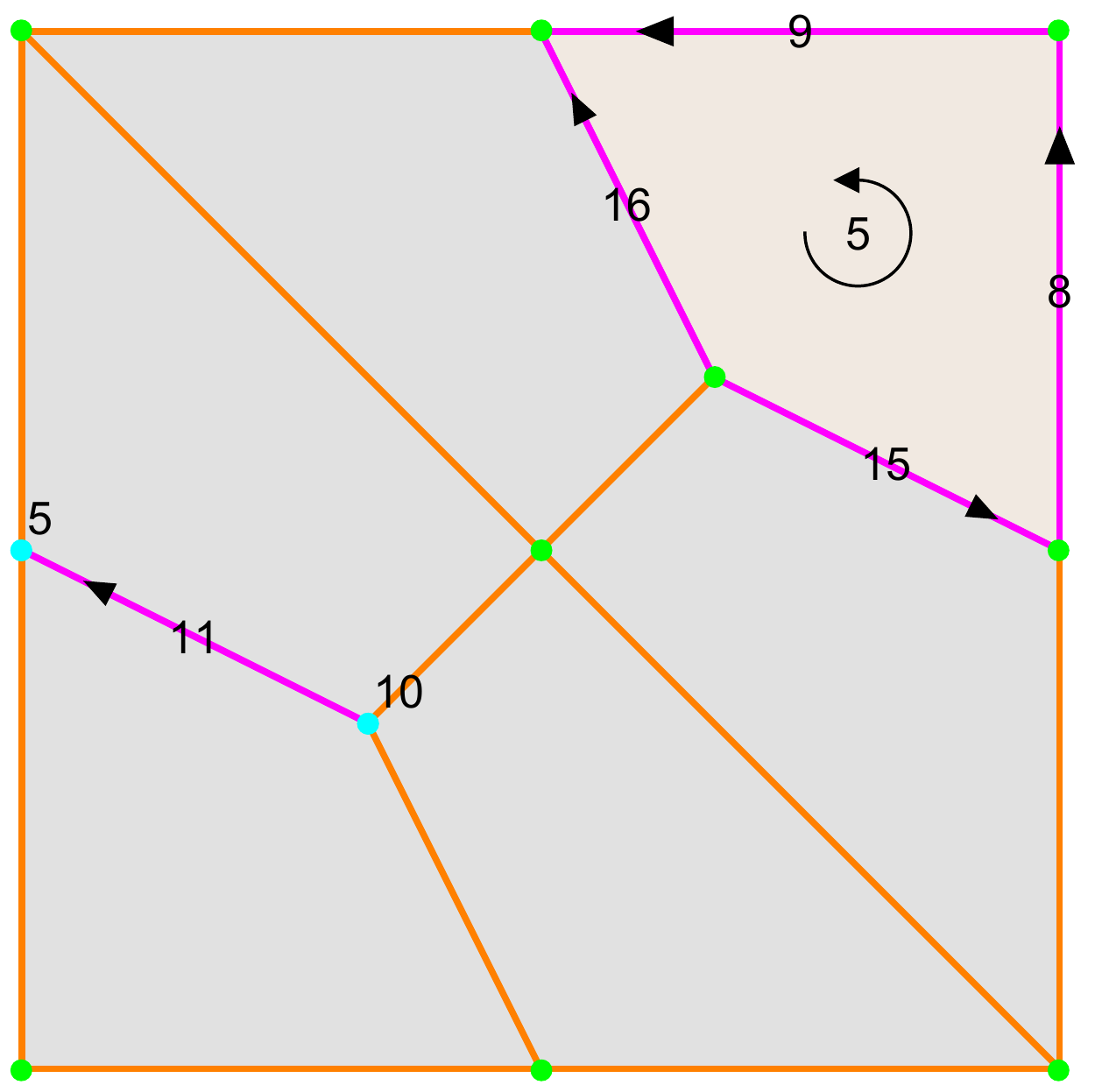}
    \caption{}
    \label{fig:triangulation_0p5_forman_bd}
  \end{subfigure}
  \begin{subfigure}{.45\textwidth}
    \centering
    \includegraphics[scale=.6]{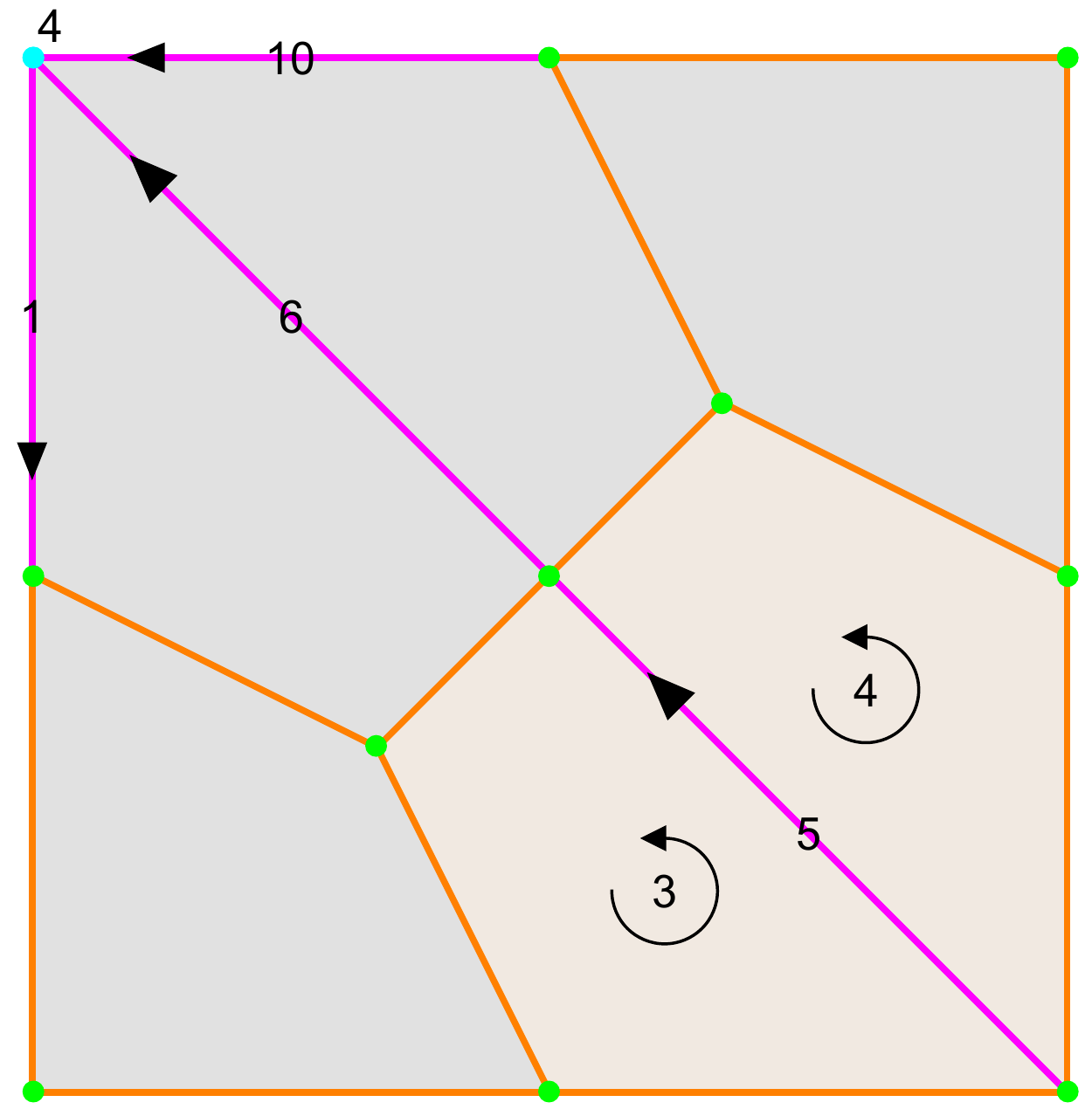}
    \caption{}
    \label{fig:triangulation_0p5_forman_cbd}
  \end{subfigure}
  \caption{Boundary and coboundary operators}
  \label{fig:triangulation_0p5_forman_bd_cbd}
\end{figure}
\begin{definition}
\label{thm:homology}
  The $p$-\textbf{homology} of a chain complex $(C_{\bullet} M, \partial)$ is defined by
  \begin{equation*}
    H_p(C^{\bullet} M, \partial) := (\Ker \partial_p) / (\Im \partial_{p + 1}).
  \end{equation*}
\end{definition}
\begin{definition}
\label{thm:cohomology}
  The $p$-\textbf{cohomology} of a cochain complex $(C^{\bullet} M, \delta)$ is defined by
  \begin{equation*}
    H^p(C^{\bullet} M, \delta) := (\Ker \delta^p) / (\Im \delta^{p - 1}).
  \end{equation*}
\end{definition}

\section{Finite dimensional (discrete) Hodge theory}
\label{sec:hodge}
In this section we show some standard topological results on a mesh based on a choice of a discrete inner product.
Let $(C^{\bullet} K, \delta)$ be a cochain complex:
\begin{equation}
  \begin{tikzcd}
    0 \arrow[r, ""] &
    C^0 K \arrow[r, "\delta^0"] &
    C^1 K \arrow[r, "\delta^1"] &
    \cdots \arrow[r, "\delta^{d - 2}"] &
    C^{d - 1} K \arrow[r, "\delta^{d - 1}"] &
    C^d K \arrow[r, ""] &
    0.
  \end{tikzcd}
\end{equation}
For each $p$ fix an inner product $\inner{\cdot}{\cdot}_p \colon C^p K \times C^p K \to \R$.
\begin{remark}
  For \ref{sec:hodge_general} and \ref{sec:hodge_decomposition} it is not necessary for $(C^{\bullet} K, \delta)$ to be the polytopal cochain complex of some mesh $K$. However, we keep the reference to $K$ in order to emphasize that these results are especially important when considering the topological properties of a quasi-cubical mesh $K$, which is important in \ref{sec:hodge_poincare}.
\end{remark}

\subsection{Adjoint coboundary operator and Laplacian}
\label{sec:hodge_general}
\begin{definition}
  Define the \textbf{adjoint coboundary operator} $\delta^\star_p \colon C^p K \to C^{p - 1} K$ to be the adjoint of $\delta^{p - 1}$ with respect to $\inner{\cdot}{\cdot}$, i.e.,
  \begin{equation}
    \inner{\delta^{p - 1} \sigma^{p - 1}}{\tau^p}_p
      = \inner{\sigma^{p - 1}}{\delta^\star_p \tau^p}_{p - 1}
  \end{equation}
  (in other words $\delta^\star_{p + 1} = (\delta^p)^\star$).
\end{definition}
\begin{claim}
  $\delta^\star_p \circ \delta^\star_{p + 1} = 0$, i.e. we have the chain complex
  $(C^{\bullet} K, \delta^\star)$:
  \begin{equation}
    \begin{tikzcd}
      0 \arrow[r, ""] &
      C^d K \arrow[r, "\delta^\star_d"] &
      C^{d - 1} K \arrow[r, "\delta^\star_{d - 1}"] &
      \cdots \arrow[r, "\delta^\star_2"] &
      C^1 K \arrow[r, "\delta^\star_1"] &
      C^0 K \arrow[r, ""] &
      0.
    \end{tikzcd}
  \end{equation}
\end{claim}
\begin{proof}
  Let $\sigma^{p - 1} \in C^{p - 1},\ \tau^p \in C^p K$ be arbitrary. Then
  \begin{equation*}
    \inner{\sigma^{p - 1}}{(\delta^\star_p \circ \delta^\star_{p + 1}) \tau^{p + 1}}
      = \inner{\delta^{p - 1} \sigma^{p - 1}}{\delta^\star_{p + 1} \tau^{p + 1}}
      = \inner{(\delta^p \circ \delta^{p - 1}) \sigma^{p - 1}}{\tau^{p + 1}}
      = 0
  \end{equation*}
  Since $\sigma^{p - 1}$ and $\tau^p$ are arbitrary, $\delta^p \circ \delta^{p - 1} = 0$.
\end{proof}
\begin{definition}
  The \textbf{Laplacian} $\Delta_p \colon C^p K \to C^p K$ is defined by
  \begin{equation}
    \Delta_p = \delta^{p - 1} \circ \delta^\star_p + 
               \delta^\star_{p + 1} \circ \delta^p.
  \end{equation}
\end{definition}
\begin{claim}
  $\Delta$ is a morphism of both $(C^{\bullet} K, \delta)$ and $(C^{\bullet} K, \delta^\star)$, i.e.,
  \begin{enumerate}
    \item 
      $\Delta_{p + 1} \circ \delta^p = \delta^p \circ \Delta_p$,
    \item 
      $\Delta_{p - 1} \circ \delta^\star_p = \delta^\star_p \circ \Delta_p$.
  \end{enumerate}
\end{claim}
\begin{proof}
  Because we are not going to have sign issues, we omit dimension indices.
  \begin{enumerate}
    \item 
      $\Delta \circ \delta 
        = (\delta \circ \delta^\star + \delta^\star \circ \delta) \circ \delta 
        = \delta \circ \delta^\star \circ \delta 
        = \delta \circ (\delta^\star \circ \delta + \delta \circ \delta^\star) 
        = \delta \circ \Delta$.
    \item 
      $\Delta \circ \delta^\star 
        = (\delta \circ \delta^\star + \delta^\star \circ \delta) \circ \delta^\star 
        = \delta^\star \circ \delta \circ \delta^\star 
        = \delta^\star \circ (\delta \circ \delta^\star + \delta^\star \circ \delta) 
       = \delta^\star \circ \Delta$. \qedhere
  \end{enumerate}
\end{proof}
\begin{claim}
  $\Delta_p$ is symmetric, positive semi-definite
  (with respect to $\inner{\cdot}{\cdot}$) and 
  \begin{equation}
    \Ker \Delta_p = \Ker \delta^p \cap \Ker \delta^\star_p.
  \end{equation}
\end{claim}
\begin{proof}
  Let $\sigma^p, \tau^p \in C^p K$. Then
  \begin{equation}
    \inner{\Delta_p \sigma^p}{\tau^p} 
      =   \inner{(\delta^{p - 1} \circ \delta^\star_p
        + \delta^\star_{p + 1} \circ \delta^p) \sigma^p}{\tau^p} 
      =   \inner{\delta^\star_p \sigma^p}{\delta^\star_p \tau^p} 
        + \inner{\delta^p \sigma^p}{\delta^p \tau^p}.
  \end{equation}
  Analogous computation leads to the same result for $\inner{\sigma^p}{\Delta_p \tau^p}$ which means that $\Delta_p$ is symmetric. Since $\inner{\cdot}{\cdot}$ is positive definite, then 
  \begin{equation}
    \inner{\Delta_p \sigma^p}{\sigma^p}
      =   \inner{\delta^\star_p \sigma^p}{\delta^\star_p \sigma^p}
        + \inner{\delta^p \sigma^p}{\delta^p \sigma^p}
      \geq 0 + 0 
      = 0
  \end{equation}
  and the equality $\Ker \Delta_p = \Ker \delta^p \cap \Ker \delta^\star_p$ follows directly.
\end{proof}
\begin{remark}
  The elements of $\Ker \Delta_p$ are called \textbf{harmonic cochains} (or \textbf{harmonic forms} in a context where we work on forms instead of cochains).
\end{remark}
\subsection{Hodge decomposition for a complex of finite dimensional spaces}
\label{sec:hodge_decomposition}
Let all the vector spaces of $C^p K$ be finite-dimensional.
\begin{lemma}
\label{thm:closed_equals_harmonic_plus_exact}
  $(\Im \delta^{p - 1})^\perp = \Ker \delta^\star_p$ and
  $(\Im \delta^\star_{p + 1})^\perp = \Ker \delta^p$.
\end{lemma}
\begin{proof}
  We prove only the first equality as the second one is proven in the same way. Let $\sigma^p \in \Ker \delta^\star_p$, $\tau^{p - 1} \in C^{p - 1} K$. Then
  \begin{equation}
    0 = \inner{\delta^\star_p \sigma^p}{\tau^{p - 1}} 
      = \inner{\sigma^p}{\delta^{p - 1} \tau^{p - 1}}
  \end{equation}
  and therefore $\sigma^p \perp \Im \delta^{p - 1}$. Hence, 
  $\Ker \delta^\star_p \subseteq (\Im \delta^{p - 1})^\perp$.
	
  Inversely, let $\sigma^p \in (\Im \delta^{p - 1})^\perp$. 
  Then for any $\tau^{p - 1} \in C^{p - 1} K$,
  \begin{equation}
    0 = \inner{\delta^{p - 1} \tau^{p - 1}}{\sigma}
      = \inner{\tau^{p - 1}}{\delta^\star_p \sigma^p}
  \end{equation}
  and because $\tau^{p - 1}$ is arbitrary, $\delta^\star_p \sigma^p = 0$. Hence, $(\Im \delta^{p - 1})^\perp \subseteq \Ker \delta^\star_p$.
\end{proof}
\begin{theorem}
\label{thm:hodge_lemma}
  The following orthogonal decompositions hold:
  \begin{enumerate}
    \item
      $\Ker \delta^p = \Ker \Delta_p \oplus \Im \delta^{p - 1}$;
    \item
      $\Ker \delta^\star_p = \Ker \Delta_p \oplus \Im \delta^\star_{p + 1}$.
  \end{enumerate}
\end{theorem}
\begin{proof}
  We prove only the first equality as the second one is proven in the same way.
  \begin{equation}
    \begin{split}
      (\Ker \Delta_p \oplus \Im \delta^{p - 1})^\perp 
        & = ((\Ker \delta^p \cap \Ker \delta^\star_p)
    	    \oplus \Im \delta^{p - 1})^\perp \\
        & = (\Ker \delta^p \cap \Ker \delta^\star_p)^\perp
    	    \cap (\Im \delta^{p - 1})^\perp \\
        & = ((\Ker \delta^p)^\perp \oplus (\Ker \delta^\star_p)^\perp)
            \cap \Ker \delta^\star_p \\
        & = (\Im \delta^\star_{p + 1} \oplus (\Ker \delta^\star_p)^\perp) 
    	    \cap \Ker \delta^\star_p  \\
        & = (\Im \delta^\star_{p + 1} \cap \Ker \delta^\star_p) 
            \oplus ((\Ker \delta^\star_p)^\perp \cap \Ker \delta^\star_p) \\
        & = \Im \delta^\star_{p + 1} \oplus 0 \\
        & = (\Ker \delta^p)^\perp
    \end{split}.
  \end{equation}
  After dropping $\perp$ we get the desired equation.
\end{proof}
\begin{theorem}[Hodge decomposition]
\label{thm:hodge}
  With respect to $\inner{\cdot}{\cdot}$ the following orthogonal decomposition holds:
  \begin{equation}
    C^p K = 
      \Im \delta^{p - 1}
      \oplus \Ker \Delta_p
      \oplus \Im \delta^\star_{p + 1}.
  \end{equation}
\end{theorem}
\begin{proof}
  This is equivalent to $(\Ker \Delta_p \oplus \Im \delta^{p - 1})^\perp = \Im \delta^\star_{p + 1}$ which was shown in the proof of \Cref{thm:hodge_lemma}.
\end{proof}
\begin{corollary}
\label{thm:cohomology_equals_harmonic}
  \begin{equation}
    H^p(C^{\bullet} K, \delta) 
      \cong \Ker \Delta_p 
      \cong H_p(C^{\bullet} K, \delta^\star).
  \end{equation}
\end{corollary}
\begin{proof}
  The left and the right equality are just special case of the fact that $A = B \oplus C \Rightarrow B \cong A / C$ applied to the first and second part respectively of \Cref{thm:hodge_lemma}.
\end{proof}
\begin{remark}
\label{thm:cohomology_equals_homology}
  If $K$ is a finite mesh and $\inner{\cdot}{\cdot}_p$ is defined to be the dot product with respect to the standard basis of cochains, then $\delta^\star_p$ has the same matrix representation as $\partial_p$ in the standard basis of chains. Hence, \Cref{thm:cohomology_equals_harmonic} implies
  \begin{equation}
    H^p(C^{\bullet} K, \delta)
      \cong H_p(C^{\bullet} K, \delta^\star)
      \cong H_p(C_{\bullet} K, \partial).
  \end{equation}
  In other words, cohomology agrees with homology.
\end{remark}

\subsection{Hodge star and Poincar\'e duality on meshes}
\label{sec:hodge_poincare}
Let $K$ be a compatibly oriented $d$-mesh and $(C^\bullet K, \delta)$ be the induced polytopal cochain complex. Let $\smile$ be a cup product on $K$ (and hence $(C^\bullet K, \smile, \delta)$ is a differential graded algebra).
\begin{definition}
  The \textbf{Hodge star operator} $\star_p \colon C^p K \to C^{d - p} K$ is defined such that for any  $\sigma^{n - p} \in C^{n - p} K ,\ \tau^p \in C^p K,$
  \begin{equation}
    \inner{\sigma^{n - p}}{\star_p \tau^p} = (\sigma^{n - p} \smile \tau^p)[K].
  \end{equation}	 
\end{definition}
\begin{claim}
  $\inner{\star_p \sigma^p}{\tau^{d - p}}
    = (-1)^{p (d - p)} \inner{\sigma^p}{\star_{d - p} \tau^{d - p}}$, 
  $\sigma^p \in C^p K,\ \tau^{d - p} \in C^{d - p} K$.
\end{claim}
\begin{proof}
  $\inner{\star_p \sigma^p}{\tau^{d - p}} 
    = \inner{\tau^{d - p}}{\star_p \sigma^p} 
    = (\tau^{d - p} \smile \sigma^p)[K] \\
    = (-1)^{p (d - p)} (\sigma^p \smile \tau^{d - p})[K] 
    = (-1)^{p (d - p)}\inner{\sigma^p}{\star_{d - p} \tau^{d - p}}$.
\end{proof}
\begin{claim}
  Let $K$ be a closed mesh. Then
  $\star_{p + 1} \circ \delta^p 
    = (-1)^{d - p} \delta^\star_{d - p} \circ \star_p$.
\end{claim}
\begin{proof}
  Let $\sigma^{d - p - 1} \in C^{d - p - 1} K,\ \tau^p \in C^p K$. Then
  \begin{equation*}
    \begin{split}
      \inner{\sigma^{d - p - 1}}{\delta^\star_{d - p} (\star_p \tau^p)} 
        & = \inner{\delta^{d - p - 1} \sigma^{d - p - 1}}{\star_p \tau^p} \\
        & = ((\delta^{d - p - 1} \sigma^{d - p - 1}) \smile\tau^p)[K] \\
        & = (\delta^{n - 1}(\sigma^{d - p - 1} \smile \tau^p) 
          - (-1)^{d - p - 1} \sigma^{d - p - 1} \smile (\delta^p \tau^p))[K] \\
        & = (\sigma^{d - p - 1} \smile \tau^p)(\partial [K])
          + (-1)^{d - p} (\sigma^{d - p - 1} \smile (\delta^p \tau^p))[K] \\
        & = (-1)^{d - p} \inner{\sigma^{d - p - 1}}{\star_{p + 1} (\delta^p \tau^p)}
    \end{split}.
  \end{equation*}
  Since $\sigma^{d - p - 1}$ and $\tau^p$ are arbitrary, we get the desired equation.
\end{proof}
The following two theorems both imply Poincar\'e duality (\Cref{thm:poincare_duality}).
\begin{theorem}
  Let $K$ be a closed mesh and $\smile$ be non-degenerate on cochain level. Then $\star$ is well defined on cohomology and induces an isomorphism
  \begin{equation}
    \overline{\star} \colon H^p(C^\bullet K, \delta) \to H_{d - p}(C^\bullet K, \delta^\star).
  \end{equation}
\end{theorem}
\begin{proof}
  See \cite[Theorem 3.3.6]{arnold2012discrete}.
\end{proof}
\begin{theorem}
  Let $K$ be a closed mesh and be $\smile$ is non-degenerate on cochain level. Then $\star$ restricts to an isomorphism
  \begin{equation}
    \left. \star \right|_{\Ker \Delta_p} \colon \Ker \Delta_p \to \Ker \Delta_{d - p}.
  \end{equation}
\end{theorem}
\begin{proof}
  See \cite[Lemma 6.2 (3)]{wilson2007cochain}.
\end{proof}
\begin{corollary}[Poincar\'e duality]
\label{thm:poincare_duality}
  Let $K$ be a closed mesh and $\smile$ be non-degenerate on cohomology level. Then
  \begin{equation}
    H^p(C^{\bullet} K, \delta) \cong H_{d - p}(C_{\bullet} K, \partial).
  \end{equation}
\end{corollary}
\begin{proof}
  $H^p(C^{\bullet} K, \delta)
    \cong \Ker \Delta_p
    \cong \Ker \Delta_{d - p}
    \cong H^{d - p}(C^{\bullet} K, \delta)
    \cong H_{d - p}(C_{\bullet} K, \partial)$.
\end{proof}
\begin{remark}
  Meshes with non-degenerate cup product on cohomology level include:
  \begin{itemize}
    \item
      simplicial (via simplicial Whitney forms), see \cite{wilson2007cochain};
    \item
      quasi-cubical (for cubes with cubical Whitney forms),
      see \cite{arnold2012discrete};
    \item
      mixed in 2D containing both triangles and quadrilaterals (using an appropriate formula depending on the cell type).
  \end{itemize}
\end{remark}


\begin{thebibliography}{10}

\bibitem{Boisse2021_Composites}
Boisse P.
\newblock Composite Reinforcements for Optimum Performance.
\newblock Elsevier, London; 2021.
\newblock Available from:
  \url{https://www.sciencedirect.com/book/9780128190050/composite-reinforcements-for-optimum-performance}.

\bibitem{Sankaran2017_Alloys}
Sankaran KK, Mishra RS.
\newblock Metallurgy and Design of Alloys with Hierarchical Microstructures.
\newblock Elsevier, London; 2017.
\newblock Available from:
  \url{https://www.sciencedirect.com/book/9780128120682/metallurgy-and-design-of-alloys-with-hierarchical-microstructures}.

\bibitem{DebRoy2018_AMmetals}
DebRoy T, Wei H, Zuback J, Mukherjee T, Elmer J, Milewski J, et~al.
\newblock Additive manufacturing of metallic components – Process, structure
  and properties.
\newblock Progress in Materials Science. 2018;92:112--224.
\newblock Available from:
  \url{https://www.sciencedirect.com/science/article/abs/pii/S0079642517301172}.

\bibitem{Ngo2018_AMcomposites}
Ngo T, Kashani A, Imbalzano G, Nguyen K, D H.
\newblock Additive manufacturing (3D printing): A review of materials, methods,
  applications and challenges.
\newblock Composites Part B: Engineering. 2018;143:172--196.
\newblock Available from:
  \url{https://www.sciencedirect.com/science/article/abs/pii/S1359836817342944}.

\bibitem{West2017_Fractional}
West BJ.
\newblock Nature’s Patterns and the Fractional Calculus.
\newblock de Gruyte, Berlin; 2017.
\newblock Available from:
  \url{https://www.degruyter.com/document/doi/10.1515/9783110535136/html}.

\bibitem{Sun2018_Fractional}
Sun HG, Zhang Y, Baleanu D, Chen W, Chen YQ.
\newblock A new collection of real world applications of fractional calculus in
  science and engineering.
\newblock Communications in Nonlinear Science and Numerical Simulation.
  2018;64:213--231.
\newblock Available from:
  \url{https://www.sciencedirect.com/science/article/abs/pii/S1007570418301308}.

\bibitem{Sherief2010_Fractional}
Sherief H, El-Sayed A, Abd El-Latief A.
\newblock Fractional order theory of thermoelasticity.
\newblock International Journal of Solids and Structures. 2010;47:269–275.
\newblock Available from:
  \url{https://www.sciencedirect.com/science/article/pii/S002076830900376X}.

\bibitem{DiPaola2011_Fractional}
Di~Paola M, Pirrotta A, Valenza A.
\newblock Visco-elastic behavior through fractional calculus: An easier method
  for best fitting experimental results.
\newblock Mechanics of Materials. 2011;43:799–806.
\newblock Available from:
  \url{https://www.sciencedirect.com/science/article/abs/pii/S0167663611001657}.

\bibitem{Bologna2015_Fractional}
Bologna M, Svenkeson A, West B, Grigolini P.
\newblock Diffusion in heterogeneous media: An iterative scheme for finding
  approximate solutions to fractional differential equations with
  time-dependent coefficients.
\newblock Journal of Computational Physics. 2015;293:297–311.
\newblock Available from:
  \url{https://www.sciencedirect.com/science/article/pii/S0021999114005816}.

\bibitem{Hirani2003_DEC}
Hirani A.
\newblock Discrete Exterior Calculus [Ph.D. thesis].
\newblock California Institute of Technology; 2003.
\newblock Available from: \url{https://thesis.library.caltech.edu/1885/}.

\bibitem{Arnold2018_FEEC}
Arnold D.
\newblock Finite Element Exterior Calculus. vol.~93 of CBMS-NSF Regional
  Conference Series in Applied Mathematics.
\newblock Society for Industrial and Applied Mathematics (SIAM), Philadelphia,
  PA; 2018.
\newblock Available from:
  \url{https://epubs.siam.org/doi/book/10.1137/1.9781611975543}.

\bibitem{Hirani2015_DarcyDEC}
Hirani A, Nakshatrala K, Chaudhry J.
\newblock Numerical method for Darcy flow derived using discrete exterior
  calculus.
\newblock International Journal for Computational Methods in Engineering
  Science and Mechanics. 2015;16(3):151--169.
\newblock Available from:
  \url{https://www.tandfonline.com/doi/full/10.1080/15502287.2014.977500}.

\bibitem{Mohamed2016_NavierDEC}
Mohamed M, Hirani A, Samtaney R.
\newblock Discrete exterior calculus discretization of incompressible
  Navier–Stokes equations over surface simplicial meshes.
\newblock Journal of Computational Physics. 2016;312:175--191.
\newblock Available from:
  \url{https://www.sciencedirect.com/science/article/pii/S0021999116000929}.

\bibitem{Schulz2020_DEC}
Schulz E, G T.
\newblock Convergence of discrete exterior calculus approximations for Poisson
  problems.
\newblock Discrete \& Computational Geometry. 2020;63:364--376.
\newblock Available from:
  \url{https://link.springer.com/article/10.1007/s00454-019-00159-x}.

\bibitem{Boom2021_Poisson}
Boom P, Seepujak A, Kosmas O, Margetts L, Jivkov A.
\newblock Parallelized discrete exterior calculus for three-dimensional
  elliptic problems.
\newblock Computer Physics Communications. 2021;under review.
\newblock Available from: \url{http://arxiv.org/abs/2104.05999}.

\bibitem{Boom2022_Elasticity}
Boom P, Kosmas O, Margetts L, Jivkov A.
\newblock A geometric formulation of linear elasticity based on Discrete
  Exterior Calculus.
\newblock International Journal of Solids and Structures. 2022;236-237:111345.
\newblock Available from:
  \url{https://www.sciencedirect.com/science/article/abs/pii/S0020768321004212}.

\bibitem{Tonti2013_Structure}
Tonti E.
\newblock The Mathematical Structure of Classical and Relativistic Physics: A
  General Classification Diagram.
\newblock Springer; 2013.
\newblock Available from:
  \url{https://link.springer.com/book/10.1007%2F978-1-4614-7422-7}.

\bibitem{Forman1998_Morse}
Forman R.
\newblock Morse Theory for Cell Complexes.
\newblock Advances in Mathematics. 1998;134:90--145.
\newblock Available from:
  \url{https://www.sciencedirect.com/science/article/pii/S0001870897916509}.

\bibitem{Forman2002_Novikov}
Forman R.
\newblock Combinatorial Novikov-Morse theory.
\newblock International Journal of Mathematics. 2002;13(4):333--368.
\newblock Available from:
  \url{https://www.worldscientific.com/doi/abs/10.1142/S0129167X02001265}.

\bibitem{Harker2014_Morse}
Harker S, Mischaikow K, Mrozek M, Nanda V.
\newblock Discrete Morse Theoretic Algorithms for Computing Homology of
  Complexes and Maps.
\newblock Foundations of Computational Mathematics. 2014;14(1):151--184.
\newblock Available from:
  \url{https://link.springer.com/article/10.1007%2Fs10208-013-9145-0}.

\bibitem{Mrozek2021_Flows}
Mrozek M, Wanner T.
\newblock Creating semiflows on simplicial complexes from combinatorial vector
  fields.
\newblock Journal of Differential Equations. 2021;304:375--434.
\newblock Available from:
  \url{https://www.sciencedirect.com/science/article/abs/pii/S0022039621006069}.

\bibitem{Forman2003_Bochner}
Forman R.
\newblock Bochner's method for cell complexes and combinatorial Ricci
  curvature.
\newblock Discrete and Computational Geometry. 2003;29(3):323--374.
\newblock Available from:
  \url{https://link.springer.com/article/10.1007/s00454-002-0743-x}.

\bibitem{lange2005combinatorial}
Lange CEMC.
\newblock Combinatorial curvatures, group actions, and colourings: Aspects of
  topological combinatorics [Ph.D. thesis].
\newblock Technical University of Berlin; 2005.
\newblock Available from:
  \url{https://depositonce.tu-berlin.de/handle/11303/1149}.

\bibitem{Watanabe2019_Ricci}
Watanabe K.
\newblock Combinatorial Ricci curvature on cell-complex and Gauss-Bonnet
  theorem.
\newblock Tohoku Mathematical Journal. 2019;71:533--547.
\newblock Available from:
  \url{https://projecteuclid.org/journals/tohoku-mathematical-journal/volume-71/issue-4/Combinatorial-Ricci-curvature-on-cell-complex-and-Gauss-Bonnnet-theorem/10.2748/tmj/1576724792.short}.

\bibitem{weber2016forman}
Weber M, Jost J, Saucan E.
\newblock Forman-Ricci flow for change detection in large dynamic data sets.
\newblock Axioms. 2016;5(4):26.
\newblock Available from: \url{https://www.mdpi.com/2075-1680/5/4/26}.

\bibitem{arnold2012discrete}
Arnold RF.
\newblock The Discrete Hodge Star Operator and Poincar\`e Duality [Ph.D.
  thesis].
\newblock Virginia Tech; 2012.
\newblock Available from:
  \url{https://vtechworks.lib.vt.edu/handle/10919/27485}.

\bibitem{wilson2007cochain}
Wilson SO.
\newblock Cochain algebra on manifolds and convergence under refinement.
\newblock Topology and its Applications. 2007;154(9):1898--1920.
\newblock Available from:
  \url{https://www.sciencedirect.com/science/article/pii/S0166864107000314}.

\bibitem{whitney1957geometric}
Whitney H.
\newblock Geometric integration theory.
\newblock Princeton University Press; 1957.
\newblock Available from:
  \url{https://press.princeton.edu/books/hardcover/9780691652900/geometric-integration-theory}.

\bibitem{borodin2021nanoceramics}
Borodin A, Jivkov A, Sheinerman A, MYu G.
\newblock Optimisation of rGO-enriched nanoceramics by combinatorial analysis.
\newblock Materials \& Design. 2021;212:110191.
\newblock Available from:
  \url{https://www.sciencedirect.com/science/article/pii/S0264127521007462}.

\bibitem{nistal2018nanoplates}
Nistal A, Garcia E, Perez-Coll D, Prieto C, Belmonte M, Osendi M, et~al.
\newblock Low percolation threshold in highly conducting graphene
  nanoplatelets/glass composite coatings.
\newblock Carbon. 2018;139:556–563.
\newblock Available from:
  \url{https://www.sciencedirect.com/science/article/abs/pii/S0008622318306729}.

\bibitem{micaela2016nanotubes}
Micaela C, Massimo R, Imran S, T A.
\newblock Conductivity in carbon nanotube polymer composites: A comparison
  between model and experiment.
\newblock Composites: Part A. 2016;87:237--242.
\newblock Available from:
  \url{https://www.sciencedirect.com/science/article/abs/pii/S1359835X16301208}.

\bibitem{ziegler1995lectures}
Ziegler GM.
\newblock Lectures on polytopes. vol. 152.
\newblock Springer; 1995.
\newblock Available from:
  \url{https://link.springer.com/book/10.1007%2F978-1-4613-8431-1}.

\bibitem{chen2014digital}
Chen LM.
\newblock Digital and Discrete Geometry: Theory and Algorithms.
\newblock Springer; 2014.
\newblock Available from:
  \url{https://link.springer.com/book/10.1007%2F978-3-319-12099-7}.

\bibitem{hatcher2002topology}
Hatcher A.
\newblock Algebraic topology.
\newblock Cambridge University Press; 2002.
\newblock Available from:
  \url{https://www.cambridge.org/core/journals/proceedings-of-the-edinburgh-mathematical-society/article/hatcher-aalgebraic-topology-cambridge-university-press-2002-556-pp-0-521-79540-0-softback-2095-0-521-79160-x-hardback-60/9F3482890F390477D49007403D5E701F}.

\end{thebibliography}

\end{document}